\documentclass[11pt]{article} % use larger type; default would be 10pt
\usepackage[utf8]{inputenc} % set input encoding (not needed with XeLaTeX)

\usepackage{geometry} % to change the page dimensions
\usepackage{graphicx, amsmath,amsfonts,cite,appendix, amssymb,xcolor,subfig, hyperref, enumerate, mathrsfs, amsthm,url, mathtools} 
\usepackage{booktabs} % for much better looking tables
\usepackage{array} % for better arrays (eg matrices) in maths
\usepackage{paralist} % very flexible & customisable lists (eg. enumerate/itemize, etc.)
\usepackage{verbatim} % adds environment for commenting out blocks of text & for better verbatim
\usepackage{subfig} % make it possible to include more than one captioned figure/table in a single float
\usepackage{bbm}
\usepackage{pgfgantt}
\usepackage{caption}
\usepackage{algpseudocode}
\usepackage{sidecap}

\usepackage[font=small,labelfont=bf,labelsep=space]{caption}
\newcommand{\un}{\underline}
\newcommand{\be}{\begin{equation}}
\newcommand{\ee}{\end{equation}}
\newcommand{\ben}{\begin{equation*}}
\newcommand{\een}{\end{equation*}}
\newcommand{\mc}{\mathcal}

\newtheorem{lem}{Lemma}
\newtheorem{applem}{Lemma}[section]
\newtheorem{defi}{Definition}[section]
\newtheorem{thm}{Theorem}
\newtheorem{fact}{Fact}

\newcommand{\e}{\epsilon}

\newcommand{\abs}[1]{\left \lvert #1\right \rvert}
\newcommand{\norm}[1]{\left \lVert #1\right \rVert}
\newcommand{\expec}{\mathbb{E}}

\newcommand{\mscrs}{\mathscr{S}}

\usepackage{sectsty}
\allsectionsfont{\sffamily\mdseries\upshape} % (See the fntguide.pdf for font help)
\title{\vspace{-0.6in} Analysis of Approximate Message Passing\\
with A Class of Non-Separable Denoisers}
\author{{Yanting Ma}\\ {North Carolina State University}
  \\ {\small {yma7@ncsu.edu}}
\and {Cynthia Rush}\\ {Columbia University}
  \\ {\small {cynthia.rush@columbia.edu}}
 \and {Dror Baron} \\ {North Carolina State University}
  \\ {\small{barondror@ncsu.edu}}
}

\begin{document}
\maketitle

\vspace{-10pt}
\begin{abstract}
Approximate message passing (AMP) is a class of efficient algorithms for solving 
high-dimensional linear regression tasks 
where one wishes to recover an unknown signal $\beta_0$ from noisy, linear measurements $y = A\beta_0 + w$.  When applying a separable denoiser at each iteration of the algorithm, the performance of AMP (for example, the mean squared error of its estimate) can be accurately tracked by a simple, scalar recursion called state evolution. Separable denoisers are sufficient when the unknown signal has independent %and identically distributed 
entries, however, in many real-world applications, like image or audio signal reconstruction, the signal contains dependencies between entries.  In these cases, a coordinate-wise independence structure is not a good approximation to the true prior of the unknown signal.  In this paper we assume the unknown signal has dependent entries, and using a class of non-separable sliding-window denoisers, we prove that a new form of state evolution still accurately predicts AMP performance.
This is an early step in understanding the role of non-separable denoisers within AMP, and will lead to a characterization of more general denoisers in problems including compressive image reconstruction.
%This paper studies the performance of approximate message passing (AMP) with a class of non-separable denoisers for estimating a high dimensional vector $\beta_0$ from observation $y=A\beta_0+w$, where $\beta_0$ is produced by a stationary Markov Chain and the measurement matrix $A$ has i.i.d. Gaussian entries. 
%Apart from being low-complexity and scalable, AMP has the attractive feature that its average performance concentrates to predicted values calculated from a simple scalar iteration, called state evolution, with exponential rate in the problem dimension.
%Previous analyses of the performance of AMP have been limited to separable denoisers or block-separable denoisers, thereby assuming that the vector $\beta_0$ to be recovered is i.i.d..  
%In this paper, we define a state evolution for AMP with a class of non-separable sliding-window denoisers and prove that the average performance of AMP with this class of denoisers concentrates to this state evolution prediction, also with exponential rate in the problem dimension.
\end{abstract}

%\input{sec_intro}
% !TEX root = main.tex

\section{Introduction} \label{sec:intro}

%\subsection{Motivation}
In this work, we study the high-dimensional linear regression model, where one wishes to recover an unknown signal $\beta_0 \in \mathbb{R}^N$ from noisy observations as in the following model:
\be
y = A \beta_0 + w,
\label{eq:model1}
\ee
where $y \in \mathbb{R}^n$ is the output, $A \in \mathbb{R}^{n \times N}$ is a known measurement matrix, and $w \in \mathbb{R}^n$ is zero-mean noise with finite variance $\sigma^2$.  We assume that the ratio of the dimensions of the measurement matrix is a constant value, $\delta := n/N$, with $\delta \in (0, \infty)$.
%Note the problem is generally considered `high-dimensional' when $\delta < 1$ meaning the measurement matrix has more columns than rows.

Approximate message passing (AMP) \cite{DonMalMont09, MontChap11, BayMont11, krz12, Rangan11} is a class of low-complexity, scalable algorithms studied to solve the high-dimensional regression task of \eqref{eq:model1}.  The performance of AMP depends on a sequence of functions $\{\eta_t\}_{t\geq 0}$ used to generate a sequence of estimates $\{\beta^t\}_{t\geq 0}$ from auxiliary observation vectors computed in every iteration of the algorithm. A nice property of AMP is that under some technical conditions these observation vectors can be approximated as the input signal $\beta_0$ plus independent and identically distributed, or i.i.d., Gaussian noise. This fact allows one to choose functions $\{\eta_t\}_{t\geq 0}$ based on statistical knowledge of $\beta_0$, for example, a common choice is for $\eta_t$ to be the Bayes-optimal estimator of $\beta_0$ conditional on the value of the observation vector.  For this reason, the functions $\{\eta_t\}_{t\geq 0}$ are referred to as 
`denoisers.'

%Previous analysis of the performance of AMP only considers the case %of coordinate-wise denoisers
%when the unknown signal has a prior with i.i.d.\ entries.  In this case, there is no loss of performance by limiting choices of the denoiser functions $\{\eta_t\}_{t\geq 0}$ to the class of functions that acts coordinate-wise when applied to a vector; such functions are referred to as \emph{separable}.
Previous analysis of the performance of AMP only considers denoisers $\{\eta_t\}_{t\geq 0}$ that act coordinate-wise when applied to a vector; such functions are referred to as \emph{separable}. If the unknown signal $\beta_0$ has a prior distribution assuming
i.i.d.\ entries, restricting consideration to only separable denoisers causes no loss in performance.
However, in many real-world applications, the unknown signal $\beta_0$ contains dependencies between entries and therefore a coordinate-wise independence structure is not a good approximation for the prior of $\beta_0$.  For example, when the signals are images \cite{Tan15,Metzler16} or sound clips \cite{Ma16}, \emph{non-separable} denoisers outperform reconstruction techniques based on over-simplified i.i.d.\ models.
In such cases, a more appropriate model might be a finite memory model, well-approximated with a Markov chain prior.
In this paper, we extend the previous performance guarantees for AMP to a class of non-separable sliding-window denoisers, whose promising empirical performance was shown by Ma \emph{et al.}\ \cite{Ma16}, when the unknown signal is produced by a Markov chain starting from its stationary distribution.

When the measurement matrix $A$ has i.i.d.\ Gaussian entries and the empirical distribution\footnote{For an $N$-length vector, by empirical distribution we mean the probability distribution that puts mass $1/N$ on the values taken by each element of the vector.} of the unknown signal $\beta_0$ converges to some probability distribution on $\mathbb{R}$, Bayati and Montanari \cite{BayMont11} proved 
that at each iteration
the performance of AMP can be accurately predicted by a simple, scalar iteration referred to as \emph{state evolution} in the large system limit ($n,N \to \infty$ such that $\frac{n}{N} =\delta$ 
is a constant).  For example, if $\beta^t$ is the estimate produced by AMP at iteration $t$, their result implies that the normalized squared error, $\frac{1}{N}\norm{\beta^t - \beta_0}^2$, and other performance measures converge to known values predicted by state evolution using the prior distribution of $\beta_0$.\footnote{Throughout the paper, $\|\cdot\|$ denotes the Euclidean norm.}
%(The large system limit  is defined as $n,N \to \infty$ such that $\frac{n}{N} =\delta$, a constant.)
Recently, Rush and Venkataramanan \cite{RushV16} provided a concentration version of the asymptotic result when the prior distribution of $\beta_0$ is i.i.d.\ sub-Gaussian.  The result implies that
the probability of $\e$-deviation between various performance measures and their limiting constant values fall exponentially in $n$. Extensions of AMP performance guarantees beyond separable denoisers have been considered in special cases  
\cite{JavMonState13,RushGVISIT15} for
certain classes of block-separable denoisers that allow dependencies within blocks of the signal $\beta_0$ with independence across blocks.  However these settings are more restricted than the types of dependencies we consider here.

\subsection{AMP Algorithm for Sliding-Window Denoiser}
%\emph{\textbf{The AMP Algorithm for the Sliding-Window Denoiser:}}
The AMP algorithm, in the case of a dependent signal, generates successive estimates of the unknown vector denoted by $\beta^t \in \mathbb{R}^N$ for $t=1,2,\ldots$.  These values are calculated as follows: given the observed vector $y$, set $\beta^0=0$, the all-zeros vector. For $t=0,1,\ldots$, fix $k \geq 0$ an integer, and AMP computes
\begin{align}
z^t & = y - A\beta^t + \frac{z^{t-1}}{n} \sum_{i=k+1}^{N-k} \eta_{t-1}^{\prime}( [A^*z^{t-1} + \beta^{t-1}]_{i-k}^{i+k}), \label{eq:amp1}\\
\beta^{t+1}_i & = \begin{cases}
\eta_{t}( [A^*z^t + \beta^t]_{i-k}^{i+k}) \quad &\text{ if } k + 1 \leq i \leq N-k, \\
0 \quad &\text{ otherwise},
 \label{eq:amp2}
\end{cases}
\end{align}
for an appropriately-chosen sequence of non-separable denoiser functions $\{\eta_{t}\}_{t \geq 0}: \mathbb{R}^{2k + 1} \to \mathbb{R}$, where the notation
\ben
[x]_{i-k}^{i+k} = (x_{i-k}, \ldots, x_{i+k}) \in \mathbb{R}^{2k + 1} \quad \text{ for } \quad x \in \mathbb{R}^N,
%\label{def:notation1}
\een
and $A^*$ denotes the transpose of $A$.
%For i.i.d.\ signals $\beta_0$, the denoiser functions $\{\eta_t\}_{t \geq 0}$ act component-wise when applied to a vector and this is referred to as the \emph{separable} case.  For our problem, where $\beta_0$ has dependent entries, we use \emph{non-separable} denoisers $\{\eta_{t}\}: \mathbb{R}^{2k+1} \to \mathbb{R}$ defined for $t \geq 0$.
We let $\eta^{\prime}_{t}$ denote the partial derivative of $\eta_{t}$ with respect to
(w.r.t.) the $(k+1)^{th}$ coordinate, or the center element, assuming the function is differentiable. Quantities with a negative index in \eqref{eq:amp1} and \eqref{eq:amp2} are set to zero.

\subsection{Contributions and Outline}

Our main result proves concentration for order-2 pseudo-Lipschitz (PL) loss functions\footnote{A function $f: \mathbb{R}^{m} \to \mathbb{R}$ is order-2 \emph{pseudo-Lipschitz} if there exists a constant $L >0$ such that for all $x, y \in \mathbb{R}^{m}$, $|f(x)-f(y)|\leq L(1+\norm{x}+\norm{y})\norm{x-y}$.} for the AMP estimate of \eqref{eq:amp2} at any iteration $t$ of the algorithm to constant values predicted by the state evolution equations.  
We envision that our work in understanding the role of sliding-window denoisers within AMP is an early step in characterizing the role of non-separable denoisers within AMP. This work will lead to a characterization of more general denoisers in problems including compressive image reconstruction \cite{Tan15,Metzler16}.

To characterize AMP performance for sliding-window denoisers when the input signal is a Markov chain, we need concentration inequalities for PL functions of Markov chains and sequences of Gaussian vectors that are constructed in a certain way. Specifically, in the constructed sequences, successive elements are successive $(2k+1)$-length overlapping  blocks of some original sequences (another Markov chain or Gaussian sequence, respectively), as suggested by the structure of the denoiser $\eta_t$ in \eqref{eq:amp2}.
%To characterize AMP performance in this way for sliding-window denoisers when the input signal is a Markov chain,
%requires concentration inequalities for PL functions of Markov chains and sequences of Gaussian vectors that are constructed such that successive elements are successive $(2k+1)$-length overlapping  blocks of some original sequences (another Markov chain or Gaussian sequence, respectively), as suggested by the structure of the denoiser $\eta_t$ in \eqref{eq:amp2}. 
These concentration results are proved in Lemmas \ref{lem:PL_overlap_gauss_conc} and \ref{lem:PLMCconc_new} in Appendix \ref{app:conc_dependent}. 
%Because we also need to characterize the bounded derivative of the denoiser function, concentration results for bounded functions of the sequences constructed as described above are also needed; these results are stated in Lemmas \ref{lem:janson} and \ref{lem:BoundedMCconc} in Section \ref{sec:amp_proof}.

%\subsection{Paper Outline}

The rest of the paper is organized as follows. Section \ref{sec:mainresult} provides model assumptions, state evolution for sliding-window denoisers, and the main performance guarantee (Theorem \ref{thm:main_amp_perf}), a concentration result for PL loss functions acting on the AMP estimate from \eqref{eq:amp2}
to the state evolution predictions. Section \ref{sec:amp_proof} proves Theorem \ref{thm:main_amp_perf} with a proof based on two technical results, Lemma \ref{lem:hb_cond} and Lemma \ref{lem:main_lem}, which are proved in Section \ref{sec:main_lem_proof}. %The proof of Lemma \ref{lem:hb_cond} is the same as \cite[Lemma 4]{RushV16} and we prove Lemma \ref{lem:main_lem} in Section \ref{sec:main_lem_proof}.
%Finally, Section \ref{sec:conclusion} concludes the paper and discusses future work.

\section{Main Results} \label{sec:mainresult}

\subsection{Definitions and Assumptions}
\label{subsec:def_assumption}
First we include definitions of properties of Markov chains that will be useful to clarify our assumptions on the unknown signal $\beta_0$.
%throughout the work.

\begin{defi}
\label{def:geom_ergo}
Consider a Markov chain on a state space $S$ with transition probability measure $r(x, dy)$ and stationary distribution $\gamma$. Denote the set of all $\gamma$-square-integrable functions by $L^2(\gamma):=\{f:\mathbb{R}\rightarrow\mathbb{R}:\int_S |f(x)|^2\gamma(dx)<\infty\}$. Define a linear operator $R$ associated with $r(x,dy)$ as $Rf(x):=\int_S f(y)r(x,dy)$ for $f\in L^2(\gamma)$. The chain is said to be \textbf{geometrically ergodic on $L^2(\gamma)$}
if there exists $0 < \rho < 1$ such that for each probability measure $\nu$ that satisfies $\int_S |\frac{d\nu}{d\gamma}|^2d\gamma<\infty$, there is a constant $C_\nu < \infty$ such that
\ben
\sup_{A\in\mathcal{B}(S)} \abs{\int_S r^n(x,A)\nu(dx) - \gamma(A)} < C_\nu \rho^n, \qquad n \in \mathbb{N},
\een
where $\mc{B}(S)$ is the Borel sigma-algebra on $S$ and $r^n(x,dy)$ denotes the $n$-step transition probability measure. In other words, geometrical ergodicity means the chain converges to its stationary distribution $\gamma$ geometrically fast. The chain is said to be \textbf{reversible} if $r(x,dy)\gamma(dx)=r(y,dx)\gamma(dy)$. Moreover, a chain is said to have a \textbf{spectral gap on $L^2(\gamma)$} if
\begin{equation*}
g_2:=1-\sup \{\lambda\in \Lambda:\lambda\neq 1\}>0,
\end{equation*}
where $\Lambda$ is a set of values for $\lambda$ such that $(\lambda \mathsf{I}-R)^{-1}$ does not exist as a bounded linear operator on $L^2(\gamma)$. Note that for a countable state space $S$, $\Lambda$ is the set of all eigenvalues of the transition probability matrix, hence $g_2$ is the distance between the largest and the second largest eigenvalues.
\end{defi}

It has been proved that a Markov chain has spectral gap on $L^2(\gamma)$ if and only if it is reversible and geometrically ergodic \cite{roberts1997}. We use the existence of a spectral gap to prove concentration results for PL functions with dependent input, where the dependence is characterized by a Markov chain.
%The existence of spectral gap is required in our proof for the concentration results about PL functions functions of Markov chains. 
Such concentration results are crucial for obtaining the main technical lemma, Lemma \ref{lem:main_lem}, and hence our main result, Theorem \ref{thm:main_amp_perf}. If the spectral gap does not exist, meaning that $g_2=0$, then our proof only bounds the probability of tail events in Lemma \ref{lem:main_lem} by constant 1, which is useless.

With this definition, we now clarify the assumptions under which our result is proved. 
%The assumptions on the measurement matrix $A$ and noise $w$ are the same as in Rush and Venkataramanan \cite{RushV16} and are not repeated here. We state our assumptions on the non-i.i.d.\ signal $\beta_0$ and the non-separable denoiser $\eta_t$ below.

\textbf{\emph{Assumptions}}:

\emph{Signal:} Let $S\subset\mathbb{R}$ be a bounded state space (countable or uncountable).  We assume that the signal $\beta_0 \in S^N$ is produced by a time-homogeneous, reversible, geometrically ergodic Markov chain in its (unique) stationary distribution.  Note that this means the `sequence' $\beta_{0_1}, \beta_{0_2}, \ldots, \beta_{0_N}$, where $\beta_{0_i}$ is element $i$ of $\beta_0$, forms a Markov chain.  We refer to the stationary distribution as $\gamma_{\beta}$. 

\emph{Denoiser functions:} The denoiser functions $\eta_t:\mathbb{R}^{2k+1}\rightarrow\mathbb{R}$ used in \eqref{eq:amp2} are assumed to be Lipschitz\footnote{A function $f: \mathbb{R}^{m} \to \mathbb{R}$ is \emph{Lipschitz} if there exists a constant $L >0$ such that for all $x, y \in \mathbb{R}^{m}$,
$\abs{f(x) - f(y)} \leq L\norm{x-y}$, where $\norm{\cdot}$ denotes the Euclidean norm.}
 for each $t>0$ and differentiable %therefore, are also weakly differentiable with bounded derivative. The weak partial derivative 
 w.r.t.\ the $(k+1)^{th}$ (middle) coordinate with bounded derivative denoted by $\eta_t'$. Further, the derivative $\eta_t'$ is assumed to be differentiable with bounded derivative.
%, meaning that all $2k+1$ partial derivatives exist and are bounded. 
Note that this implies $\eta_t'$ is Lipschitz. (It is possible to put a weaker condition on $\eta_t$ as in \cite{RushV16}. That is, $\eta_t$ is Lipschitz, hence weakly differentiable with bounded derivative. The weak derivative w.r.t. the $(k+1)^{th}$ coordinate, denoted by $\eta_t'$, is assumed to be differentiable except at a finite number of points; the derivative of $\eta_t'$ is assumed to be bounded when it exists.)

%(It is possible to weaken this condition to allow  $\eta_t'$ to have a finite number of discontinuities, if needed, as in \cite{RushV16}.)
%, except possibly at a finite number of points, with bounded derivative where it exists.

\emph{Matrix:} The entries of the matrix $A$ are i.i.d.\ $\sim\mc{N}(0,1/n)$.

\emph{Noise:} The entries of the measurement noise vector $w$ are i.i.d.\ according to some sub-Gaussian distribution $p_w$ with mean 0 and finite variance $\sigma^2$. The sub-Gaussian assumption implies that for all $\e\in(0,1)$,
\begin{equation*}
P\left(\abs{\frac{1}{n}\|w\|^2 - \sigma^2}\geq \e\right)\leq Ke^{-\kappa n \e^2}
\end{equation*}
for some constants $K,\kappa>0$ \cite{BLMConc}.

\subsection{Performance Guarantee}

As mentioned in Section \ref{sec:intro}, the behavior of the AMP algorithm is predicted by a simple, scalar iteration referred to as state evolution, which we introduce here.  Let the stationary distribution $\gamma_{\beta}$ and the transition probability measure $r(x,dy)$ define the prior distribution for the unknown vector $\beta_0$ in \eqref{eq:model1}. Let the random variable $\beta \in S$ be distributed as $\gamma_{\beta}$ and the random vector $\underline{\beta} \in S^{2k+1}$ be distributed as $\pi$, where
\begin{equation}
\pi(dx)=\pi((dx_1,...,dx_{2k+1}))=\prod_{i=2}^{2k+1}r(x_{i-1},dx_{i})\gamma_{\beta}(dx_1)
\label{eq:pi_def}
\end{equation} 
is the probability of seeing such a length-$(2k+1)$ sequence in the $\beta_0$ Markov chain (i.e.\ it is the $(2k+1)$-dimension marginal distribution of $\beta_0$).  %Note that you could equivalently think of $\beta$ as being the $(k+1)^{th}$ entry of $\underline{\beta}$ since $\underline{\beta}$ is a Markov chain starting in its stationary distribution (i.e.\ $\gamma_\beta$ is the one-dimension marginal distribution of $\beta_0$). 
Define $\sigma_\beta^2 = \mathbb{E}[\beta^2] > 0$ and $\sigma_0^2 =\sigma_\beta^2/\delta$. Iteratively define the quantities $\{\sigma_t^2\}_{t \geq 1}$ and $\{\tau_t^2\}_{t \geq 0}$ as follows,
\be
\begin{split}
\tau_t^2 &= \sigma^2 + \sigma_t^2, \\
\sigma_{t+1}^2 &= \frac{1}{\delta}\left((1-w_k)\mathbb{E}\left[\left(\eta_{t}(\underline{\beta} + \tau_{t} \underline{Z}) - \underline{\beta}_{k+1}\right)^2 \right] + w_k \sigma_{\beta}^2\right), \label{eq:taut_sigmat_def}
\end{split}
\ee
with $\underline{\beta}_{k+1}$ the $(k+1)^{th}$ entry of $\underline{\beta}$, $\underline{Z} \in \mathbb{R}^{2k+1} \sim \mc{N}(0, \mathbb{I}_{2k+1})$ independent of $\underline{\beta}$, $w_k = \frac{2k}{N}$, and $\delta = \frac{n}{N}$.

Theorem \ref{thm:main_amp_perf} provides our main performance guarantee, 
which is a concentration inequality for pseudo-Lipschitz (PL) loss functions.
%\footnote{A function $\phi: \mathbb{R}^{m} \to \mathbb{R}$ is pseudo-Lipschitz (of order $2$)  if there exists a constant $L >0$ such that for all $x,y \in \mathbb{R}^{m}$,
%$\abs{\phi(x) - \phi(y)} \leq L ( 1 + \norm{x} + \norm{y} ) \norm{x-y}$,
%where $\norm{\cdot}$ denotes the Euclidean norm.}

\begin{thm}\label{thm:main_amp_perf}
With the assumptions of Section \ref{subsec:def_assumption}, for any order-$2$ pseudo-Lipschitz function $\phi: \mathbb{R}^{2} \rightarrow \mathbb{R}$, $\e \in (0,1)$, and  $t \geq 0$,
\be
P\left(\left \lvert \sum_{i=k+1}^{N-k} \frac{\phi(\beta^{t+1}_{i}, \beta_{0_i})}{N-2k}  - \mathbb{E}[\phi(\eta_t (\underline{\beta} + \tau_t \underline{Z}), \underline{\beta}_{k+1})]\right \lvert \geq \e\right) \leq K_{k,t}e^{-\kappa_{k,t} n \e^2}.
\label{eq:PL_conc_result}
\ee
In the expectation in \eqref{eq:PL_conc_result}, $\underline{\beta} \in S^{2k+1} \sim \pi$, $\underline{\beta}_{k+1}$ is the $(k+1)^{th}$ element of $\underline{\beta}$, and $\underline{Z} \in \mathbb{R}^{2k+1}  \sim \mc{N}(0, \mathbb{I}_{2k+1})$ independent of $\underline{\beta}$.  The constant $\tau_t$ is defined in \eqref{eq:taut_sigmat_def}
and constants $K_{k,t}, \kappa_{k,t} > 0$ depend on the iteration index $t$ and half window-size $k$, but not on $n$ or $\e$ and are not exactly specified.
\end{thm}

\textbf{\emph{Remarks}}:

(1) The probability in \eqref{eq:PL_conc_result} is w.r.t.\ the product measure on the space of the matrix $A$, signal $\beta_0$, and noise $w$.

(2) Theorem \ref{thm:main_amp_perf} shows concentration for the loss when considering only the inner $N-2k$ elements of the signal.  This is due to the nature of the sliding-window denoiser, which updates each element of the estimate $\beta^t$ using the $k$ elements on either side of that location.  In practice, as in Ma \emph{et al.}\ \cite{Ma16}, one could run a slightly different algorithm than that given in \eqref{eq:amp1}-\eqref{eq:amp2}: instead of setting the end elements, meaning the first $k$ and last $k$ elements, of the estimate $\beta^t$ equal to $0$, update these elements using the sliding-window denoiser but with missing input values replaced by the median of the other inputs.  Such a strategy shows good empirical performance -- even at the end elements -- and suggests that the concentration result of Theorem \ref{thm:main_amp_perf} could be extended to show concentration for the loss of the full signal.  Proving this requires a delicate handling of the end elements and is left for future research.

(3) The state evolution constants $\{\tau_t^2\}_{t \geq 0}$ defined in \eqref{eq:taut_sigmat_def} are the sum of $\sigma^2$ 
and two weighted terms, where the weight depends on $k$, the length of the window in the sliding-window denoiser.  Since we only estimate the middle $N-2k$ elements of the signal, as $k$ increases the state evolution constants $\{\tau_t^2\}_{t \geq 0}$ depend more on the second moment of the one-dimensional marginals of the original signal, corresponding to the estimation error in the un-estimated part of the signal.

(4) By choosing PL loss, $\phi(a,b) = (a - b)^2$, Theorem \ref{thm:main_amp_perf} gives the following concentration result for the mean squared error of the middle $N-2k$ coordinates of the estimates.  For all $t \geq 0$,
\ben
P\left(\left \lvert \sum_{i=k+1}^{N-k} \frac{(\beta^{t+1}_{i}- \beta_{0_i})^2}{N-2k}  - \frac{n(\tau_{t+1}^2 - \sigma^2) - 2k\sigma_{\beta}^2}{N - 2k}\right \lvert \geq \e\right)
\leq K_{k,t} e^{-\kappa_{k,t} n \e^2},
%\label{eq:PL_conc_result}
\een
with $\tau_{t+1}^2$ defined in \eqref{eq:taut_sigmat_def}. A numerical example demonstrating that the MSE of the AMP estimates $\{\beta^t\}_{t\geq 0}$ is tracked by the state evolution iteration \eqref{eq:taut_sigmat_def} is proved in Section \ref{subsec:numerical}.

\subsection{A Numerical Example}
\label{subsec:numerical}

We now provide a concrete numerical example where AMP is used to estimate $\beta_0$ from the linear system \eqref{eq:model1}, when the entries of $\beta_0$ form a Markov chain on state space $\{0,1\}$ starting from its stationary distribution. The transition probability measure is $r(0,1) = 3/70$ and $r(1,0) = 1/10$, which yields a unique stationary distribution $\gamma_\beta(1) = 1 - \gamma_{\beta}(0) = 3/10$.

We define the denoiser function $\eta_t$ in \eqref{eq:amp2} as the Bayesian sliding-window denoiser. Note that an important key property of AMP is the following: for large $n$ and for $k+1 \leq i \leq N-k$, the observation vector $[A^*z^t + \beta^t]_{i-k}^{i+k}$ used as input to the estimation function in \eqref{eq:amp2} is approximately distributed as $\underline{\beta} + \tau_t \underline{Z}$, where $\underline{\beta} \sim \pi$ with $\pi(x_1,...,x_{2k+1}):=\prod_{i=2}^{2k+1}r(x_{i-1},x_i)\gamma_\beta(x_1),$ $\underline{Z}\sim\mc{N}(0,\textsf{I}_{2k+1})$ independent of $\underline{\beta}$, and $\tau_t$ is defined in \eqref{eq:taut_sigmat_def}.

The above property gives us a natural way to define the Bayesian sliding-window denoiser. That is, suppose $\underline{V}_i = [A^*z^t + \beta^t]_{i-k}^{i+k} \in \mathbb{R}^{2k+1}$.  Then, define 
 \[\eta_t([A^*z^t + \beta^t]_{i-k}^{i+k}) = \eta_t(\underline{V}_i) := \mathbb{E}\left[\underline{\beta}_{k+1}\left\vert  \underline{\beta} + \tau_t \underline{Z} =  \underline{V}_i \right.\right],\] 
 where $\underline{\beta}_{k+1}$ denotes the $(k+1)^{th}$ element of $\underline{\beta}$.  Figure \ref{fig:SE} shows that the mean squared error (MSE) achieved by AMP with the non-separable sliding-window denoiser defined above is tracked by state evolution at every iteration.

%We define the denoiser function $\eta_t$ in \eqref{eq:amp2} as the Bayesian sliding-window denoiser. That is, suppose $\underline{V} = \underline{\beta} + \tau_t \underline{Z} \in \mathbb{R}^{2k+1}$, where $\underline{\beta} \sim \pi$ with $\pi(x_1,...,x_{2k+1}):=\prod_{i=2}^{2k+1}r(x_{i-1},x_i)\gamma_\beta(x_1),$
%$\underline{Z}\sim\mc{N}(0,\textsf{I}_{2k+1})$ independent of $\underline{\beta}$, and $\tau_t$ is defined in \eqref{eq:taut_sigmat_def}.  Then, define $\eta_t([A^*z^t + \beta^t]_{i-k}^{i+k}): = \mathbb{E}\left[\underline{\beta}_{k+1}\left\vert\underline{V}=[A^*z^t + \beta^t]_{i-k}^{i+k}\right.\right].$  Figure \ref{fig:SE} shows that the mean squared error (MSE) achieved by AMP with the non-separable sliding-window denoiser defined above is tracked by state evolution at every iteration. 

Notice that when $k=0$, the denoisers $\{\eta_t\}_{t \geq 0}$ are separable and the empirical distribution of $\beta_0$ converges to the stationary probability distribution $\gamma_{\beta}$ on 
$\mathbb{R}$.  For this case, the state evolution analysis for AMP with separable denoisers ($k=0$) was justified by Bayati and Montanari \cite{BayMont11}. However, it can be seen in Figure \ref{fig:SE} that the MSE achieved by the separable denoiser ($k=0$) is significantly higher (worse) than that achieved by the non-separable denoisers ($k=1,2$).

\begin{figure}
\centering
\includegraphics[width=0.6\textwidth]{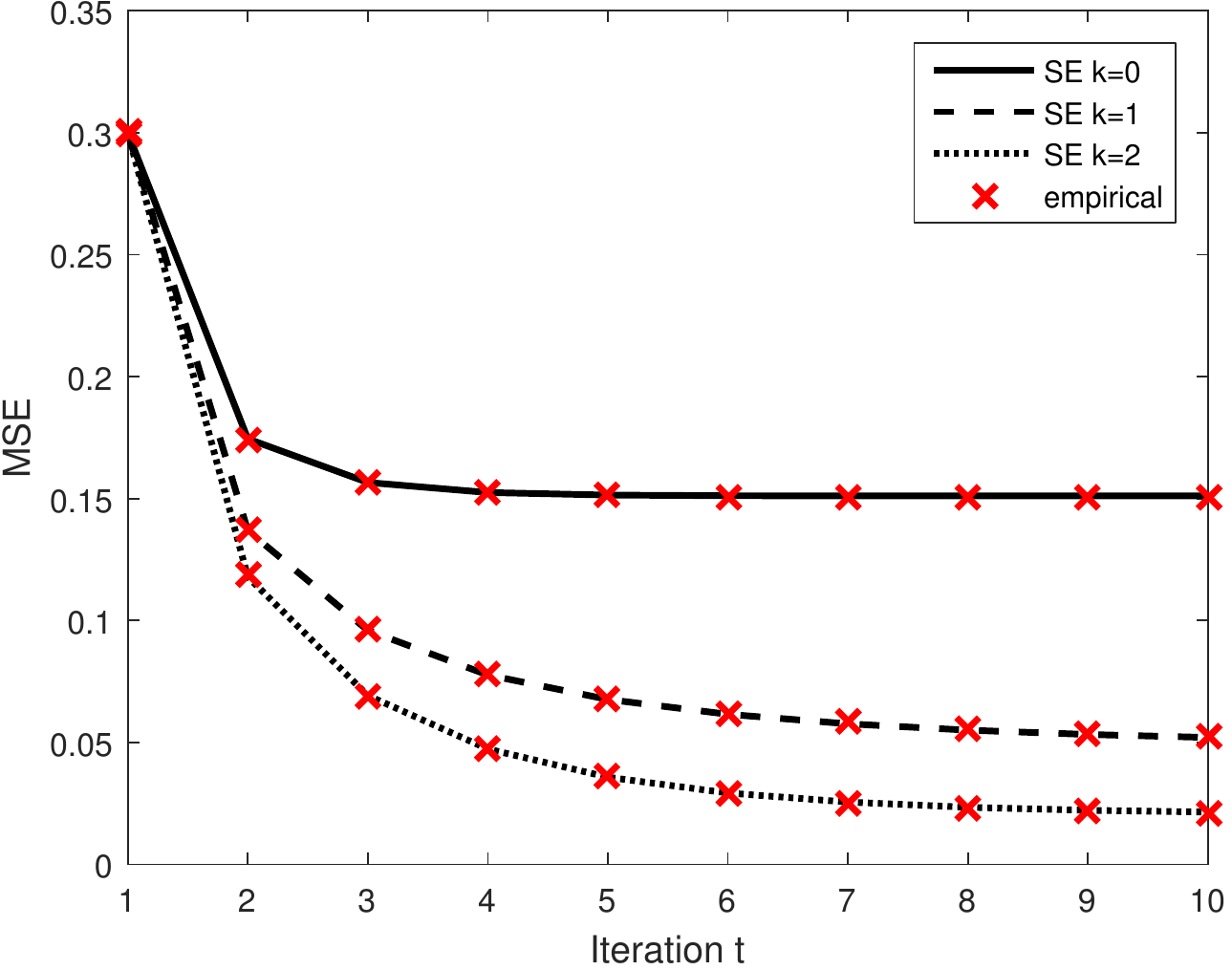}
\caption{ The plot provides the results of a numerical example to demonstrate the validation of state evolution of AMP with non-separable sliding-window denoisers. The black curves are theoretical state evolution predictions given by \eqref{eq:taut_sigmat_def} with three different half window-sizes, $k$. The red crosses are empirical MSE achieved by the AMP algorithm defined in \eqref{eq:amp1} and \eqref{eq:amp2}. ($N=10,000$, $\delta=0.3$, $\sigma^2=0.1$.)}
\label{fig:SE}
\end{figure}

% !TEX root = main.tex

\section{Proof of Theorem \ref{thm:main_amp_perf}} \label{sec:amp_proof}

The proof of Theorem \ref{thm:main_amp_perf} follows 
the work of Rush and Venkataramanan~\cite{RushV16},
with modifications for the dependent structure of the unknown vector $\beta_0$ in \eqref{eq:model1}.  
For this reason, we use much of the same notation.
The main ingredients in the proof of  Theorem \ref{thm:main_amp_perf} are two technical lemmas corresponding to \cite[Lemmas 4 and 5]{RushV16}. 
%For brevity, we present only the main differences from \cite{RushV16}, while the comprehensive version of the proof can be found in \cite{long_version}.
We first cover some preliminary results and establish notation used in the proof. We then discuss the lemmas used to prove Theorem \ref{thm:main_amp_perf}.

% !TEX root = main.tex

\subsection{Proof Notation} \label{subsec:defs}
As mentioned above, in order to streamline the proof of our technical lemmas we use notation similar to \cite{RushV16} and consequently to \cite{BayMont11}.  As in the previous work, the technical lemmas are proved for a more general recursion which we define in the following, with AMP being a specific example of the general recursion as shown below.

Given noise $w \in \mathbb{R}^n$ and unknown signal $\beta_0 \in S^N$, fix the half-window-size $k>0$ an integer, define column vectors $h^{t+1}, q^{t+1} \in \mathbb{R}^N$ and $b^t, m^t \in \mathbb{R}^n$ for $t \geq 0$ recursively as follows, starting with initial condition $q^0 \in \mathbb{R}^N$:
\begin{equation}
\begin{split}
h^{t+1} := A^*m^t - \xi_t q^t, \qquad &  q^{t}_i  := \begin{cases}f_t ([h^t]_{i-k}^{i+k}, [\beta_0]_{i-k}^{i+k}), \qquad &\text{ if } k+1 \leq i \leq N-k,\\
-\beta_{0_i},\qquad &\text{ otherwise},\end{cases}\\
b^t := A q^t - \lambda_t m^{t-1},\qquad & m^t  := g_t(b^t, w),
\end{split}
\label{eq:hqbm_def}
\end{equation}
with scalar values $\xi_t$ and $\lambda_t$ defined as
\be 
\xi_t := \frac{1}{n} \sum_{i=1}^n g_t'(b^t_i, w_i), \quad 
\lambda_t := \frac{1-w_k}{\delta(N-2k)} \sum_{i=k+1}^{N-k} f_t'([h^t]_{i-k}^{i+k}, [\beta_{0}]_{i-k}^{i+k}), 
\label{eq:xi_lamb_def} 
\ee 
where $w_k = 2k/N$.  For the derivatives in \eqref{eq:xi_lamb_def}, the derivative of $g_t: \mathbb{R}^2 \rightarrow \mathbb{R}$ is with respect to the first argument and the derivative of $f_t: \mathbb{R}^{2(2k+1)} \rightarrow \mathbb{R}$ is with respect to the $(k+1)^{th}$, or center coordinate, of the first argument.  The functions $\{f_t\}_{t \geq 0}$, and $\{g_t\}_{t \geq 0}$ are assumed to be Lipschitz continuous and differentiable with bounded derivatives $g_t'$ and $f_t'$.  %meaning the the weak derivatives  $g_t'$ and $f_t'$ exist and are bounded.  
Further, $g'_t$ and $f'_t$ are each assumed to be differentiable in the first argument with bounded derivative.  For $f'_t$ this means that we assume the first $2k+1$ partial derivatives exist and are bounded.
%, except possibly at a finite number of points, with bounded derivative everywhere it exists.

Recall that the unknown vector $\beta_0\in S^N$ is assumed to have a Markov chain prior with transition probability measure $r(x,dy)$ and stationary probability measure $\gamma_\beta$. Let $\beta\in S\sim \gamma_\beta$ and $\underline{\beta}\in S^{2k+1}\sim\pi$ where $\pi$ is defined in \eqref{eq:pi_def}. Note that $\pi$ is the $(2k+1)$-dimension marginal distribution of $\beta_0$ and $\gamma_{\beta}$ is the one-dimensional marginal distribution.

Let $\underline{0}\in\mathbb{R}^{2k+1}$ be a vector of zeros. Define
\begin{align}
\sigma_{\beta}^2 := \mathbb{E}[\beta^2], \quad \text{ and } \quad \sigma_0^2 &:= \frac{1}{\delta}\left((1-w_k) \mathbb{E}\left[f_0^2(\underline{0},\underline{\beta})\right] + w_k \sigma_{\beta}^2 \right) >0. \label{eq:sigma0}
\end{align}
%for $\beta\in\mathbb{R}\sim\gamma_{\beta}$, $\underline{\beta} \sim \pi$, and $\underline{0} \in \mathbb{R}^{2k+1}$ a vector of 0's.  We assume $\sigma_0^2$ is strictly positive.  
Further, let
\be
 q^{0}_i  := \begin{cases}f_0 (\underline{0}, [\beta_0]_{i-k}^{i+k}), \qquad &\text{ if } k+1 \leq i \leq N-k \\
-\beta_{0_i},\qquad &\text{ otherwise}\end{cases}
 \label{eq:q0def}
\ee
and assume that there exist constants $K, \kappa > 0$ such that
\be
P\left(\left\lvert \frac{1}{n} \norm{q^0}^2 - \sigma_0^2\right \lvert \geq \e\right) \leq K e^{-\kappa n \e^2}.
\label{eq:qassumption}
\ee
Define the state evolution scalars $\{\tau_t^2\}_{t \geq 0}$ and $\{\sigma_t^2\}_{t \geq 1}$ for the general recursion as follows.
\be
\tau_t^2 := \mathbb{E}\left[(g_{t}(\sigma_{t} Z, W))^2\right],  \quad \quad \sigma_t^2 := \frac{1}{\delta}\left((1-w_k)\mathbb{E}\left[(f_{t}(\tau_{t-1} \underline{Z}, \underline{\beta}))^2\right] + w_k \sigma_{\beta}^2\right), \label{eq:sigmatdef}
\ee
where random variables $W \sim p_w$ and $Z \sim \mc{N}(0,1)$ and random vectors $\underline{\beta}\in S^{2k+1} \sim \pi$ and $\underline{Z}\in\mathbb{R}^{2k+1} \sim \mc{N}(0,\mathsf{I}_{2k+1})$ are independent.  We assume that both $\sigma_0^2$ and $\tau_0^2$ are strictly positive.

We note that the AMP algorithm introduced in \eqref{eq:amp1} and \eqref{eq:amp2} is a special case of the general recursion introduced \eqref{eq:hqbm_def} and \eqref{eq:xi_lamb_def}.  Indeed, define the following vectors  recursively for $t\geq 0$, starting with $\beta^0=0$ and $z^0=y$. 
\begin{equation}
\begin{split}
h^{t+1} = \beta_0 - (A^*z^t + \beta^t), \qquad &  q^t  =\beta^t - \beta_0, \\
b^t = w-z^t,\qquad & m^t  =-z^t.
\end{split}
\label{eq:hqbm_def_AMP}
\end{equation}
It can be verified that these vectors satisfy  \eqref{eq:hqbm_def} and \eqref{eq:xi_lamb_def} with Lipschitz functions
\be
f_t(\underline{a}, [\beta_0]_{i-k}^{i+k}) = \eta_{t-1}\left([\beta_0]_{i-k}^{i+k} - \underline{a}\right) - \beta_{0_i}, \quad \text{ and } \quad g_t(b, w) = b-w,
\label{eq:fg_def}
\ee
where $ \underline{a} \in\mathbb{R}^{2k+1}$ and $b\in\mathbb{R}$.  Using $f_t, g_t$ defined in \eqref{eq:fg_def} in \eqref{eq:sigmatdef} yields the expressions for $\sigma_t^2, \tau_t^2$ in \eqref{eq:taut_sigmat_def}.  We note that by Lemma \ref{lem:PLMCconc_new} with $d=1$ and $f(x):=x^2$,
\begin{equation}
P\left( \abs{\frac{1}{N}\|\beta_0\|^2 - \sigma_{\beta}^2} \geq \epsilon \right)\leq K e^{-\kappa N \e^2},
\label{eq:beta0_assumption}
\end{equation}
and so for the AMP algorithm using $f_t$ from \eqref{eq:fg_def} in \eqref{eq:q0def}, assumption \eqref{eq:qassumption} is satisfied.

In what follows, the notation matches that of \cite{RushV16} but is repeated here for completeness.  In the remaining analysis, the general recursion given in \eqref{eq:hqbm_def} and \eqref{eq:xi_lamb_def} is used.  We can write vector equations to represent the recursion as follows: for all $t \geq 0$,
\begin{equation}
b^{t} + \lambda_t m^{t-1} = A q^t, \quad \text{ and } \quad h^{t+1} + \xi_t q^t = A^* m^t.
\label{eq:bmq}
\end{equation}
This yields matrix equations $X_t = A^* M_t$ and $Y_t =AQ_t,$ where we define the individual matrices as
\be
\label{eq:XYMQt}
\begin{split}
X_t  := [h^1 + \xi_0 q^0 \mid h^2+\xi_1 q^1 \mid \ldots \mid h^t + \xi_{t-1} q^{t-1}],   & \qquad  Q_t  :=  [q^0 \mid \ldots \mid q^{t-1} ],\\
Y_t  := [b^0 \mid b^1 + \lambda_1 m^0 \mid \ldots \mid b^{t-1} + \lambda_{t-1} m^{t-2}], & \qquad M_t  := [m^0 \mid \ldots \mid m^{t-1} ] \\
 \Xi_{t}:= \text{diag}(\xi_0, \ldots, \xi_{t-1}) &\qquad H_t := [h^1 | \ldots | h^{t}], \\
 \Lambda_t := \text{diag}(\lambda_0, \ldots, \lambda_{t-1}) &\qquad  B_t := [b^0 | \ldots | b^{t-1}]. 
\end{split}
\ee
In the above, $[c_1 \mid c_2 \mid \ldots \mid c_k]$ denotes a matrix with columns $c_1, \ldots, c_k$ and $M_0$, $Q_0$, $B_0$, $H_0$, and $\Lambda_0$ are defined to be the all-zero vector.  From the above matrix definitions we have the following matrix equations $Y_t = B_t + \Lambda_t [0 | M_{t-1}]$ and $X_t = H_t + \Xi_t Q_{t}.$

The values $m^t_{\|}$ and $q^t_{\|}$ are projections of $m^t$ and $q^t$ onto the column space of $M_t$ and $Q_t$, with $m^t_{\perp} := m^t - m^t_{\|},$ and $q^t_{\perp} := q^t - q^t_{\|}$ being the projections onto the orthogonal complements of $M_t$ and $Q_t$.  Finally, define the vectors
 \be 
 \alpha^t := (\alpha^t_0, \ldots, \alpha^t_{t-1})^*, \qquad  \gamma^t :=  (\gamma^t_0, \ldots, \gamma^t_{t-1})^* 
 \label{eq:vec_alph_gam_conc}
 \ee 
 to be the coefficient vectors of the parallel projections, i.e.,
 \be
 m^t_{\| } := \sum_{i=0}^{t-1} \alpha^t_i m^i, \qquad  q^t_{\|} := \sum_{i=0}^{t-1} \gamma^t_i q^i.
 \label{eq:mtqt_par}
 \ee
The technical lemma, Lemma \ref{lem:main_lem},  shows that for large $n$, the entries of the vectors $\alpha^t$ and ${\gamma}^t$ concentrate to constant values which are defined in the following section.

 %%%%%%%%%%%%%%%%%%%%%%%%%%
%%%%%%%%%%%%%%%%%%%%%%%%%%

\subsection{Concentrating Constants} \label{subsec:concvals}

Recall that $\beta_0 \in S^N$ is the unknown vector to be recovered and $w \in \mathbb{R}^n$  is the measurement noise.   Using the definitions in \eqref{eq:hqbm_def_AMP}, note that the vector $h^{t+1}$ is the noise in the observation $A^*z^t + \beta^t$ (from the true $\beta_0$), while $q^t$ is the error in the estimate $\beta^t$.  The technical lemma will show that $h^{t+1}$ can be approximated as i.i.d.\ $\mc{N}(0, \tau_t^2)$ in functions of interest for the problem, namely when used as input to PL functions, and $b^t$ can be approximated as i.i.d.\ $\mc{N}(0, \sigma_t^2)$ in PL functions. Moreover, the deviations of the quantities  $\frac{1}{n}\|m^t\|^2$ and $\frac{1}{n}\|q^t\|^2$  from $\tau_t^2$ and $\sigma_t^2$, respectively, fall exponentially in $n$.  In this section we introduce the concentrating values for various inner products of the values $\{h^t, m^t, q^t, b^t\}$ that are used in Lemma \ref{lem:main_lem}.

First define the concentrating values for $\lambda_{t+1}$ and $\xi_{t}$ defined in \eqref{eq:xi_lamb_def} as
\be
\hat{\lambda}_{t+1} := \left(\frac{1-w_k}{\delta} \right) \mathbb{E}\left[f'_{t}(\tau_{t} \underline{\tilde{Z}}_{t}, \underline{\beta})\right], \quad \text{ and } \hat{\xi}_{t} = \mathbb{E}\left[g'_{t}(\sigma_{t} \breve{Z}_{t}, W)\right].
\label{eq:hatlambda_hatxi}
\ee
Next, let $\{\breve{Z}_t \}_{t \geq 0}$ be a sequence of zero-mean jointly Gaussian random variables on $\mathbb{R}$,
and let  $\{ \underline{\tilde{Z}}_t \}_{t \geq 0}$ be a sequence of zero-mean jointly Gaussian random vectors on $\mathbb{R}^{2k+1}$, where $\underline{\tilde{Z}}_t$ has i.i.d.\ coordinates for all $t \geq 0$, and $\underline{\tilde{Z}}_{t_i}$ and $\underline{\tilde{Z}}_{r_j}$ are independent when $i\neq j$. The covariance of the two random sequences is defined recursively as follows. For $r,t \geq 0$,
 \be 
\expec[\breve{Z}_r \breve{Z}_t] = \frac{\tilde{E}_{r,t}}{\sigma_r \sigma_t}  , \qquad \expec[\underline{\tilde{Z}}_{r_i} \underline{\tilde{Z}}_{t_i}] = \frac{\breve{E}_{r,t}}{\tau_r \tau_t}, \quad \text{ for } i = 1, \dots, 2k+1 
\label{eq:tildeZcov}
\ee
where
\be
\begin{split}
\tilde{E}_{r,t} := \frac{1}{\delta}\left((1 - w_k)\mathbb{E}[f_{r}(\tau_{r-1} \underline{\tilde{Z}}_{r-1}, \underline{\beta}) f_{t}(\tau_{t-1} \underline{\tilde{Z}}_{t-1}, \underline{\beta})] + w_k \sigma_{\beta}^2\right), \quad \breve{E}_{r,t} := \mathbb{E}[g_{r}(\sigma_{r} \breve{Z}_{r}, W) g_{t}(\sigma_{t} \breve{Z}_{t}, W)],
\end{split}
\label{eq:Edef}
\ee
with $w_k = 2k/N$ and $\sigma_{\beta}^2$ was defined in \eqref{eq:sigma0}. Note that both terms of the above are scalar values and we take $f_0(\cdot, \underline{\beta}) := f_0(\underline{0},\underline{\beta})$, the initial condition.  Moreover, $\tilde{E}_{t,t} = \sigma_t^2$ and $\breve{E}_{t,t} = \tau_t^2$, as can be seen from \eqref{eq:sigmatdef}, thus $\expec[\underline{\tilde{Z}}^2_{t_i}] = \expec[\breve{Z}^2_t] =1$.

Define matrices $\tilde{C}^t, \breve{C}^t \in \mathbb{R}^{t \times t}$ for $t \geq 1$ taking values from \eqref{eq:Edef} as
\be
\tilde{C}^{t}_{i+1,j+1} = \tilde{E}_{i,j}, \quad \text{ and } \quad \breve{C}^{t}_{i+1,j+1} = \breve{E}_{i,j}, \quad 0\leq i,j \leq t-1.
\label{eq:Ct_def}
\ee
Then,  concentrating values for $\gamma^t$ and $\alpha^t$ defined in \eqref{eq:vec_alph_gam_conc} (as long as $\tilde{C}^{t}$ and $\breve{C}^{t}$ are invertible) are
\be
\hat{\gamma}^{t} := (\tilde{C}^t)^{-1}\tilde{E}_t,  \quad \text{ and } \quad \hat{\alpha}^{t} := (\breve{C}^t)^{-1}\breve{E}_t,
\label{eq:hatalph_hatgam_def}
\ee
where
\be \tilde{E}_t:= (\tilde{E}_{0,t} \ldots, \tilde{E}_{t-1,t})^*, \quad \text{ and } \quad \breve{E}_t:=(\breve{E}_{0,t} \ldots, \breve{E}_{t-1,t})^*. 
\label{eq:Et_def}
\ee
%We prove that $\tilde{C}^t$ and $\breve{C}^t$ are invertible in the next subsection.
%Note that the upper bounds on $\tilde{E}_{r,t}$ and $\breve{E}_{r,t}$ imply upper bounds for the entries of $\hat{\gamma}^{t}$ and $\hat{\alpha}^{t}$.  The see this note that $\hat{\gamma}^{t}_i$ for $1 \leq i \leq t$ is the inner product of the $i^{th}$ row of $(\tilde{C}^t)^{-1}$ with $\tilde{E}_t$, both of which have bounded norm since the elements of $(\tilde{C}^t)^{-1}$.  
Finally, define the values $(\sigma^{\perp}_0)^2 := \sigma_0^2$ and $(\tau^{\perp}_0)^2 := \tau_0^2$, and for $t > 0$
\be
\begin{split}
& (\sigma_{t}^{\perp})^2 := \sigma_t^2 - (\hat{\gamma}^{t})^* \tilde{E}_{t}= \tilde{E}_{t,t} - \tilde{E}^*_{t} (\tilde{C}^{t})^{-1}  \tilde{E}_{t}, \\
& (\tau^{\perp}_{t})^2 := \tau_{t}^2 - (\hat{\alpha}^{t})^* \breve{E}_t =  \breve{E}_{t,t} -   \breve{E}^*_{t} (\breve{C}^{t})^{-1}  \breve{E}_{t}.
\label{eq:sigperp_defs}
\end{split}
\ee
\begin{lem} \label{lem:Ct_invert}
If $(\sigma_k^{\perp})^2$ and $(\tau_k^{\perp})^2$ are bounded below by some positive constants for $k \leq t$, then the matrices $\tilde{C}^{k+1}$ and $\breve{C}^{k+1}$ defined in \eqref{eq:Ct_def} are invertible for $k \leq t$. \end{lem}

\begin{proof}
The proof can be found in \cite[Lemma 4.1]{RushV16}. 

%We prove this by induction. Note that  $\tilde{C}^1= \sigma_0^2$ and $\breve{C}^1 = \tau_0^2$ are both strictly positive by assumption and hence invertible.  Assume  that for some $k <t$, $\tilde{C}^k$ and $\breve{C}^k$  are invertible.   The matrix $\tilde{C}^{k+1}$ can be written as
%\ben
%\tilde{C}^{k+1} = \begin{bmatrix}
%    \mathbf{M}_1  &  \textbf{M}_2 \\
%    \textbf{M}_3 &  \textbf{M}_4
%\end{bmatrix},
%\een
%where $\textbf{M}_1 = \tilde{C}^{k} \in \mathbb{R}^{k \times k}$, $\textbf{M}_4 = \tilde{E}_{k,k} = \sigma_{k}^2$, and $\textbf{M}_2 = \textbf{M}^*_3 = \tilde{E}_{k} \in \mathbb{R}^{k \times 1}$ defined in \eqref{eq:Et_def}.  
%
%By the block inversion formula, $\tilde{C}^{k+1}$ is invertible if $\mathbf{M}_1$ and the Schur complement $\mathbf{M}_4 - \mathbf{M}_3\mathbf{M}_1^{-1}\mathbf{M}_2$ are both invertible. By the induction hypothesis, $\mathbf{M}_1 = \tilde{C}^{k}$ is invertible and 
%\be \mathbf{M}_4 - \mathbf{M}_3\mathbf{M}_1^{-1}\mathbf{M}_2 =   \tilde{E}_{k,k} -   \tilde{E}^*_{k} (\tilde{C}^{k})^{-1}  \tilde{E}_{k} = (\sigma_{k}^{\perp})^2  \geq \tilde{c} > 0,
%\label{eq:sigt_perp_pos}
%\ee
%Hence  $\tilde{C}^{t+1}$ is invertible.  Showing that $\breve{C}^{t+1}$ is invertible is very similar.
\end{proof}
%%%%%%%%%%%%%%%%%%%
%%%%%%%%%%%%%%%%%%%

\subsection{Conditional Distribution Lemma} \label{sec:cond_dist_lemma}

As mentioned, the proof of Theorem \ref{thm:main_amp_perf} relies on two technical lemmas.  The first lemma, presented in this section, provides the conditional distribution of the vectors $h^{t+1}$ and $b^t$ given the matrices in \eqref{eq:XYMQt} as well as $\beta_0, w$.  Lemma \ref{lem:hb_cond} shows that these conditional distributions can be represented as the sum of a standard  Gaussian vector an a deviation term.  Then the second technical lemma, Lemma \ref{lem:main_lem}, shows that the deviation terms are small, meaning that their standardized norms concentrate on zero, and also provides concentration results for various inner products involving the other terms in recursion \eqref{eq:hqbm_def}, namely $\{ h^{t+1}, q^t, b ^t, m^t \}$.

The following notation is used for the concentration lemmas. Considering two random vectors $X, Y$ and a sigma-algebra $\mscrs$, we denote the fact that that conditional distribution of $X$ given $\mscrs$ equals the distribution of $Y$ as $X |_\mscrs \stackrel{d}{=} Y$.   We represent a $t \times t$ identity matrix as $\mathsf{I}_{t}$, dropping the $t$ subscript when it's obvious. For a matrix $A$ with full column rank, $\mathsf{P}^{\parallel}_{A} := A(A^*A)^{-1}A^*$ is the orthogonal projection matrix onto the column space of $A$, and $\mathsf{P}^\perp_{A}:=\mathsf{I}- \mathsf{P}^{\parallel}_{A}$.

Define $\mathscr{S}_{t_1, t_2}$ to be the sigma-algebra generated by the terms
\[ b^0, ..., b^{t_1 -1}, m^0, ..., m^{t_1 - 1}, h^1, ..., h^{t_2}, q^0, ..., q^{t_2},\text{ and }  \beta_0, w. \]

\begin{lem}\cite[Lemma 4]{RushV16}
For vectors $h^{t+1}$ and $b^t$ defined in \eqref{eq:hqbm_def}, the following conditional distributions hold for $t \geq 1$:
\begin{align}
h^{1} \lvert_{\mscrs_{1, 0}} \stackrel{d}{=} \tau_0 Z_0 + \Delta_{1,0}, \quad &\text{ and } \quad h^{t+1} \lvert_{\mscrs_{t+1, t}} \stackrel{d}{=} \sum_{r=0}^{t-1} \hat{\alpha}^t_r h^{r+1} + \tau^{\perp}_t Z_t + \Delta_{t+1,t}, \label{eq:Ha_dist} \\ 
b^{0} \lvert_{\mscrs_{0, 0}} \stackrel{d}{=} \sigma_0 Z'_0 + \Delta_{0,0}, \quad &\text{ and } \quad b^{t} \lvert_{\mscrs_{t, t}}\stackrel{d}{=} \sum_{r=0}^{t-1} \hat{\gamma}^{t}_r b^r + \sigma^{\perp}_t Z'_t + \Delta_{t,t}. \label{eq:Ba_dist}
\end{align}
where $Z_0, Z_t \in \mathbb{R}^N$ and $Z'_0, Z'_t \in \mathbb{R}^n$ are i.i.d.\ standard Gaussian random vectors that are independent of the corresponding conditioning sigma algebras. The terms $\hat{\gamma}^{t}_i$ and $\hat{\alpha}^t_{i}$ for $i = 0,...,t-1$ are defined in \eqref{eq:hatalph_hatgam_def} and the terms $(\tau_{t}^{\perp})^2$ and $(\sigma_{t}^{\perp})^2$ in \eqref{eq:sigperp_defs}.  The deviation terms are 
\begin{align}
\Delta_{0,0} &= \left(\frac{\norm{q^0}}{\sqrt{n}} - \sigma_0\right)Z'_0, \label{eq:D00} \\
\Delta_{1,0} &= \left[ \left(\frac{\norm{m^0}}{\sqrt{n}}  - \tau_0\right)\mathsf{I}_N -\frac{\norm{m^0}}{\sqrt{n}} \mathsf{P}^{\parallel}_{q^0}\right] Z_0 + q^0 \left(\frac{\norm{q^0}^2}{n}\right)^{-1} \left(\frac{(b^0)^*m_0}{n} - \xi_0 \frac{\norm{q^0}^2}{n}\right), \label{eq:D10}
\end{align}
%where $\mathsf{I}$ is the identity matrix and for any matrix $A$, $\mathsf{P}^{\parallel}_A$ is the orthogonal projection matrix onto the column space of $A$.  
and for $t >0$,
\begin{align}
\Delta_{t,t} = & \, \sum_{r=0}^{t-1} (\gamma^t_r - \hat{\gamma}^{t}_r) b^r + \left[  \left(\frac{\norm{q^t_{\perp}}}{\sqrt{n}} - \sigma_{t}^{\perp}\right) \mathsf{I}_n  - \frac{\norm{q^t_{\perp}} }{\sqrt{n}} \mathsf{P}^{\parallel}_{M_t}\right]Z'_t  \nonumber \\
& \quad + M_t\left(\frac{M_{t}^* M_{t}}{n}\right)^{-1} \left(\frac{H_t^* q^t_{\perp}}{n} - \frac{M_t}{n}^*\left[\lambda_t m^{t-1} - \sum_{r=1}^{t-1} \lambda_{r} \gamma^t_{r} m^{r-1}\right]\right),\label{eq:Dtt} 
\end{align}
\begin{align}
\Delta_{t+1,t} =& \,  \sum_{r=0}^{t-1} (\alpha^t_r - \hat{\alpha}^t_r) h^{r+1} + \left[\left(\frac{\norm{m^t_{\perp}}}{\sqrt{n}} - \tau_{t}^{\perp}\right)  \mathsf{I}_N  -\frac{\norm{m^t_{\perp}}}{\sqrt{n}} \mathsf{P}^{\parallel}_{Q_{t+1}}\right]Z_t \nonumber \\
&\quad + Q_{t+1} \left(\frac{Q_{t+1}^* Q_{t+1}}{n}\right)^{-1} \left(\frac{B^*_{t+1} m^t_{\perp}}{n} - \frac{Q_{t+1}^*}{n}\left[\xi_t q^t - \sum_{i=0}^{t-1} \xi_i \alpha^t_i q^i\right]\right).\label{eq:Dt1t}  
\end{align} 
\label{lem:hb_cond}
\end{lem}

\begin{proof}
The proof can be found in \cite{RushV16}.
\end{proof}

Lemma \ref{lem:hb_cond} holds only when $M^*_t M_t$ and $Q^*_{t_1} Q_{t_1}$  are invertible.

%%%%%%%%%%%%%%%%%%%%%%%%%%
%%%%%%%%%%%%%%%%%%%%%%%%%%

%The conditional distribution representation in Lemma \ref{lem:hb_cond} implies that for each $t \geq 0$,  $h^{t+1}$ is the sum of an i.i.d.\ $\mc{N}(0, \tau_t^2)$ random vector plus a deviation term. This is straightforward to verify for the  special case of the Bayes-optimal AMP recursion \eqref{eq:hqbm_def_AMP} with the de-noising function $\eta_t(\cdot)$ chosen as the conditional expectation of $\beta$ given the noisy observation $\beta+ \tau_t Z$, as in \eqref{eq:opt_cond_expec}.  Using \eqref{eq:simpl_exps} in Lemma \ref{lem:hb_cond}, we obtain
%\be
%h^{t+1} \lvert_{\mscrs_{t+1, t}} \stackrel{d}{=} (\tau_t^2/\tau_{t-1}^2) h^{t} + \tau^{\perp}_t Z_t + \Delta_{t+1,t}.
%\label{eq:bt_simp}
%\ee
%Assuming $h^{t}$ has representation  $\tau_{t-1} \tilde{Z}_{t-1} +  \Delta_{t}$, substituting in \eqref{eq:bt_simp} gives
%\ben
%\begin{split} 
%h^{t+1} & \stackrel{d}{=}   (\tau_t^2/\tau_{t-1}) \tilde{Z}_{t-1}   +  \tau^{\perp}_t Z_t  + \Delta_{t}  + \Delta_{t+1,t}  \stackrel{d}{=}   \tau_t \tilde{Z}_t + \Delta_{t}  + \Delta_{t+1,t} .  \end{split} \een
%To obtain the last equality above, we combine independent Gaussians $\tilde{Z}_{t-1}$ and $Z_t$ using the expression for $\tau_t^{\perp}$ in \eqref{eq:simpl_exps}.   Similarly $b^t$ is the sum of an i.i.d.\ $\mc{N}(0, \sigma_t^2)$ random vector and a deviation term.  The next lemma  shows that these deviation terms are small with high probability. 

% !TEX root = main.tex

\subsection{Main Concentration Lemma}

We use the shorthand $X_n \doteq c$ to denote the concentration inequality $P(\abs{X_n-c} \geq \epsilon) \leq K_{k,t} e^{-\kappa_{k,t} n \e^2}$.
As specified in the theorem statement, the lemma holds for all $\e\in(0,1)$, with $K_{k,t}, \kappa_{k,t}$ denoting generic constants depending on half window-size $k$ and iteration index $t$, but not on $n$ or $\e$.
%\be
%K_t = what?, \quad \text{ and } \quad \kappa_t = what? %K_t = K_1(K_2)^t(t!)^{10}, \quad \text{ and } \quad \kappa_t = \kappa_1[\kappa_2^t (t!)^{18}]^{-1},
%\label{eq:Kkappa_def}
%\ee 
%where $K_1, K_2, \kappa_1, \kappa_2 > 0$ are universal constants (not depending on $t$, $n$, or $\e$).  As specified in the theorem statement, the lemma holds for all $\e \in (0,1)$.  

\begin{lem}
With the $\doteq$ notation defined above, the following statements hold for $t\geq 0$.
%$0 \leq t < T^*$.

\begin{enumerate}[(a)]

\item %For all $0 \leq r \leq t$,
\begin{align}
P\left(\frac{1}{N}\norm{\Delta_{{t+1,t}}}^2 \geq \epsilon \right) \leq K_{k,t} e^{-\kappa_{k,t} n \e}, \label{eq:Ha} \\
P \left(\frac{1}{n}\norm{\Delta_{{t,t}}}^2 \geq \epsilon \right) \leq K_{k,t} e^{-\kappa_{k,t} n \e}. \label{eq:Ba} 
\end{align}
%\begin{align}
%P\left(\frac{1}{N}\norm{\Delta_{{t+1,t}}}^2 \geq \epsilon \right) \leq t^3 K_{t-1} \exp\left\{-\frac{\kappa_{t-1} n \epsilon}{t^5}\right\}, \label{eq:Ha} \\
%P \left(\frac{1}{n}\norm{\Delta_{{t,t}}}^2 \geq \epsilon \right) \leq t^3 K_{t-1} \exp\left\{-\frac{\kappa_{t-1} n \epsilon}{t^5}\right\}. \label{eq:Ba} 
%\end{align}

\item %\emph{i)}  
For pseudo-Lipschitz functions $\phi_h: \mathbb{R}^{(t+2)(2k+1)} \rightarrow \mathbb{R}$
\be
\frac{1}{N-2k}\sum_{i=k+1}^{N-k} \phi_h\left([h^1]_{i-k}^{i+k}, \ldots, [h^{t+1}]_{i-k}^{i+k}, [\beta_{0}]_{i-k}^{i+k}\right)  \doteq 
 \expec\bigg[\phi_h\left(\tau_0 \underline{\tilde{Z}}_0, \ldots, \tau_t \underline{\tilde{Z}}_t, \underline{\beta} \right)\bigg].  \label{eq:Hb1}
\ee
The random vectors $\underline{\tilde{Z}}_{0}, \ldots, \underline{\tilde{Z}}_t \in \mathbb{R}^{2k+1}$ are jointly Gaussian with zero mean entries which are independent of the other entries in the same vector with covariance across iterations given by \eqref{eq:tildeZcov}, and are independent of $\underline{\beta} \sim \pi$.

%\emph{ii)} Let $\psi_h: \mathbb{R}^{2(2k+1)} \rightarrow \mathbb{R}$ be a bounded function that is differentiable in the first argument with bounded derivative, meaning that the first $2k+1$ partial derivatives exist and are bounded.  Then,
%\be
%\frac{1}{N-2k} \sum_{i=k+1}^{N-k} \psi_h([h^{t+1}]_{i-k}^{i+k}, [\beta_{0}]_{i-k}^{i+k}) \doteq \mathbb{E}\left[ \psi_h( \tau_{t} \underline{\tilde{Z}}_{t}, \underline{\beta})\right].   \label{eq:Hb2}
% \ee
% As above, $\underline{\tilde{Z}}_t \in \mathbb{R}^{2k+1}$ has entries that are i.i.d.\ $\sim \mc{N}(0,1)$ and $\underline{\beta} \in \mathbb{R}^{2k+1} \sim \pi$ are independent.  
 %The case where the derivative of $\psi_h$ has a finite number of discontinuities is handled in the Supplementary Material.

%\emph{ii)} 
For pseudo-Lipschitz functions $\phi_b: \mathbb{R}^{t+2} \rightarrow \mathbb{R}$
\be
\frac{1}{n}\sum_{i=1}^n \phi_b\left(b^0_i, \ldots, b^{t}_i, w_{i}\right)  \doteq 
 \expec\bigg[\phi_b\left(\sigma_0 \breve{Z}_0, \ldots, \sigma_t \breve{Z}_t, W \right)\bigg].  \label{eq:Bb1}
\ee
The random variables $\breve{Z}_{0}, \ldots, \breve{Z}_t$ are jointly Gaussian with zero mean and covariance given by \eqref{eq:tildeZcov}, and are independent of $W \sim p_{w}$.

%\emph{iv)} Let $\psi_b: \mathbb{R} \rightarrow \mathbb{R}$ be a bounded function that is differentiable in the first argument with bounded derivative.  Then,
%%except possibly at a finite number of points, with bounded derivative where it exists.  Then,
%\be
%\frac{1}{n} \sum_{i=1}^n \psi_b(b^{t}_i, w_{i}) \doteq  \mathbb{E}\left[ \psi_b(\sigma_{t} \breve{Z}_{t}, W)\right].   \label{eq:Bb2}
% \ee
%  As above, $\breve{Z}_t \sim \mc{N}(0,1)$ and $W \sim p_w$ are independent.  
  %The case where the derivative of $\psi_b$ has a finite number of discontinuities is handled in the Supplementary Material.

%%%%

\item
\begin{align}
\frac{(h^{t+1})^* q^0}{n} &\doteq 0, \quad \frac{(h^{t+1})^* \beta_0}{n} \doteq 0, \label{eq:Hc} \\
\frac{(b^t)^* w}{n} &\doteq 0.  \label{eq:Bc}
\end{align}

%%%%%%%%

\item For all $0 \leq r \leq t$, 
\begin{align}
\frac{(h^{r+1})^* h^{t+1}}{N} &\doteq \breve{E}_{r,t}, \label{eq:Hd} \\
\frac{(b^r)^*b^t}{n} &\doteq \tilde{E}_{r,t}. \label{eq:Bd}
\end{align}

%%%%%

\item For all $0 \leq r \leq t$,
\begin{align}
\frac{(q^{0})^* q^{t+1}}{n} &\doteq \tilde{E}_{0,t+1}, \quad \frac{(q^{r+1})^*q^{t+1}}{n} \doteq \tilde{E}_{r+1,t+1},  \label{eq:He}  \\
\frac{(m^r)^* m^t}{n} &\doteq \breve{E}_{r,t} \label{eq:Be}
\end{align}

%%%%%%%
\item For all $0 \leq r \leq t$,  
\begin{align}
\lambda_t \doteq \hat{\lambda}_{t},  \quad  \frac{(h^{t+1})^*q^{r+1}}{n} &\doteq  \hat{\lambda}_{r+1} \breve{E}_{r,t},  \quad  \frac{(h^{r+1})^*q^{t+1}}{n} \doteq  \hat{\lambda}_{t+1} \breve{E}_{r,t}, \label{eq:Hf} \\
\xi_t \doteq \hat{\xi}_{t},  \quad \frac{(b^r)^*m^t}{n} &\doteq \hat{\xi}_{t}\tilde{E}_{r,t} ,  \quad \frac{(b^t)^*m^r}{n} \doteq \hat{\xi}_{r} \tilde{E}_{r,t}.\label{eq:Bf}
\end{align}

%%%%%%%%
\item Let $\textbf{Q}_{t+1} = \frac{1}{n} Q_{t+1}^* Q_{t+1}$ and $\textbf{M}_{t} = \frac{1}{n} M_{t}^* M_{t}$.  Then,
\begin{align}
P\left(\textbf{Q}_{t+1} \text{ is singular}\right) \leq K_{k,t} e^{-\kappa_{k,t} n}, \label{eq:Qsing} \\
P\left(\textbf{M}_{t} \text{ is singular}\right) \leq K_{k,t} e^{-\kappa_{k,t} n}.  \label{eq:Msing} 
\end{align}
When the inverses of $\textbf{Q}_{t+1}$ and $\textbf{M}_{t} $ exist, for all $0 \leq  i,j  \leq t$ and $0 \leq  i',j' \leq t-1$:
% at the following rates:
\begin{align}
\left[\mathbf{Q}_{t+1}^{-1} \right]_{i+1,j+1} &\doteq  [(\tilde{C}^{t+1})^{-1}]_{i+1,j+1}, \quad \gamma^{t+1}_{i} \doteq \hat{\gamma}^{t+1}_{i}, \label{eq:Hg}\\ %\quad  \  0 \leq k \leq t,  
\left[\mathbf{M}_t^{-1} \right]_{i'+1,j'+1} &\doteq  [(\breve{C}^t)^{-1}]_{i'+1,j'+1}, \quad \alpha^{t}_{k'} \doteq \hat{\alpha}^{t}_{i'}, \quad   t \geq 1,\label{eq:Bg}  % 0 \leq k' \leq t-1,  
\end{align}
where $\hat{\gamma}^{t+1}_{k}$ and $\hat{\alpha}^{t}_{k'}$ are defined in \eqref{eq:hatalph_hatgam_def},

%%%%
\item With $\sigma_{t+1}^{\perp}, \tau_{t}^{\perp}$ defined in \eqref{eq:sigperp_defs},
\begin{align} 
\frac{1}{n}\norm{q^{t+1}_{\perp}}^2 &\doteq (\sigma_{t+1}^{\perp})^2, \label{eq:Hh} \\
\frac{1}{n}\norm{m^t_{\perp}}^2 &\doteq (\tau_{t}^{\perp})^2. \label{eq:Bh}
\end{align}

\end{enumerate}
\label{lem:main_lem}
\end{lem}

%\textbf{Remarks}:

\subsection{Proof of Theorem \ref{thm:main_amp_perf}}

\begin{proof}
Applying Part (b)(i) of Lemma \ref{lem:main_lem} to a pseudo-Lipschitz (PL) function $\phi_h: \mathbb{R}^{2(2k +1)} \to \mathbb{R}$,
\ben
P\left(\left\lvert \frac{1}{N-2k} \sum_{i=k+1}^{N-k} \phi_h([h^{t+1}]_{i-k}^{i+k}, [\beta_{0}]_{i-k}^{i+k}) - \mathbb{E}\left[\phi_h(\tau_t \underline{Z}, \underline{\beta})\right] \right \lvert \geq \e\right) \leq K_{k,t} e^{-\kappa_{k,t} n \e^2} 
\een
where the random vectors $\underline{\beta} \in S^{2k+1} \sim \pi$ and $\underline{Z} \in \mathbb{R}^{2k+1} $, whose entries are i.i.d.\ standard normal random variables, are independent. Now for $i=k+1,...,N-k$ let
\be
 \phi_h([h^{t+1}]_{i-k}^{i+k}, [\beta_{0}]_{i-k}^{i+k}) := \phi(\eta_{t}([\beta_{0} - h^{t+1}]_{i-k}^{i+k}), \beta_{0_i}),
 \label{eq:phih_map}
\ee
where $\phi: \mathbb{R}^{2} \to \mathbb{R}$ is the PL function in the statement of the theorem. The function  $\phi_h([h^{t+1}]_{i-k}^{i+k}, [\beta_{0}]_{i-k}^{i+k})$ in \eqref{eq:phih_map} is PL since $\phi$ is PL and $\eta_t$ is Lipschitz. We therefore obtain
\be
\begin{split}
P\left( \left\lvert \frac{1}{N-2k} \sum_{i=k+1}^{N-k} \phi(\eta_{t}([\beta_{0} - h^{t+1}]_{i-k}^{i+k}), \beta_{0_i})  - \mathbb{E}\left[ \phi(\eta_{t}(\underline{\beta} - \tau_t \underline{Z}), \underline{\beta}_{k+1})\right] \right \lvert \geq \e \right) \leq K_{k,t} e^{-\kappa_{k,t} n \e^2}. \nonumber
\end{split}
\ee
The proof is completed by noting from \eqref{eq:amp2} and \eqref{eq:hqbm_def_AMP} that 
$\beta^{t+1}_i = \eta_{t}([A^* z^t + \beta^t]_{i-k}^{i+k}) = \eta_{t}([\beta_0 - h^{t+1}]_{i-k}^{i+k})$.
\end{proof}

% !TEX root = main.tex

\section{Proof of Lemma \ref{lem:main_lem}} \label{sec:main_lem_proof}

%We first list some results that will be used in the proof.

\subsection{Mathematical Preliminaries} \label{subsec:mathpre}
%Some of the results written below can be found in \cite[Section III.G]{BayMont11}, but we summarize them here for completeness.

\begin{fact}\cite[Fact 1]{RushV16}
 Let  $u \in \mathbb{R}^N$ and $v \in \mathbb{R}^n$ be deterministic vectors, and let $\tilde{A} \in \mathbb{R}^{n \times N}$ be a matrix with independent  $\mc{N}(0, 1/n)$ entries. Then:
 
(a)
\begin{equation*}
\tilde{A} u \overset{d}{=} \frac{\norm{u}}{\sqrt{n}} Z_u  \quad \text{ and }  \quad \tilde{A}^*v \overset{d}{=} \frac{\norm{v}}{\sqrt{n}} Z_v,
%\label{eq:Au_dist}
\end{equation*}
where $Z_u \in \mathbb{R}^n$ and $Z_v \in \mathbb{R}^N$ are i.i.d.\ standard Gaussian random vectors.

%Let $\tilde{A} \in \mathbb{R}^{n \times N}$ be a matrix with independent  $\mc{N}(0, 1/n)$ entries. 
(b)  Let $\mc{W}$  be a $d$-dimensional subspace of $\mathbb{R}^n$ for $d \leq n$. Let $(w_1, ..., w_d)$ be an orthogonal basis of $\mc{W}$ with 
$\norm{w_\ell}^2 = n$ for
$\ell \in [d]$, and let  $\mathsf{P}^{\parallel}_\mc{W}$ denote  the orthogonal projection operator onto $\mc{W}$.  Then for $D = [w_1\mid \ldots \mid w_d]$, we have 
$\mathsf{P}^{\parallel}_{\mc{W}} \tilde{A} u \overset{d}{=}   \frac{\norm{u}}{\sqrt{n}} \mathsf{P}^{\parallel}_{\mc{W}} Z_u\overset{d}{=}  \frac{\norm{u}}{\sqrt{n}} Dx$ where $x \in \mathbb{R}^d$ is a random vector with i.i.d.\ $\mc{N}(0, 1/n)$ entries. 
\label{fact:gauss_p0}
\end{fact}

\begin{fact}[Stein's lemma]
For zero-mean jointly Gaussian random variables $Z_1, Z_2$, and any function $f:\mathbb{R} \to \mathbb{R}$ for which $\expec[Z_1 f(Z_2)]$ and $\expec[f'(Z_2)]$  both exist, we have $\expec[Z_1 f(Z_2)] = \expec[Z_1Z_2] \expec[f'(Z_2)]$.
\label{fact:stein}
\end{fact}

We also make use of concentration results that are listed in Appendices \ref{app:conc_lemma}, \ref{app:gauss_lemmas}, and \ref{app:other}.  Many of these results and their proofs can be found in Rush and Venkataramanan \cite{RushV16}.  Appendix \ref{app:conc_dependent} holds concentration results for dependent random variables that were needed to provide the new results in this paper, 
%many of which we believe to be interesting results on their own, 
such as concentration for psuedo-Lipschitz functions acting on Markovian input.

The proof of  Lemma \ref{lem:main_lem}. proceeds by induction on $t$.  We label as $\mathcal{H}_{t+1}$ the results \eqref{eq:Ha}, \eqref{eq:Hb1}, %\eqref{eq:Hb2},
\eqref{eq:Hc}, \eqref{eq:Hd}, \eqref{eq:He}, \eqref{eq:Hf}, \eqref{eq:Qsing}, \eqref{eq:Hg}, \eqref{eq:Hh} and similarly as $\mathcal{B}_t$ the results \eqref{eq:Ba}, \eqref{eq:Bb1}, 
%\eqref{eq:Bb2}, 
\eqref{eq:Bc}, \eqref{eq:Bd}, \eqref{eq:Be}, \eqref{eq:Bf}, \eqref{eq:Msing}, \eqref{eq:Bg}, \eqref{eq:Bh}.  The proof consists of four steps: (1) $\mathcal{B}_0$ holds; (2) $\mathcal{H}_1$ holds; (3) if $\mathcal{B}_r, \mathcal{H}_s$ holds for all $r < t $ and $s \leq t $, then $\mathcal{B}_t$ holds; and (4) if $\mathcal{B}_r, \mathcal{H}_s$ holds for all $r \leq t $ and $s \leq t$, then $\mathcal{H}_{t+1}$ holds.

For each step, in parts $(a)$--$(h)$ of the proof, we use $K$ and $\kappa$ to label universal constants, meaning they do not depend on $n$ or $\e$, but may depend on $t$, in the concentration upper bounds.  
%At the end of each step, these constants are multiplied to obtain $K_1, K_2, \kappa_1, \kappa_2$ of \eqref{eq:Kkappa_def}.

%---------------------------------------------------------------------------------------------------------------------------------------------------------
%-----------------------------------------------------------------------------------------
\subsection{Step 1: Showing $\mc{B}_0$ holds}

We wish to show results (a)--(h) in \eqref{eq:Ba}, \eqref{eq:Bb1}, 
%\eqref{eq:Bb2}, 
\eqref{eq:Bc}, \eqref{eq:Bd}, \eqref{eq:Be}, \eqref{eq:Bf}, \eqref{eq:Msing}, \eqref{eq:Bg}, \eqref{eq:Bh} for $t=0$.  The proof of these results is the same as in the step $\mc{B}_0$ of the proof in \cite{RushV16} and therefore is not repeated here.

%---------------------------------------------------------------------------------------------------------------------------------------------------------
%---------------------------------------------------------------------------------------------------------------------------------------------------------

\subsection{Step 2: Showing $\mc{H}_1$ holds}

We wish to show results (a)--(h) in \eqref{eq:Ha}, \eqref{eq:Hb1}, 
%\eqref{eq:Hb2}, 
\eqref{eq:Hc}, \eqref{eq:Hd}, \eqref{eq:He}, \eqref{eq:Hf}, \eqref{eq:Qsing}, \eqref{eq:Hg}, \eqref{eq:Hh} for $t=0$.

%---------------------------------------------------------------------------------------------------------------------------------------------------------

\textbf{(a)} The proof of $\mc{H}_1 (a)$ follows as the corresponding proof in \cite{RushV16}.

\textbf{(b)(i)} 
For $t=0$, the LHS of \eqref{eq:Hb1} can be bounded as
\begin{equation}
\begin{split}
P&\left(\left \lvert \frac{1}{N-2k} \sum_{i=k+1}^{N-k} \phi_h ([h^{1}]_{i-k}^{i+k}, [\beta_{0}]_{i-k}^{i+k}) - \mathbb{E}[\phi_h(\tau_0 \underline{\tilde{Z}}_0, \underline{\beta})] \right \lvert \geq \epsilon \right) \\
&\overset{(a)}{=} P\left(\left \lvert \frac{1}{N-2k} \sum_{i=k+1}^{N-k} \phi_h([\tau_0 Z_{0} + \Delta_{1,0}]_{i-k}^{i+k}, [\beta_{0}]_{i-k}^{i+k})  - \mathbb{E}[\phi_h(\tau_0 \underline{\tilde{Z}}_0, \underline{\beta})] \right \lvert \geq \epsilon \right)  \\
&\overset{(b)}{\leq} P\left(\left \lvert \frac{1}{N-2k} \sum_{i=k+1}^{N-k}\mathbb{E}_{\underline{\tilde{Z}}_0}[\phi_h(\tau_0 \underline{\tilde{Z}}_0,  [\beta_{0}]_{i-k}^{i+k})] - \mathbb{E}_{\underline{\tilde{Z}}_0, \underline{\beta}}[\phi_h(\tau_0 \underline{\tilde{Z}}_0, \underline{\beta})] \right \lvert \geq \frac{\epsilon}{3} \right) \\
& \quad + P\left(\left \lvert \frac{1}{N-2k} \sum_{i=k+1}^{N-k} \left[\phi_h(\tau_0 [Z_{0}]_{i-k}^{i+k},  [\beta_{0}]_{i-k}^{i+k}) - \mathbb{E}_{\underline{\tilde{Z}}_0}[\phi_h(\tau_0 \underline{\tilde{Z}}_0,  [\beta_{0}]_{i-k}^{i+k})]\right] \right \lvert \geq \frac{\epsilon}{3} \right) \\
&\quad + P\left(\left \lvert \frac{1}{N-2k} \sum_{i=k+1}^{N-k} \left[\phi_h([\tau_0 Z_{0} + \Delta_{1,0}]_{i-k}^{i+k},  [\beta_{0}]_{i-k}^{i+k}) - \phi_h(\tau_0 [Z_{0}]_{i-k}^{i+k}, [\beta_{0}]_{i-k}^{i+k}) \right] \right \lvert \geq \frac{\epsilon}{3} \right).
\end{split}
\label{eq:phih_0}
\end{equation}
Step $(a)$ follows from the conditional distribution of $h^1$ given in Lemma \ref{lem:hb_cond} \eqref{eq:Ha_dist} and step $(b)$ Lemma \ref{sums}.  
Label the terms on the RHS of \eqref{eq:phih_0} as $T_1 -T_3$.   We show that each of these terms  is bounded above by $K e^{-\kappa N \e^2}.$ Term $T_1$ is upper bounded by $K e^{-\kappa (N-2k) \e^2}$ using Lemma \ref{lem:PLMCconc_new} since the function $\tilde{\phi}_h: \mathbb{R}^{2k+1} \rightarrow \mathbb{R}$ defined as $\tilde{\phi}_h(\underline{s}) := \mathbb{E}_{\underline{Z}}[\phi_h(\tau_0 \underline{Z}, \underline{s})]$ is PL(2) by Lemma \ref{lem:PLexamples}.  Term $T_2$ is upper bounded by $K e^{-\kappa (N-2k) \e^2}$ using Lemma \ref{lem:PL_overlap_gauss_conc} since the function $\hat{\phi}_{h,i}:\mathbb{R}^{2k+1}\rightarrow\mathbb{R}$ defined as 
\be
\hat{\phi}_{h,i}(\underline{s}):=\phi_h(\underline{s},[\beta]_{i-k}^{i+k}) \in PL(2), \quad \text{ for } k+1 \leq i \leq N-k,
\label{eq:hat_phi_def}
\ee
%by Lemma \ref{lem:PLexamples} \textcolor{blue}{I don't think we used \ref{lem:PLexamples} here. I think that $\hat{\phi}_{h,i}$ is PL(2) directly follows from the fact that $\phi_h$ is PL(2).}, 
where we have used the fact that $\mathbb{E}_{[Z_{0}]_{i-k}^{i+k}}[\phi_h(\tau_0 [Z_{0}]_{i-k}^{i+k},  [\beta_{0}]_{i-k}^{i+k})] = \mathbb{E}_{\underline{\tilde{Z}}_0}[\phi_h(\tau_0 \underline{\tilde{Z}}_0,  [\beta_{0}]_{i-k}^{i+k})]$ for each $k+1 \leq i \leq N-k$. Finally consider $T_3$, the third term on the RHS of \eqref{eq:phih_0}.
\begin{align}
T_3 &\overset{(a)}{\leq} P\left(\frac{1}{N-2k} \sum_{i=k+1}^{N-k} L \left(1 + \norm{[\tau_0 Z_{0} + \Delta_{1,0}]_{i-k}^{i+k}} + \norm{\tau_0 [Z_{0}]_{i-k}^{i+k}}\right) \norm{[\Delta_{1,0}]_{i-k}^{i+k}} \geq \frac{\epsilon}{3} \right) \nonumber \\
&\overset{(b)}{\leq} P\left( \frac{\norm{\Delta_{1,0}}}{\sqrt{N-2k}}  \cdot \left(1 + \sqrt{2k+1}\frac{\norm{\Delta_{1,0}}}{\sqrt{N-2k}} + 2\tau_0 \sqrt{2k+1} \frac{\norm{Z_{0}}}{\sqrt{N-2k}}\right)\geq \frac{\epsilon}{3L\sqrt{3(2k+1)}} \right).  \label{eq:B1func1eq1}
\end{align}
Step $(a)$ follows from the fact that $\phi_h$ is PL(2).
%the triangle inequality and \eqref{eq:hat_phi_def} with the PL constant denoted by $L$ 
Step $(b)$ uses $\norm{[\tau_0 Z_{0} + \Delta_{1,0}]_{i-k}^{i+k}} \leq \norm{\tau_0 [Z_{0}]_{i-k}^{i+k}} + \norm{[\Delta_{1,0}]_{i-k}^{i+k}}$ by the triangle inequality, the Cauchy-Schwarz inequality, the fact that for $a \in \mathbb{R}^N$, $\sum_{i=k+1}^{N-k} \norm{[a]_{i-k}^{i+k}}^2 \leq (2k+1) \norm{a}^2$, and the following application of Lemma \ref{lem:squaredsums}:
\ben
\begin{split}
\sum_{i=k+1}^{N-k} \left(1 + \norm{[\Delta_{1,0}]_{i-k}^{i+k}} + 2\norm{\tau_0 [Z_{0}]_{i-k}^{i+k}}\right)^2 &\leq 3\left((N-2k) + (2k+1)\norm{\Delta_{1,0}}^2 + 4 \tau_0^2 (2k+1) \norm{Z_{0}}^2\right).
\end{split}
\een
From \eqref{eq:B1func1eq1},  we have
\ben
T_3 \leq  P\left( \frac{\norm{Z_{0}}}{\sqrt{N-2k}} \geq 2 \right)  + P\left( \frac{\norm{\Delta_{1,0}}}{\sqrt{N-2k}} \geq \frac{\frac{\e}{\sqrt{2k+1}}  \min\left\{1, \frac{1}{6L \sqrt{3}} \right\}}{2 + 4 \tau_0 \sqrt{2k+1}} \right)  \overset{(a)}{\leq} e^{-(N-2k)} + K e^{-\kappa (N-2k) \e^2},
\een
where to obtain $(a)$, we use Lemma \ref{subexp} and  $\mc{H}_1 (a)$.

\textbf{(c)} We first show concentration for $(h^1)^*\beta_0/n$.  This result follows directly from $\mc{H}_1(b)$: we can write $\abs{(h^1)^* \beta_0} = \left \lvert \sum_{i=1}^N h^1_i \beta_{0_i} \right \lvert \leq \left \lvert \sum_{i=1}^{N/2} h^1_i \beta_{0_i} \right \lvert + \left \lvert \sum_{j=N/2+1}^N h^1_j \beta_{0_j} \right \lvert $ and it follows by Lemma \ref{sums},
\begin{align*}
P\left(\left \lvert \frac{(h^1)^* \beta_0}{n}\right \lvert  \geq \e \right) &\leq P\left(\left \lvert \sum_{i=1}^{N/2} \frac{h^1_i \beta_{0_i}}{N/2}\right \lvert  \geq \e \delta \right) + P\left(\left \lvert \sum_{j=N/2+1}^N \frac{h^1_j \beta_{0_j}}{N/2}\right \lvert  \geq  \e \delta \right) \overset{(a)}{\leq} K e^{-\frac{\kappa N \e^2 \delta^2}{4}} + K e^{-\frac{\kappa N \e^2 \delta^2}{4}}.
\end{align*}
In the above, step $(a)$ follows by applying $\mc{H}_1(b)$ using PL(2) functions $\phi_{1,h}, \phi_{2,h}$ both defined from $\mathbb{R}^{2(2k+1)} \rightarrow \mathbb{R}$ equal to $\phi_{1,h}(\underline{x}, \underline{y}) = \underline{x}_1 \underline{y}_1$, and $\phi_{2,h}(\underline{x}, \underline{y}) = \underline{x}_{2k+1} \underline{y}_{2k+1}.$  Note that $\mathbb{E}[ \tau_0 \tilde{\underline{Z}}_{0_1} \underline{\beta}_1] = 0.$

Next we show concentration for $(h^1)^*q^0/n$.  Note that 
\[(h^1)^*q^0 = \sum_{i=1}^N h^1_i q^0_i = \sum_{i=k+1}^{N-k} h^1_i f_0(\underline{0}, [\beta_0]_{i-k}^{i+k}) - \sum_{i=1}^{k} h^1_i \beta_{0_i} - \sum_{i=N-k+1}^{N} h^1_i \beta_{0_i}\]
where the last equality follows by definition of $q^0$ provided in \eqref{eq:q0def}.  It follows by Lemma \ref{sums},
\begin{align*}
&P\left(\left \lvert \frac{(h^1)^* q^0}{n}\right \lvert  \geq \e \right) \\
&\leq P\left(\left \lvert \sum_{i=k+1}^{N-k} \frac{h^1_i f_0(\underline{0}, [\beta_0]_{i-k}^{i+k})}{N-2k} \right \lvert  \geq \frac{\e n}{3(N-2k)} \right) + P\left(\left \lvert \sum_{i=1}^{k} \frac{h^1_i \beta_{0_i}}{k} \right \lvert  \geq \frac{\e n}{3k} \right) + P\left(\left \lvert \sum_{i=N-k+1}^{N} \frac{h^1_i \beta_{0_i}}{k} \right \lvert  \geq \frac{\e n}{3k} \right) \\
&\overset{(a)}{\leq} K\exp\left\{-\frac{\kappa n^2\e^2}{9(N-2k)}\right\} + K\exp\left\{-\frac{\kappa n^2 \e^2}{9k}\right\} + K\exp\left\{-\frac{\kappa n^2 \e^2}{9k}\right\}.
\end{align*}
In the above, step $(a)$ follows from $\mc{H}_1(b)$ using PL(2) functions $\phi_{1,h}, \phi_{2,h}, \phi_{3,h}: \mathbb{R}^{2(2k+1)} \rightarrow \mathbb{R}$ equal to $\phi_{1, h}(\underline{x}, \underline{y}) = \underline{x}_{k+1} f_0(\underline{0}, \underline{y})$, $\phi_{2, h}(\underline{x}, \underline{y}) = \underline{x}_1 \underline{y}_1$, $\phi_{3, h}(\underline{x}, \underline{y}) = \underline{x}_{2k+1} \underline{y}_{2k+1}$ which are all PL(2) since products of Lipschitz functions are PL(2) by Lemma \ref{lem:Lprods}.  Note that $\mathbb{E}[\tau_0 \underline{\tilde{Z}}_{0, k+1} f(\underline{0}, \underline{\beta})] = 0$ and also that $n^2/(N-2k) \leq \kappa n$ where $\kappa$ depends on $k$ and $\delta$.

\textbf{(d)}  The result follows as in $\mc{H}_1(c)$.  
%We can write $\norm{h^1}^2 =  \lvert \sum_{i=1}^N (h^1_i)^2  \lvert \leq  \lvert \sum_{i=1}^{N/2} (h^1_i)^2  \lvert + \lvert \sum_{j=N/2+1}^N (h^1_j)^2 \lvert $ and therefore it follows by Lemma \ref{sums},
%\begin{align*}
%P\left(\left \lvert \frac{\norm{h^1}^2}{N} - \tau_0^2 \right \lvert  \geq \e \right) &\leq P\left(\left \lvert \sum_{i=1}^{N/2} \frac{(h^1_i)^2}{N/2} - \tau_0^2 \right \lvert  \geq \e \right) + P\left(\left \lvert \sum_{j=N/2+1}^N \frac{(h^1_j)^2}{N/2} - \tau_0^2\right \lvert  \geq \e \right) \\
%&\overset{(a)}{\leq} K e^{-\frac{\kappa N \e^2}{4}} + K e^{-\frac{\kappa N \e^2}{4}}.
%\end{align*}
We can write $\norm{h^1}^2 =  \sum_{i=1}^N (h^1_i)^2  = \sum_{i=1}^{N/2} (h^1_i)^2  + \sum_{j=N/2+1}^N (h^1_j)^2$ and therefore it follows by Lemma \ref{sums},
\begin{align*}
&P\left(\left \lvert \frac{\norm{h^1}^2}{N} - \tau_0^2 \right \lvert  \geq \e \right) = P\left(\left \lvert  \sum_{i=1}^{N/2}\frac{ (h^1_i)^2 }{N/2} +\sum_{j=N/2+1}^N \frac{ (h^1_j)^2}{N/2}  - 2\tau_0^2 \right \lvert  \geq 2\e \right)  \\
&\leq P\left(\left \lvert \sum_{i=1}^{N/2} \frac{(h^1_i)^2}{N/2} - \tau_0^2 \right \lvert  \geq \e \right) + P\left(\left \lvert \sum_{j=N/2+1}^N \frac{(h^1_j)^2}{N/2} - \tau_0^2\right \lvert  \geq \e \right) \overset{(a)}{\leq} K \exp\left\{-\frac{\kappa N\e^2}{2}\right\} + K \exp\left\{-\frac{\kappa N \e^2}{2}\right\}.
\end{align*}
In the above, step $(a)$ follows by applying $\mc{H}_1(b)$ using PL(2) functions $\phi_{1,h}, \phi_{2,h}$ both defined from $\mathbb{R}^{2(2+1)} \rightarrow \mathbb{R}$ equal to $\phi_{1,h}(\underline{x}, \underline{y}) = (\underline{x}_1)^2$, and $\phi_{2,h}(\underline{x}, \underline{y}) = (\underline{x}_{2k+1})^2.$

\textbf{(e)} We prove concentration for $(q^0)^*q^{1}$ first.  Notice that
\[(q^0)^*q^{1} = \sum_{i=1}^N q^0_i q^{1}_i = \sum_{i=k+1}^{N-k} f_0(\underline{0}, [\beta_{0}]_{i-k}^{i+k}) f_1([h^1]_{i-k}^{i+k}, [\beta_{0}]_{i-k}^{i+k})  + \sum_{i=1}^k \beta_{0_i}^2  + \sum_{i=N-k+1}^N \beta_{0_i}^2.\]
Therefore it follows by Lemma \ref{sums},
\begin{align*}
&P\left(\left \lvert \frac{(q^0)^* q^1}{n} - \tilde{E}_{0,1}\right \lvert  \geq \e \right) \\
&\leq P\left(\left \lvert \sum_{i=k+1}^{N-k} \frac{f_0(\underline{0}, [\beta_{0}]_{i-k}^{i+k}) f_1([h^1]_{i-k}^{i+k}, [\beta_{0}]_{i-k}^{i+k})}{N-2k} - \mathbb{E}[f_0(\underline{0}, \underline{\beta}) f_1(\tau_0 \underline{\tilde{Z}}_0, \underline{\beta})] \right \lvert  \geq \frac{\e n}{3(N-2k)} \right) \\
&\quad + P\left(\left \lvert \sum_{i=1}^{k} \frac{\beta_{0_i}^2}{k} -\sigma_{\beta}^2 \right \lvert  \geq \frac{\e n}{3k} \right) + P\left(\left \lvert \sum_{i=N-k+1}^{N} \frac{\beta_{0_i}^2}{k} - \sigma_{\beta}^2\right \lvert  \geq \frac{\e n}{3k} \right) \\
&\overset{(a)}{\leq} K\exp\left\{-\frac{\kappa n^2 \e^2}{9(N-2k)}\right\} + K\exp\left\{-\frac{\kappa n^2 \e^2}{9k}\right\} + K\exp\left\{-\frac{\kappa n^2 \e^2}{9k}\right\}.
\end{align*}
In the above, step $(a)$ follows from $\mc{H}_1(b)$ using PL(2) functions $\phi_{1,h}, \phi_{2,h}, \phi_{3,h}: \mathbb{R}^{2(2k+1)} \rightarrow \mathbb{R}$ equal to $\phi_{1, h}(\underline{x}, \underline{y}) = f_0(\underline{0}, \underline{y}) f_1(\underline{x}, \underline{y})$, $\phi_{2, h}(\underline{x}, \underline{y}) = \underline{y}_1^2$, $\phi_{3, h}(\underline{x}, \underline{y}) =\underline{y}_{2k+1}^2$ which are all PL(2) since products of Lipschitz functions are PL(2) by Lemma \ref{lem:Lprods}.  The result follows by noting $\mathbb{E}[\underline{\beta}_1^2] = \mathbb{E}[\underline{\beta}_{2k+1}^2]= \sigma_{\beta}^2.$

Concentration for $\norm{q^{1}}^2$ follows similarly by applying $\mc{H}_1(b)$ with the representation
\[\norm{q^{1}}^2 = \sum_{i=1}^N (q^{1}_i)^2 = \sum_{i=k+1}^{N-k} (f_1([h^1]_{i-k}^{i+k}, [\beta_{0}]_{i-k}^{i+k}))^2  + \sum_{i=1}^k \beta_{0_i}^2  + \sum_{i=N-k+1}^N \beta_{0_i}^2.\]

\textbf{(f)} The concentration of $\lambda_0$ around $\hat{\lambda}_0$ follows from $\mathcal{H}_1 (b)(i)$ applied to the function $\phi_h([h^1]_{i-k}^{i+k}, [\beta_0]_{i-k}^{i+k}) := f_0'([h^1]_{i-k}^{i+k}, [\beta_0]_{i-k}^{i+k})$.  The only other result to prove is concentration for $(h^{1})^*q^{1}$.  Notice that
\[(h^1)^*q^{1} = \sum_{i=1}^N h^1_i q^{1}_i = \sum_{i=k+1}^{N-k} h^1_i f_1([h^1]_{i-k}^{i+k}, [\beta_{0}]_{i-k}^{i+k})  + \sum_{i=1}^k h^1_i \beta_{0_i}  + \sum_{i=N-k+1}^N h^1_i \beta_{0_i}.\]
Therefore it follows by Lemma \ref{sums},
\begin{align*}
P\left(\left \lvert \frac{(h^1)^* q^1}{n} - \hat{\lambda}_{1} \breve{E}_{0,0} \right \lvert \geq \e \right) &\leq P\left(\left \lvert \sum_{i=k+1}^{N-k} \frac{h^1_i f_1([h^1]_{i-k}^{i+k}, [\beta_{0}]_{i-k}^{i+k})}{N-2k} - \frac{n \hat{\lambda}_{1} \breve{E}_{0,0}}{N-2k} \right \lvert  \geq \frac{\e n}{3(N-2k)} \right) \\
&\qquad + P\left(\left \lvert \sum_{i=1}^{k} \frac{h^1_i \beta_{0_i}}{k}  \right \lvert  \geq \frac{\e n}{3k} \right)  + P\left(\left \lvert \sum_{i=k+1}^{N-k} \frac{h^1_i \beta_{0_i} }{k} \right \lvert  \geq \frac{\e n}{3k} \right) \\
&\overset{(a)}{\leq}  K\exp\left\{-\frac{\kappa n^2 \e^2}{9(N-2k)}\right\} + K\exp\left\{-\frac{\kappa n^2 \e^2}{9k}\right\} + K\exp\left\{-\frac{\kappa n^2 \e^2}{9k}\right\},
\end{align*}
In the above, step $(a)$ follows from $\mc{H}_1(b)$ using PL(2) functions $\phi_{1,h}, \phi_{2,h}, \phi_{3,h}: \mathbb{R}^{2(2k+1)} \rightarrow \mathbb{R}$ equal to $\phi_{1, h}(\underline{x}, \underline{y}) = \underline{x}_{k+1} f_1(\underline{x}, \underline{y})$, $\phi_{2, h}(\underline{x}, \underline{y}) = \underline{x}_1 \underline{y}_1$, $\phi_{3, h}(\underline{x}, \underline{y}) = \underline{x}_{2k+1} \underline{y}_{2k+1}$ which are all PL(2) since products of Lipschitz functions are PL(2) by Lemma \ref{lem:Lprods}.  The result follows by noting that $\mathbb{E}[\tau_0 \underline{\tilde{Z}}_{k+1} \underline{\beta}_{k+1}] = 0$ and $\mathbb{E}[\tau_0 \underline{\tilde{Z}}_{0, k+1} f_1(\tau_0 \underline{\tilde{Z}}_0, \underline{\beta})] =  \left(\frac{n}{N-2k}\right) \hat{\lambda}_{1} \breve{E}_{0,0},$ which follows by Stein's Lemma given in Fact \ref{fact:stein}.  We demonstrate this in the following.  Think of a function $\tilde{f}: \mathbb{R} \rightarrow \mathbb{R}$ defined as $\tilde{f}(x) =  f_1(\tau_0 \underline{\tilde{Z}}_{0, 1}, \ldots, \tau_0 \underline{\tilde{Z}}_{0,k}, x, \tau_0 \underline{\tilde{Z}}_{0, k+1}, \ldots, \tau_0 \underline{\tilde{Z}}_{0, 2k+1}, \underline{\beta})$.  Then,
\begin{align*}
\mathbb{E}[\tau_0 \underline{\tilde{Z}}_{0,k+1} f_1(\tau_0 \underline{\tilde{Z}}_0, \underline{\beta})] = \mathbb{E}[\tau_0 \underline{\tilde{Z}}_{0, k+1} \tilde{f}(\tau_0 \underline{\tilde{Z}}_{0, k+1})] \overset{(b)}{=} \tau_0^2 \mathbb{E}[f'_1(\tau_0 \underline{\tilde{Z}}_0, \underline{\beta})] =  \breve{E}_{0,0} \left(\frac{n}{N-2k}\right) \hat{\lambda}_{1}.
\end{align*}
In the above, step $(b)$ follows by Fact \ref{fact:stein}

\textbf{(g), (h)}  The proof of $\mc{H}_1 (g), (h)$ follow as the corresponding proofs in \cite{RushV16}.

\subsection{Step 3: Showing $\mc{B}_t$ holds}

We wish to show results (a) -- (h) in \eqref{eq:Ba}, \eqref{eq:Bb1}, 
%\eqref{eq:Bb2}, 
\eqref{eq:Bc}, \eqref{eq:Bd}, \eqref{eq:Be}, \eqref{eq:Bf}, \eqref{eq:Bg}, \eqref{eq:Bh} assuming that $\mc{B}_{r}$ and $\mc{H}_{r+1}$ hold for all $0 \leq r \leq t-1$ due to the inductive hypothesis.  
The proof of these results is the same as in the step $\mc{B}_t$ of the proof in \cite{RushV16} and therefore is not repeated here.

\subsection{Step 4: Showing $\mc{H}_{t+1}$ holds}
We wish to show results (a) -- (h) in \eqref{eq:Ha}, \eqref{eq:Hb1}, 
%\eqref{eq:Hb2},  
\eqref{eq:Hc},  \eqref{eq:Hd}, \eqref{eq:He}, \eqref{eq:Hf}, \eqref{eq:Hg}, \eqref{eq:Hh} assuming $\mc{B}_r$ holds for all $0 \leq r \leq t$ and $\mc{H}_{s+1}$ holds for all $0 \leq s \leq t-1$.

The probability statements in the lemma and the other parts of $\mc{H}_{t+1}$ are conditioned on the event that the matrices $\mathbf{Q}_{1}, \ldots, \mathbf{Q}_{t+1}$ are invertible, but for the sake of brevity, we do not explicitly state the conditioning in the probabilities.  The following lemma, whose proof is the same as in \cite{RushV16}, will be used to prove $\mc{H}_{t+1}$. 

%\textcolor{cyan}{The most recent version of the Finite Sample paper -- on ArXiv now -- doesn't have Lemma \ref{lem:mat_inverse} any longer.  We prove the invertibility as part of the recursion so we need to update our results to match these so that we can reference that work. \textcolor{blue}{Instead of citing the latest version, can we cite an earlier version that has Lemma \ref{lem:mat_inverse}?}}

%\begin{lem}\cite[Lemma 6]{RushV16}
%\label{lem:mat_inverse}
%The symmetric matrix $\mathbf{Q}_{t+1} := \frac{1}{n}Q_{t+1}^* Q_{t+1}$ is invertible with high probability. In particular, 
%\be
%P\left(\mathbf{Q}_{t+1} \text{ singular}\right) \leq K e^{-\kappa n \e^2}  %\ \ P\left(\mathbf{M}_t \text{ singular}\right) \leq t K_{t-1} e^{-\kappa_{t-1} \varepsilon_4 n}.
%\ee
%\end{lem}

\begin{lem}\cite[Lemma 8]{RushV16}
\label{lem:Qv_conc}
Let $v := \frac{1}{n}B^*_{t+1} m_t^{\perp} - \frac{1}{n}Q_{t+1}^*(\xi_t q^t - \sum_{i=0}^{t-1} \alpha^t_i \xi_i q^i)$ and $\mathbf{Q}_{t+1} := \frac{1}{n}Q_{t+1}^* Q_{t+1}$.  Then for $j \in [t+1]$,
\ben
P\left(\left \lvert \left[\mathbf{Q}_{t+1}^{-1} v\right]_{j} \right \lvert \geq \e\right) \leq e^{-\kappa n \e^2}.
\een
\end{lem}

\textbf{(a)}  Recall the definition of $\Delta_{t+1,t}$ from Lemma \ref{lem:hb_cond} \eqref{eq:Dt1t}.  Using Fact \ref{fact:gauss_p0}, we have
\ben
\frac{\norm{m^t_{\perp}}}{\sqrt{n}}  \mathsf{P}^{\parallel}_{Q_{t+1}} Z_t \overset{d}{=} \frac{\norm{m^t_{\perp}}}{\sqrt{n}}  \frac{1}{\sqrt{N}}\tilde{Q}_{t+1} \bar{Z}_{t+1},
\een
where matrix $\tilde{Q}_{t+1} \in \mathbb{R}^{N \times (t+1)}$ forms an orthogonal basis for the column space of $Q_{t+1}$ such that $\tilde{Q}_{t+1}^* \tilde{Q}_{t+1} = N\mathsf{I}_{t+1}$ and $\bar{Z}_{t+1} \in \mathbb{R}^{t+1}$ is an independent random vector with i.i.d.\ $\mc{N}(0,1)$ entries.  We can then write
\begin{align*}
\Delta_{t+1,t} \overset{d}{=} & \, \sum_{r=0}^{t-1} (\alpha^t_r - \hat{\alpha}^{t}_r) h^{r+1} + Z_t \left(\frac{\norm{m^t_{\perp}}}{\sqrt{n}} - \tau_{t}^{\perp}\right) - \frac{\norm{m^t_{\perp}}}{\sqrt{n}} \frac{\tilde{Q}_{t+1} \bar{Z}_{t+1}}{\sqrt{N}} + Q_{t+1} \mathbf{Q}_{t+1}^{-1}v,
\end{align*}
where $ \mathbf{Q}_{t+1} \in \mathbb{R}^{(t+1) \times (t+1)}$ and $v \in \mathbb{R}^{t+1}$ are defined in Lemma \ref{lem:Qv_conc}. By Lemma \ref{lem:squaredsums},
\begin{align*}
\frac{\norm{\Delta_{t+1,t}}^2}{(2t+3)} & \leq \sum_{r=0}^{t-1} (\alpha^t_r - \hat{\alpha}^t_r)^2 \norm{h^{r+1}}^2 +\norm{Z_t}^2 \left \lvert \frac{\norm{m^t_{\perp}}}{\sqrt{n}} - \tau_{t}^{\perp} \right \lvert^2  + \frac{\norm{m^t_{\perp}}^2}{n}\frac{ || \tilde{Q}_{t+1} \bar{Z}_{t+1}||^2 }{N} + \sum_{j=0}^{t} \norm{q^j}^2 \left[\textbf{Q}_{t+1}^{-1} v \right]_{j+1}^2, 
\end{align*}
where we have used $Q_{t+1}\textbf{Q}_{t+1}^{-1}v = \sum_{j=0}^{t} q^j \left[\textbf{Q}_{t+1}^{-1} v \right]_{j+1}$. Applying Lemma \ref{sums},
\begin{align}
& P\left(\frac{\norm{\Delta_{t+1,t}}^2}{N} \geq \epsilon \right) \leq \sum_{r=0}^{t-1} P\left(\left \lvert\alpha^t_r - \hat{\alpha}^t_r \right \lvert \frac{||h^{r+1}||}{\sqrt{N}} \geq  \sqrt{\e_t} \right) +  P\left(\left \lvert \frac{||m^t_{\perp}||}{\sqrt{n}} - \tau_{t}^{\perp}\right \lvert \frac{\norm{Z_{t}}}{\sqrt{N}} \geq  \sqrt{\e_t} \right)\nonumber \\
& \quad + P\left(\frac{||m^t_{\perp}||}{\sqrt{n}}  \frac{||\tilde{Q}_{t+1} \bar{Z}_{t+1}||}{N} \geq  \sqrt{\e_t}\right)  + \sum_{j=0}^{t} P\left(\left \lvert \left[\textbf{Q}_{t+1}^{-1} v \right]_{j+1} \right \lvert \frac{||q^j||}{\sqrt{n}} \geq \sqrt{\e_t} \right),
\label{eq:deltt1t_sq_conc}
\end{align}
where $\e_t = \frac{\e}{(2t+3)^2}$.  We now show each of the terms in \eqref{eq:deltt1t_sq_conc} has the desired upper bound.  For  $0 \leq r \leq t-1$,
\begin{align*}
& P\left(\left \lvert\alpha^t_r - \hat{\alpha}^t_r \right \lvert \frac{\norm{h^{r+1}}}{\sqrt{N}} \geq  \sqrt{\e_t} \right)  
\leq P\left(\left \lvert\alpha^t_r - \hat{\alpha}^t_r \right \lvert \left(  \left\lvert\frac{\norm{h^{r+1}}}{\sqrt{N}}  -\tau_r \right \lvert+\tau_r\right) \geq \sqrt{\e_t} \right)  \\ 
& \leq P\left(\left \lvert \alpha^t_r - \hat{\alpha}^t_r \right \lvert \geq \frac{\sqrt{\e_t}}{2}  \min\{1 , \tau_r^{-1} \} \right)+ P\left(\left \lvert \frac{\norm{h^{r+1}}}{\sqrt{N}} - \tau_r \right \lvert \geq \sqrt{\e_t}  \right) \overset{(a)}{\leq} K e^{-\kappa N \e} + K e^{- \kappa N \e} ,
\end{align*}
where step $(a)$ follows from induction hypotheses $\mathcal{B}_t (g)$, $\mathcal{H}_{1} (d)-\mathcal{H}_{t} (d)$, and Lemma \ref{sqroots}.  Next, the second term on the right side of \eqref{eq:deltt1t_sq_conc} can be bounded similarly
%\ben
%\begin{split}
% P&\left(\abs{\frac{\norm{m^t_{\perp}}}{\sqrt{n}} - \tau_{t}^{\perp}} \frac{\norm{Z_t}}{\sqrt{N}} \geq \sqrt{\e_t} \right) \leq P\left(\abs{\frac{\norm{m^t_{\perp}}}{\sqrt{n}} - \tau_{t}^{\perp}} \geq \e_t^{\frac{1}{4}} \right)  +  P\left( \frac{\norm{Z_t}}{\sqrt{N}} \geq \e_t^{\frac{1}{4}} \right) \overset{(b)}{\leq}  Ke^{-\kappa  n \sqrt{\e}} + Ke^{- \kappa N \sqrt{\e}},
%\end{split}
%\een
%where step $(b)$ is obtained 
using induction hypothesis $\mathcal{B}_t (h)$, Lemma \ref{sqroots}, and Lemma \ref{subexp}. Since $\norm{m^t_{\perp}}/\sqrt{n}$ concentrates on $\tau_{t}^{\perp}$ by $\mc{B}_t (h)$, the third term in \eqref{eq:deltt1t_sq_conc} can be bounded as 
\begin{equation}
\begin{split}
P&\left(\frac{\norm{m^t_{\perp}}}{\sqrt{n}} \cdot  \frac{\|\tilde{Q}_{t+1} \bar{Z}_{t+1}\| }{N}\geq  \sqrt{\e_t}  \right) = P\left( \left(\abs{\frac{\|m^t_\perp\|}{\sqrt{n}} - \tau_t^\perp} + \tau_t^\perp\right)\cdot \frac{\|\tilde{Q}_{t+1}\bar{Z}_{t+1}\|}{N}  \geq \sqrt{\e_t}\right) \\
 &\leq P\left( \abs{\frac{\norm{m^t_{\perp}}}{\sqrt{n}} - \tau_{t}^{\perp}} \geq  \sqrt{\e_t} \right) + P\left(\frac{1}{N}\| \tilde{Q}_{t+1}\bar{Z}_{t+1} \| \geq 
 \frac{ \sqrt{\e_t}}{2}  \min\{ 1 , (\tau_t^{\perp})^{-1}\} \right).
 \label{eq:T3split}
 \end{split}
\end{equation}
For the second term in \eqref{eq:T3split}, denoting the columns of $\tilde{Q}_{t+1} $ as $\{\tilde{q}_0, \ldots \tilde{q}_{t}\}$, we have $\|\tilde{Q}_{t+1} \bar{Z}_{t+1}\|^2 = \sum_{i=0}^{t} \norm{\tilde{q}_i}^2 (\bar{Z}_{t+1_i})^2 = N \sum_{i=0}^{t} (\bar{Z}_{t+1_i})^2,$ since the $\{\tilde{q}_i\}$ are orthogonal, and $\norm{\tilde{q}_i}^2 = N$ for $0 \leq i \leq t$.  Therefore,
\be
\begin{split}
P\left(\frac{1}{N^2} \|\tilde{Q}_{t+1}\bar{Z}_{t+1}\|^2 \geq \e' \right) \overset{(b)}{\leq} \sum_{i=0}^{t}P\left(\abs{\bar{Z}_{t+1_i}} \geq \sqrt{\frac{N\e'}{t+1}} \right) \overset{(c)}{\leq} 2 e^{-\frac{1}{2(t+1)} N \e'}.
\label{eq:Qt1Zt1conc}
\end{split}
\ee
Step $(b)$ uses Lemma \ref{sums} and step $(c)$ Lemma \ref{lem:normalconc}. Using \eqref{eq:Qt1Zt1conc} and $\mc{B}_t (h)$, the RHS of \eqref{eq:T3split} is bounded by $K\exp\{-\kappa n\e \}$.  Finally, for $0 \leq j \leq t$, the last term in \eqref{eq:deltt1t_sq_conc} can be bounded by 
\ben
\begin{split}
& P\left(\left \lvert \left[\textbf{Q}_{t+1}^{-1} v \right]_{j+1} \right \lvert \frac{\norm{q^j}}{\sqrt{n}} \geq \sqrt{\e_t} \right)
= P\left(\left \lvert \left[\textbf{Q}_{t+1}^{-1} v \right]_{j+1} \right \lvert \left( \left \lvert\frac{\norm{q^j}}{\sqrt{n}} - \sigma_j \right \lvert + \sigma_j \right) \geq \sqrt{\e_t} \right) 
\\  & \leq  P\left(\left \lvert \frac{\norm{q^j}}{\sqrt{n}} - \sigma_j \right \lvert \geq  \sqrt{\e_t} \right) + P\left(\left \lvert \left[\textbf{Q}_{t+1}^{-1} v \right]_{j+1} \right \lvert \geq \frac{\sqrt{\e_t}}{2}  \min\{1, \sigma_j^{-1}\} \right)\overset{(d)}\leq K e^{-\kappa n \e^2} + K e^{-\kappa n \e^2},
\end{split}
\een
where step $(d)$ follows from Lemma \ref{lem:Qv_conc}, the induction hypothesis $\mc{H}_{t}(e)$, and Lemma \ref{sqroots}. Thus we have bounded each term of \eqref{eq:deltt1t_sq_conc} as desired.

\textbf{(b) (i)} For brevity we define the shorthand notation $\mathbb{E}_{\phi_h}:=\mathbb{E}\left[\phi_h(\tau_0\tilde{\underline{Z}}_0,...,\tau_t\tilde{\underline{Z}}_t,\underline{\beta})\right]$,
and
\begin{equation}
a_i =\left(h^1_i,...,h^t_i,\sum_{r=0}^{t-1}\hat{\alpha}^t_r h^{r+1}_i+\tau_t^\perp Z_{t_i} + [\Delta_{t+1,t}]_i,\beta_{0_i}\right),\quad
c_i = \left(h^1_i,...,h^t_i,\sum_{r=0}^{t-1}\hat{\alpha}^t_r h^{r+1}_i + \tau_t^\perp Z_{t_i},\beta_{0_i}\right),
\label{eq:Htaici}
\end{equation}
for $i=1,...,N$. Hence $a$, $c$ are length-$N$ vectors with entries $a_i$, $c_i\in\mathbb{R}^{(t+2)}$.

Then, using the conditional distribution of $h^{t+1}$ from Lemma~\ref{lem:hb_cond} and Lemma~\ref{sums}, we have
\begin{align}
&P\left( \left\vert \frac{1}{N-2k}\sum_{i=k+1}^{N-k}\phi_h([h^1]_{i-k}^{i+k},...,[h^{t+1}]_{i-k}^{i+k},[\beta_0]_{i-k}^{i+k}) -\mathbb{E}_{\phi_h}\right\vert \geq \epsilon  \right) = P\left( \left\vert \frac{1}{N-2k}\sum_{i=k+1}^{N-k}\phi_h([a]_{i-k}^{i+k}) -\mathbb{E}_{\phi_h}\right\vert \geq \epsilon  \right)\nonumber\\
&\leq P\left(\left\vert \frac{1}{N-2k}\sum_{k+1}^{N-k} \left(\phi_h([a]_{i-k}^{i+k})-\phi_h([c]_{i-k}^{i+k})\right) \right\vert \geq \frac{\epsilon}{2}\right)
+ P\left(\left\vert \frac{1}{N-2k}\sum_{i=k+1}^{N-k}\phi_h([c]_{i-k}^{i+k})-\mathbb{E}_{\phi_h}\right\vert \geq \frac{\epsilon}{2}\right).
\label{eq:Htbi}
\end{align}
Label the two terms of \eqref{eq:Htbi} as $T_1$ and $T_2$. To complete the proof we show both are bounded by
$Ke^{-\kappa n \epsilon^2}$. First consider term $T_1$. Using the pseudo-Lipschitz property of $\phi_h$, we have
\be
\begin{split}
T_1&\leq P\left(\frac{1}{N-2k}\sum_{i=k+1}^{N-k} L (1 + \|[a]_{i-k}^{i+k}\| + \|[c]_{i-k}^{i+k}\|) \|[a-c]_{i-k}^{i+k}\| \geq \frac{\epsilon}{2}\right)\\
&\overset{(a)}{\leq}P\left( \frac{1}{N-2k} \left(\sum_{i=k+1}^{N-k} \left(1+\|[a]_{i-k}^{i+k}\|+\|[c]_{i-k}^{i+k}\|\right)^2 \right)^{1/2}\left(\sum_{i=k+1}^{N-k}\|[a-c]_{i-k}^{i+k}\|^2\right)^{1/2}  \geq \frac{\e}{2L}\right)\\
&\overset{(b)}{\leq}P\left( \left(\sum_{i=k+1}^{N-k} \frac{(1+\|[\Delta_{t+1,t}]_{i-k}^{i+k}\|^2+4\|[c]_{i-k}^{i+k}\|^2)}{N-2k} \right)^{1/2}\left(\sum_{i=k+1}^{N-k} \frac{\|[\Delta_{t+1,t}]_{i-k}^{i+k}\|^2}{N-2k}\right)^{1/2}  \geq \frac{\e}{2\sqrt{3}L}\right)\\
&\overset{(c)}{\leq} P\left( \left( 1 + \sqrt{2k+1} \frac{\|\Delta_{t+1,t}\|}{\sqrt{N-2k}} + 2\sqrt{2k+1} \frac{\|c\|}{\sqrt{N-2k}} \right) \left(\sqrt{2k+1} \frac{\|\Delta_{t+1,t}\|}{\sqrt{N-2k}}\right) \geq \frac{\e}{2\sqrt{3}L} \right).
\label{eq:bi_t1}
\end{split}
\ee
We note that in the above the notation $\|[a]_{i-k}^{i+k}\|^2$ means the sum of the $(2k+1) \times (t+2)$ squared elements of $[a]_{i-k}^{i+k}$ as defined in \eqref{eq:Htaici}.  Step $(a)$ follows by Cauchy-Schwarz, step $(b)$ uses $\|[a]_{i-k}^{i+k}\|\leq \|[c]_{i-k}^{i+k}\| + \|[\Delta_{t+1,t}]_{i-k}^{i+k}\|$, Lemma~\ref{lem:squaredsums}, and $\|[a-c]_{i-k}^{i+k}\|^2=\|[\Delta_{t+1,t}]_{i-k}^{i+k}\|^2$, and step $(c)$ uses the fact that for $a \in \mathbb{R}^N$, $\sum_{i=k+1}^{N-k} \|[a]_{i-k}^{i+k}\|^2 \leq (2k+1) \|a\|^2$.

From \eqref{eq:Htaici} and Lemma \ref{lem:squaredsums}, we have
\begin{equation}
\|c\|^2 \leq \sum_{r=0}^{t-1} \|h^{r+1}\|^2 + 2\sum_{r=0}^{t-1}\sum_{l=0}^{t-1}\hat{\alpha}_r\hat{\alpha}_l
(h^{r+1})^*h^{l+1} + 2(\tau_t^{\perp})^2\|Z_t\|^2 + \|\beta_0\|^2
\label{eq:Htupboundc}
\end{equation}
Denote the RHS of above by $\tilde{c}^2$.
From the induction hypothesis, $\frac{1}{N}(h^{r+1})^*h^{l+1}$ concentrates on $\breve{E}_{r,l}$ for $0\leq r,l\leq (t-1)$.
Using this in \eqref{eq:Htupboundc}, we will argue that $\frac{1}{N}\tilde{c}^2$ concentrates on
\begin{equation}
\mathbb{E}_{\tilde{c}}:= \sum_{l=0}^{t-1}\breve{E}_{l,l} + 2\sum_{r=0}^{t-1}\sum_{l=0}^{t-1}\hat{\alpha}^t_{r}\hat{\alpha}^t_l \breve{E}_{l,r} + 2(\tau_t^{\perp})^2 + \sigma_{\beta}^2 = \sum_{l=0}^{t-1}\tau_{l}^2 + 2\tau_t^2 + \sigma_{\beta}^2,
\label{eq:Ec_def}
\end{equation}
where the last equality is obtained using $\breve{E}_{l,l}=\tau_l^2$, and by rewriting the double sum as follows:
\begin{equation}
\sum_{r=0}^{t-1}\sum_{l=0}^{t-1}\hat{\alpha}^t_{r}\hat{\alpha}^t_{l}\breve{E}_{r,l}=(\hat{\alpha}^t)^*\breve{C}^t\hat{\alpha}^t = [\breve{E}_t^*(\breve{C}^t)^{-1}](\breve{C}^t)^{-1}[(\breve{C}^t)^{-1}\breve{E}_t]=\breve{E}_t^*(\breve{C}^t)^{-1}\breve{E}_t = \breve{E}_{t,t}-(\tau_t^{\perp})^2.
\label{eq:Htbitaut}
\end{equation}
Using Lemma~\ref{sums}, let $\e_t = \e / (t+t^2+2)$,
\begin{align}
&P \left( \left\vert \frac{\tilde{c}^2}{N}-\mathbb{E}_{\tilde{c}} \right\vert \geq \epsilon \right)
\leq \sum_{l=0}^{t-1} P\left(\left\vert \frac{\|h^{l+1}\|^2}{N} - \tau_l^2  \right\vert \leq \e_t\right)
+ P\left(\left\vert \frac{\|\beta_0\|^2}{N} - \sigma_{\beta}^2 \right\vert \geq \e_t\right)\nonumber\\
&+ \sum_{r=0}^{t-1}\sum_{l=0}^{t-1} P\left(\left\vert \frac{(h^{r+1})^*h^{l+1}}{N} - \breve{E}_{r,l} \right\vert \geq \frac{\e_t}{2\hat{\alpha}^t_r\hat{\alpha}^t_l} \right) + P\left(\left\vert \frac{\|Z_t\|^2}{N} - 1 \right\vert \geq \frac{\e_t}{2(\tau_t^{\perp})^2}\right) \overset{(a)}{\leq} K e^{-\kappa N \epsilon^2}.
\label{eq:Htbi_c}
\end{align}
In step $(a)$, we used induction hypothesis $\mathcal{H}_1(d) - \mathcal{H}_t(d)$, result \eqref{eq:beta0_assumption}, %our assumptions on $\beta_0$, 
and Lemma~\ref{subexp}.

Therefore, using \eqref{eq:bi_t1}, term $T_1$ of \eqref{eq:Htbi} can be bounded as
\begin{align*}
T_1&\leq P\left( \left( 1+ \frac{\|\Delta_{t+1,t}\|}{\sqrt{N}} + 2 \frac{\tilde{c}}{\sqrt{N}} \right) \cdot \frac{\|\Delta_{t+1,t}\|}{\sqrt{N}} \geq \frac{\epsilon(1-2k/N)}{2\sqrt{3}(2k+1)L} \right)\\
&= P\left( \left( 1+ \frac{\|\Delta_{t+1,t}\|}{\sqrt{N}} + 2 \left(\frac{\tilde{c}}{\sqrt{N}}-\mathbb{E}_{\tilde{c}}^{1/2} \right) + 2\mathbb{E}_{\tilde{c}}^{1/2} \right) \cdot \frac{\|\Delta_{t+1,t}\|}{\sqrt{N}}  \geq \frac{\epsilon(1-2k/N)}{2\sqrt{3}(2k+1)L} \right)\\
&\leq P\left(\left\vert \frac{\tilde{c}}{\sqrt{N}} -\mathbb{E}_{\tilde{c}}^{1/2} \right\vert\geq \epsilon \right)
+P\left(\frac{\|\Delta_{t+1,t}\|}{\sqrt{N}} \geq \frac{\epsilon(1-2k/N)}{2\sqrt{3}(2k+1)L(4+2\mathbb{E}_{\tilde{c}}^{1/2})} \right)\overset{(a)}{\leq} K e^{- \kappa N \epsilon^2}.
\end{align*}
In step $(a)$, we used \eqref{eq:Htbi_c}, $\mathcal{H}_{t+1}(a)$, and Lemma~\ref{sqroots}.

Next consider term $T_2$ of \eqref{eq:Htbi}. Define function $\tilde{\phi}_{h_i}:\mathbb{R}^{2k+1}\rightarrow\mathbb{R}$ as
\begin{equation}
\tilde{\phi}_{h_i}(\underline{z}):=\phi_h([h^1]_{i-k}^{i+k},...,[h^t]_{i-k}^{i+k},\sum_{r=0}^{t-1}\hat{\alpha}^t_r[h^{r+1}]_{i-k}^{i+k}+\tau_t^{\perp}\underline{z},[\beta_0]_{i-k}^{i+k})\in PL(2),
\end{equation}
for each $i=k+1,...,N-k$, where we treat all arguments except $\underline{z}$ as fixed. Let
$\underline{Z}\in\mathbb{R}^{2k+1}$ be a random vector of i.i.d. $\mc{N}(0,1)$ entries, and assume that $\underline{Z}$ is independent of $\tilde{\underline{Z}}_0,...,\tilde{\underline{Z}}_{t-1}$, then 
\begin{align*}
T_2 &= P\left(\left\vert \frac{1}{N-2k}\sum_{i=k+1}^{N-k} \tilde{\phi}_{h_i}([Z_t]_{i-k}^{i+k}) - \mathbb{E}_{\phi_h}  \right\vert\geq \frac{\epsilon}{2}\right)\\
& \leq P\left(\left\vert  \frac{1}{N-2k}\sum_{i=k+1}^{N-k}\left(\tilde{\phi}_{h_i}([Z_t]_{i-k}^{i+k}) - \mathbb{E}_{\underline{Z}}\left[\tilde{\phi}_{h_i}(\underline{Z})\right]\right)   \right\vert\geq \frac{\epsilon}{4}\right)
+ P\left(\left\vert \frac{1}{N-2k}\sum_{i=k+1}^{N-k} \mathbb{E}_{\underline{Z}}\left[\tilde{\phi}_{h_i}(\underline{Z})\right] - \mathbb{E}_{\phi_h}   \right\vert \geq \frac{\epsilon}{4}\right).
\end{align*}
The first term on the RHS of the above has the desired bound using Lemma~\ref{lem:PL_overlap_gauss_conc}. We now bound the second term.
\begin{align}
&P\left(\left\vert \frac{1}{N-2k}\sum_{i=k+1}^{N-k} \mathbb{E}_{\underline{Z}}\left[\tilde{\phi}_{h_i}(\underline{Z})\right] -\mathbb{E}_{\phi_h} \right\vert \geq \frac{\epsilon}{4}\right)\nonumber\\
&=P\left(\left\vert \frac{1}{N-2k}\sum_{i=k+1}^{N-k} \mathbb{E}_{\underline{Z}}\left[\phi_h\left( [h^1]_{i-k}^{i+k},...,[h^t]_{i-k}^{i+k},\sum_{r=0}^{t-1}\hat{\alpha}^t_r[h^{r+1}]_{i-k}^{i+k} + \tau_t^{\perp}\underline{Z},[\beta_0]_{i-k}^{i+k} \right)\right]  - \mathbb{E}_{\phi_h}\right\vert \geq \frac{\epsilon}{4}\right)\nonumber\\
&\overset{(a)}{=} P\left(\left\vert \frac{1}{N-2k}\sum_{i=k+1}^{N-k}\hat{\phi}_h\left([h^1]_{i-k}^{i+k},...,[h^t]_{i-k}^{i+k},[\beta_0]_{i-k}^{i+k}\right) - \mathbb{E}_{\phi_h} \right\vert \geq \frac{\epsilon}{4}\right).
\label{eq:HtbiT1}
\end{align}
Step $(a)$ uses the function $\hat{\phi}_h:\mathbb{R}^{(2k+1)(t+1)}\rightarrow\mathbb{R}$ defined as
\begin{equation*}
\hat{\phi}_h\left([h^1]_{i-k}^{i+k},...,[h^t]_{i-k}^{i+k},[\beta_0]_{i-k}^{i+k}\right):= \mathbb{E}_{\underline{Z}}\left[\phi_h\left( [h^1]_{i-k}^{i+k},...,[h^t]_{i-k}^{i+k},\sum_{r=0}^{t-1}\hat{\alpha}^t_r[h^{r+1}]_{i-k}^{i+k} + \tau_t^{\perp}\underline{Z},[\beta_0]_{i-k}^{i+k} \right)\right], 
\end{equation*}
which is $PL(2)$ by Lemma~\ref{lem:PLexamples}. We will now show that
\begin{equation}
\mathbb{E}\left[ \hat{\phi}_h\left(\tau_0\tilde{\underline{Z}}_0,...,\tau_{t-1}\tilde{\underline{Z}}_{t-1},\underline{\beta}\right)\right]
=\mathbb{E}\left[\phi_h(\tau_0\tilde{\underline{Z}}_0,...,\tau_{t}\tilde{\underline{Z}}_{t},\underline{\beta})\right]
=\mathbb{E}_{\phi_h},
\label{eq:HtbiEphi}
\end{equation}
and then the probability in \eqref{eq:HtbiT1} can be upper bounded by $Ke^{-\kappa n \epsilon^2}$ using the inductive hypothesis $\mathcal{H}_t(b)$. We have
\begin{equation*}
\mathbb{E}\left[\hat{\phi}_h\left(\tau_0\tilde{\underline{Z}}_0,...,\tau_{t-1}\tilde{\underline{Z}}_{t-1},\underline{\beta}\right)\right] = \mathbb{E}\left[\phi_h\left( \tau_0\tilde{\underline{Z}}_0,...,\tau_{t-1}\tilde{\underline{Z}}_{t-1},\sum_{r=0}^{t-1} \hat{\alpha}_r^t\tau_r\tilde{\underline{Z}}_r + \tau_t^{\perp}\underline{Z},\underline{\beta} \right) \right],
\end{equation*}
where we recall that $\underline{Z}$ is independent of $\tilde{\underline{Z}}_0,...,\tilde{\underline{Z}}_{t-1}$.
To prove \eqref{eq:HtbiEphi}, we need to show that 
\begin{equation*}
\left(\tau_0\tilde{\underline{Z}}_0,...,\tau_{t-1}\tilde{\underline{Z}}_{t-1},\tau_{t}\tilde{\underline{Z}}_{t},\underline{\beta}\right)
\overset{d}{=} \left( \tau_0\tilde{\underline{Z}}_0,...,\tau_{t-1}\tilde{\underline{Z}}_{t-1},\sum_{r=0}^{t-1} \hat{\alpha}_r^t\tau_r\tilde{\underline{Z}}_r + \tau_t^{\perp}\underline{Z},\underline{\beta} \right).
%\sum_{r=0}^{t-1} \hat{\alpha}_r^t\tau_r\tilde{\underline{Z}}_r + \tau_t^{\perp}\underline{Z} 
%\overset{d}{=}\tau_t \tilde{\underline{Z}}_t.
\end{equation*}
We do this by demonstrating that: 
(\emph{i}) the covariance matrix of $\sum_{r=0}^{t-1} \hat{\alpha}_r^t\tau_r\tilde{\underline{Z}}_{r} + \tau_t^{\perp}\underline{Z}$ is $\tau_t^2\textsf{I}$; and 
(\emph{ii}) the covariance $\text{Cov}\left(\tau_l \tilde{\underline{Z}}_{l} ,\sum_{r=0}^{t-1} \hat{\alpha}_r^t\tau_r\tilde{\underline{Z}}_{r} + \tau_t^{\perp}\underline{Z}\right)=\text{Cov}\left(\tau_l\tilde{\underline{Z}}_l, \tau_t\tilde{\underline{Z}}_t \right)= \breve{E}_{l,t}\textsf{I}$, for
$0 \leq l \leq (t-1)$. First consider (\emph{i}). The $(i,j)^{th}$ entry of the covariance matrix is
\begin{align*}
&\mathbb{E}\left[ \left(\sum_{r=0}^{t-1} \hat{\alpha}_r^t\tau_r\tilde{\underline{Z}}_{r,i} + \tau_t^{\perp}\underline{Z}_i\right)\left(\sum_{l=0}^{t-1} \hat{\alpha}_{l}^t\tau_l\tilde{\underline{Z}}_{l,j} + \tau_t^{\perp}\underline{Z}_j\right) \right] = \sum_{r=0}^{t-1}\sum_{l=0}^{t-1}\hat{\alpha}^t_r\hat{\alpha}^t_l \tau_r \tau_l \mathbb{E}\left[\tilde{\underline{Z}}_{r,i}\tilde{\underline{Z}}_{l,j}\right] + (\tau_t^\perp)^2\mathbb{E}\left[\underline{Z}_i \underline{Z}_j\right]\\
& \qquad\qquad\qquad \overset{(a)}{=} \begin{cases}
\sum_{r=0}^{t-1}\sum_{l=0}^{t-1}\hat{\alpha}^t_r\hat{\alpha}^t_l\breve{E}_{r,l} + (\tau_t^\perp)^2\overset{(b)}{=}\tau_t^2,\quad \text{if } i=j\\
0, \quad \text{otherwise}
\end{cases},
\end{align*}
where step $(a)$ follows from \eqref{eq:tildeZcov} and step $(b)$ follows from \eqref{eq:Htbitaut}.
Therefore, we have showed that the covariance matrix is $\tau_t^2\textsf{I}$. Next consider (\emph{ii}), for any $0\leq l \leq (t-1)$, the $(i,j)^{th}$ entry of the covariance matrix is
\begin{align*}
\mathbb{E} \left[ \tau_l \tilde{\underline{Z}}_{l,i} \sum_{r=0}^{t-1} \hat{\alpha}^t_{r} \tau_r \tilde{\underline{Z}}_{r,j} + \tau_t^{\perp} \underline{Z}_j \right] 
=\sum_{r=0}^{t-1}\hat{\alpha}^t_{r}\tau_l\tau_r \mathbb{E}\left[\tilde{\underline{Z}}_{l,i}\tilde{\underline{Z}}_{r,j}\right]\overset{(a)}{=} 
\begin{cases}
\sum_{r=0}^{t-1} \breve{E}_{l,r} \hat{\alpha}^t_{r},\quad\text{if } i=j\\
0,\quad \text{otherwise}
\end{cases},
\end{align*}
where step $(a)$ follows from \eqref{eq:tildeZcov}.
Moreover, notice that $\sum_{r=0}^{t-1} \breve{E}_{l,r}\hat{\alpha}^t_{r} = [\breve{C}^t\hat{\alpha}^t]_{l+1} = \breve{E}_{l,t},$ where the first equality holds because the required sum is the inner product of the $(l+1)^{th}$ row of $\breve{C}^t$ and $\hat{\alpha}^t$, and the second inequality follows the definition of $\hat{\alpha}^t$ in \eqref{eq:hatalph_hatgam_def}.

\textbf{(c)} We first show the concentration of $(h^{t+1})^*\beta_0/n$. Note, $\left\vert \sum_{i=1}^N h^{t+1}_i \beta_{0_i} \right\vert\leq \left\vert \sum_{i=1}^{N/2} h^{t+1}_i \beta_{0_i}  \right\vert + \left\vert \sum_{i=N/2+1}^{N} h^{t+1}_i \beta_{0_i}  \right\vert$. Then we have
\begin{align*}
P\left(\left\vert \frac{(h^{t+1})^*\beta_{0}}{n} \right\vert \geq \epsilon \right)
&\overset{(a)}{\leq} P\left(\left\vert \sum_{i=1}^{N/2} \frac{h^{t+1}_i\beta_{0_i}}{N/2} \right\vert \geq \delta\epsilon\right)
+ P\left(\left\vert \sum_{i=N/2+1}^{N} \frac{h^{t+1}_i\beta_{0_i}}{N/2} \right\vert \geq \delta\epsilon\right)\overset{(b)}{\leq} 2Ke^{-\kappa N \delta^2\epsilon^2}, 
\end{align*}
where step $(a)$ follows Lemma~\ref{sums} and step $(b)$ follows $\mathcal{H}_{t+1}(b)$ by considering $PL(2)$ functions $\phi_{1,h}, \phi_{2,h}:\mathbb{R}^{2(2k+1)}\rightarrow\mathbb{R}$ defined as $\phi_{1,h}(\underline{x},\underline{y}):=\underline{x}_1\underline{y}_1$ and $\phi_{2,h}(\underline{x},\underline{y}):=\underline{x}_{2k+1}\underline{y}_{2k+1}$.  Note that $\mathbb{E}[\tau_t \underline{\tilde{Z}}_{t_1} \underline{\beta}_1] = 0$.

We now show the concentration of $(h^{t+1})^*q^0/n$. Rewrite $(h^{t+1})^*q^0$ as
\begin{equation*}
(h^{t+1})^*q^0 = \sum_{i=k+1}^{N-k} h^{t+1}_i f_0(\underline{0},[\beta_0]_{i-k}^{i+k}) + \sum_{i=1}^k h^{t+1}_i \beta_{0_i} + \sum_{i=N-k+1}^{N} h^{t+1}_i \beta_{0_i}.
\end{equation*} 
Then we have
\begin{align*}
&P\left(\left\vert \frac{(h^{t+1})^*q^0}{n} \right\vert \geq \epsilon \right) \\
&\overset{(a)}{\leq} P\left(\left\vert \sum_{i=k+1}^{N-k} \frac{h^{t+1}_i f_0(\underline{0},[\beta_0]_{i-k}^{i+k})}{N-2k} \right\vert \geq \frac{n\epsilon}{3(N-2k)}\right)  + P\left(\left\vert \sum_{i=1}^{k} \frac{h^{t+1}_i \beta_{0_i} }{k} \right\vert \geq \frac{n\epsilon}{3k}\right)+ P\left(\left\vert \sum_{i=N-k+1}^{N} \frac{h^{t+1}_i \beta_{0_i} }{k} \right\vert \geq \frac{n\epsilon}{3k}\right)\\
& \overset{(b)}{\leq} K \exp\left\{-\frac{\kappa n^2 \epsilon^2}{9(N-2k)}\right\} + K \exp\left\{-\frac{\kappa n^2 \epsilon^2}{9k}\right\} +K \exp\left\{-\frac{\kappa n^2 \epsilon^2}{9k}\right\},
\end{align*}
where step $(a)$ follows Lemma~\ref{sums} and step $(b)$ follows $\mathcal{H}_{t+1}(b)$ by considering $PL(2)$ functions
$\phi_{1,h},\phi_{2,h},\phi_{3,h}:\mathbb{R}^{2(2k+1)}\rightarrow\mathbb{R}$ defined as
$\phi_{1,h}(\underline{x}, \underline{y}):=\underline{x}_{k+1}f_0(\underline{0},\underline{y})$, 
$\phi_{2,h}(\underline{x},\underline{y}):=\underline{x}_{1}\underline{y}_{1}$, and
$\phi_{3,h}(\underline{x},\underline{y}):=\underline{x}_{2k+1}\underline{y}_{2k+1}$.  Note that $\mathbb{E}\{\tau_t [\underline{\tilde{Z}}_{t}]_{k+1} f(\underline{0}, \underline{\beta})\} = 0$.

\textbf{(d)} Similar to $\mathcal{H}_{t+1}(c)$,  we split the inner product $(h^{r+1})^*h^{t+1}$ and then from Lemma~\ref{sums}, 
\begin{align*}
P\left(\left\vert \frac{(h^{r+1})^*h^{t+1}}{N} - \breve{E}_{r,t} \right\vert \geq \epsilon \right) &\leq P\left(\left\vert  \sum_{i=1}^{N/2} \frac{h^{r+1}_i h^{t+1}_i}{N/2} - \breve{E}_{r,t} \right\vert \geq \epsilon \right)
+  P\left(\left\vert  \sum_{i=N/2+1}^{N} \frac{h^{r+1}_i h^{t+1}_i}{N/2} - \breve{E}_{r,t} \right\vert \geq \epsilon \right)\\
&\overset{(a)}{\leq} K \exp\left\{-\frac{\kappa N \epsilon^2}{2}\right\} + K \exp\left\{-\frac{\kappa N \epsilon^2}{2}\right\},
\end{align*}
where step $(a)$ follows $\mathcal{H}_{t+1}(b)$ by considering $PL(2)$ functions $\phi_{1,h},\phi_{2,h}:\mathbb{R}^{2(2k+1)}\rightarrow\mathbb{R}$ defined as $\phi_{1,h}(\underline{x},\underline{y}):=\underline{x}_1\underline{y}_1$ and
$\phi_{1,h}(\underline{x},\underline{y}):=\underline{x}_{2k+1}\underline{y}_{2k+1}$.

\textbf{(e)} We first show the concentration of $(q^0)^*q^{t+1}/n$. Recall from \eqref{eq:Edef}, for $0 \leq r, s \leq t+1$,
\begin{equation}
\delta \tilde{E}_{r,s} = \frac{N-2k}{N}\mathbb{E}[f_r(\tau_{r-1}\tilde{\underline{Z}}_{r-1},\underline{\beta})f_s(\tau_{s-1}\tilde{\underline{Z}}_{s-1},\underline{\beta})]+\frac{2k}{N}\sigma_{\beta}^2.
\end{equation}
Then splitting $(q^0)^*q^{t+1}$ as in $\mc{H}_1 (e)$, we have
\begin{align*}
&P\left(\left\vert \frac{(q^0)^*q^{t+1}}{n} - \tilde{E}_{0,t+1} \right\vert \geq \epsilon\right)\\
&\overset{(a)}{\leq} P\left(\left\vert  \frac{1}{N-2k}\sum_{i=k+1}^{N-k} f_0(\underline{0},[\beta_0]_{i-k}^{i+k})f_{t+1}([h^{t+1}]_{i-k}^{i+k},[\beta_0]_{i+k}^{i-k}) - \mathbb{E}[f_0(\underline{0},\underline{\beta})f_{t+1}(\tau_t\tilde{\underline{Z}}_t,\underline{\beta})] \right\vert \geq \frac{n\epsilon}{3(N-2k)}\right)\\
& \qquad + P\left( \left\vert \sum_{i=1}^k \frac{(\beta_{0_i})^2}{k} - \sigma_{\beta}^2  \right\vert \geq \frac{n\epsilon}{3k}\right) 
+ P\left( \left\vert \sum_{i=N-k+1}^N \frac{(\beta_{0_i})^2}{k} - \sigma_{\beta}^2  \right\vert \geq \frac{n\epsilon}{3k} \right) \\
& \overset{(b)}{\leq} K\exp\left\{-\frac{\kappa n^2 \e^2}{9(N-2k)}\right\} + K\exp\left\{-\frac{\kappa n^2 \e^2}{9k}\right\} + K\left\{-\frac{\kappa n^2 \e^2}{9k}\right\}, 
\end{align*}
where step $(a)$ follows Lemma~\ref{sums} and step $(b)$ follows $\mathcal{H}_{t+1}(b)$ by considering
the functions
$\phi_{1,h}, \phi_{2,h}, \phi_{3,h}:\mathbb{R}^{2(2k+1)}\rightarrow\mathbb{R}$ defined as
$\phi_{1,h}(\underline{x},\underline{y}):=f_0(\underline{0},\underline{y})f_{t+1}(\underline{x},\underline{y})$,
$\phi_{2,h}(\underline{x},\underline{y}):=\underline{y}_1^2$,
and $\phi_{3,h}(\underline{x},\underline{y}):=\underline{y}_{2k+1}^2$, which are PL(2) by Lemma \ref{lem:Lprods}.

Concentration of $(q^{r+1})^*q^{t+1}/n$ can be obtained similarly by representing
\begin{equation}
(q^{r+1})^*q^{t+1}=\sum_{i=k+1}^{N-k} f_{r+1}([h^{r+1}]_{i-k}^{i+k},[\beta_0]_{i-k}^{i+k}) f_{t+1}([h^{t+1}]_{i-k}^{i+k},[\beta_0]_{i-k}^{i+k}) + \sum_{i=1}^{k} \beta_{0_i}^2 + \sum_{i=N-k+1}^{N} \beta_{0_i}^2,
\end{equation}
and using $\mathcal{H}_{t+1}(b)$ as above.

\textbf{(f)} The concentration of $\lambda_t$ around $\hat{\lambda}_t$ follows $\mathcal{H}_{t+1}(b)$ applied to
the function $\phi_h([h^{t+1}]_{i-k}^{i+k},[\beta_0]_{i-k}^{i+k}):=f_{t+1}'([h^{t+1}]_{i-k}^{i+k},[\beta_0]_{i-k}^{i+k})$.
Next, for $r \leq t$, splitting $(h^{t+1})^*q^{r+1}$ as in $\mc{H}_1 (f)$,
\begin{align*}
&P\left(\left\vert \frac{(h^{t+1})^*q^{r+1}}{n} - \hat{\lambda}_{r+1}\breve{E}_{r,t} \right\vert\geq \epsilon\right)\\
&\overset{(a)}{\leq}P\left( \left\vert \frac{1}{N-2k} \left[\sum_{i=k+1}^{N-k} h^{t+1}_i f_{r+1}([h^{r+1}]_{i-k}^{i+k},[\beta_0]_{i-k}^{i+k}) - n\hat{\lambda}_{r+1} \breve{E}_{r,t} \right] \right\vert \geq \frac{n\epsilon}{3(N-2k)}\right)\\
& \qquad +P\left( \left\vert  \sum_{i=1}^{k} \frac{h^{t+1}_i \beta_{0_i}}{k}  \right\vert   \geq \frac{n\epsilon}{3k}\right)
+ P\left( \left\vert \frac{1}{k} \sum_{i=N-k+1}^{N} \frac{h^{t+1}_i \beta_{0_i}}{k}  \right\vert   \geq \frac{n\epsilon}{3k}\right)\\
&\overset{(b)}{\leq} K \exp\left\{-\frac{\kappa n^2 \epsilon^2}{9(N-2k)}\right\} + K \exp\left\{-\frac{\kappa n^2 \epsilon^2}{9k}\right\} + K \exp\left\{-\frac{\kappa n^2 \epsilon^2}{9k}\right\},
\end{align*}
where step $(a)$ follows from Lemma \ref{sums} and step $(b)$ from $\mathcal{H}_{t+1}(b)$ by considering $PL(2)$ functions $\phi_{1,h},\phi_{2,h},\phi_{3,h}:\mathbb{R}^{2(2k+1)}\rightarrow\mathbb{R}$ defined as
$\phi_{1,h}(\underline{x},\underline{y}):=\underline{x}_{k+1}f_{r+1}(\underline{x},\underline{y})$,
$\phi_{2,h}(\underline{x},\underline{y}):=\underline{x}_1\underline{y}_1$,
$\phi_{3,h}(\underline{x},\underline{y}):=\underline{x}_{2k+1}\underline{y}_{2k+1}$. The result follows
by noticing $\mathbb{E}[\tau_t\tilde{\underline{Z}}_{t_i}\underline{\beta}_{i}]=0$, for all $i \in [N]$, 
and 
\begin{equation*}
\mathbb{E}[\tau_t\tilde{\underline{Z}}_{t_{k+1}}f_{r+1}(\tau_r\tilde{\underline{Z}}_r,\underline{\beta})] = \frac{n}{N-2k}\hat{\lambda}_{r+1} \breve{E}_{r,t},
\end{equation*}
which follows by Stein's Lemma given in Fact~\ref{fact:stein}. We demonstrate this in the following. Think of a function $\tilde{f}:\mathbb{R}\rightarrow\mathbb{R}$ defined as $\tilde{f}(x):=f_{r+1}(\tau_r\tilde{\underline{Z}}_{r_1},...,\tilde{\underline{Z}}_{r_k},x,\tilde{\underline{Z}}_{r_{k+2}},...,\tilde{\underline{Z}}_{r_{2k+1}},\underline{\beta})$. Then,
\begin{equation*}
\mathbb{E}[\tau_t\tilde{\underline{Z}}_{t_{k+1}}f_{r+1}(\tau_r\tilde{\underline{Z}}_r,\underline{\beta})]
=\mathbb{E}[\tau_t\tilde{\underline{Z}}_{t_{k+1}}\tilde{f}(\tau_r\tilde{\underline{Z}}_{r_{k+1}})]
\overset{(a)} = \tau_t\tau_r\mathbb{E}[\tilde{\underline{Z}}_{t_{k+1}}\tilde{\underline{Z}}_{r_{k+1}}]\mathbb{E}[\tilde{f}'(\tau_r\tilde{\underline{Z}}_{r})] \overset{(b)} =\frac{n\hat{\lambda}_{r+1} \breve{E}_{r,t}}{N-2k}.
\end{equation*}
Step $(a)$ applies Stein's Lemma, Fact \ref{fact:stein}.  Step $(b)$ uses the facts that $\tau_t\tau_r\mathbb{E}[\tilde{\underline{Z}}_{t_{k+1}}\tilde{\underline{Z}}_{r_{k+1}}]=\breve{E}_{r,t}$ from \eqref{eq:tildeZcov} and that the derivative of $\tilde{f}$ is the derivative of $f_t$ with respect to the middle coordinate of the first argument, along with the definition of $\hat{\lambda}_{r+1}$ in \eqref{eq:hatlambda_hatxi}. Therefore, we have obtained the desired result. 

\textbf{(g) (h)} The proof of $\mc{H}_{t+1} (g), (h)$ is similar to the proof of $\mc{B}_{t} (g), (h)$ in \cite{RushV16}.

\appendices
%\section{Appendices}
% !TEX root = main.tex

\section{Concentration Lemmas}
\label{app:conc_lemma}
In the following $\e >0$ is assumed to be a generic constant, with additional conditions specified whenever needed.  The proof of the Lemmas in this section can be found in \cite{RushV16}.

%%%State Hoeffding's inequality

%\begin{applem}[Hoeffding's inequality]
%\label{lem:hoeff_lem}
%If $X_1, \ldots, X_n$ are bounded random variables such that $a_i \leq X_i \leq b_i$, then for $\nu = 2\left[\sum_{i} (b_i -a_i)^2\right]^{-1}$
%\begin{align*}
%P\left( \frac{1}{n}\sum_{i=1}^n (X_i - \expec X_i) \geq \e \right) &\leq e^{ -\nu n^2 \e^2}, \, \, P\left( \abs{ \frac{1}{n}\sum_{i=1}^n (X_i - \expec X_i)} \geq \e \right) \leq 2e^{ -\nu n^2 \e^2}.
%\end{align*}
%\end{applem}

%%%%%%%%%
\begin{applem}[Concentration of Sums]
\label{sums}
If random variables $X_1, \ldots, X_M$ satisfy $P(\abs{X_i} \geq \e) \leq e^{-n\kappa_i \e^2}$ for $1 \leq i \leq M$, then 
\ben
P\left(\abs{  \sum_{i=1}^M X_i } \geq \e\right) \leq \sum_{i=1}^M P\left(|X_i| \geq \frac{\e}{M}\right) \leq M e^{-n (\min_i \kappa_i) \e^2/M^2}.
\een
\end{applem}

%\begin{proof}
%Notice that if $|X_i| < \frac{\e}{n}$ for all $i \in [n]$ then $|\sum_{i=0}^n X_i| < \e$ and therefore
%\ben
%P\left(|\sum_{i=0}^n X_i| \geq \e\right) \leq P\left(|X_i| \geq \frac{\e}{n} \text{ for some } i\right) \leq \sum_{i=0}^n P\left(|X_i| \geq \frac{\e}{n}\right).
%\een
%\end{proof}

\begin{applem}[Concentration of Products]
\label{products} 
For random  variables $X,Y$ and non-zero constants $c_X, c_Y$, if
\ben
P\left( \left | X- c_X \right |  \geq \e \right) \leq K e^{-\kappa n \e^2}, \quad \text{ and } \quad P\left( \left | Y- c_Y \right |  \geq \e \right) \leq K e^{-\kappa n \e^2},
\een
then the probability  $P\left( \left | XY - c_Xc_Y \right |  \geq \e \right)$ is bounded by 
\begin{align*}  
&  P\left( \left | X- c_X \right |  \geq \min\left( \sqrt{\frac{\e}{3}}, \frac{\e}{3 c_Y} \right) \right)  +  
P\left( \left | Y- c_Y \right |  \geq \min\left( \sqrt{\frac{\e}{3}}, \frac{\e}{3 c_X} \right) \right) \\
& \qquad \qquad \leq 2K \exp\left\{-\frac{\kappa n \e^2}{9\max(1, c_X^2, c_Y^2)}\right\}.
\end{align*}
\end{applem}

%\begin{proof}
%The probability of interest, $P\left( \left | XY - c_Xc_Y \right |  \geq \e \right) $, equals
%\[P\left( \left | (X -c_X)(Y-c_Y) +  (X-c_X)c_Y  + (Y-c_Y)c_X \right |  \geq \e \right). \]
%The result then follows by noting that if $ \left | X- c_X \right |  \leq \min( \sqrt{\frac{\e}{3}}, \frac{\e}{3 c_Y})$ and $ \left | Y- c_Y \right |  \leq \min ( \sqrt{\frac{\e}{3}}, \frac{\e}{3 c_X} )$, then the following terms are all bounded by $\frac{\e}{3}$:
%\ben
%\abs{(X-c_X)c_Y}, \abs{(Y-c_X)c_Y}, \text{ and } \abs{(X -c_X)(Y-c_Y)}. \qedhere
%\een
%\end{proof}

%%%%%%%%%%

\begin{applem}[Concentration of Square Roots]
\label{sqroots}
Let $c \neq 0$. Then
\ben
\text{If } P\left(\left \lvert X_n^2 - c^2 \right \lvert \geq \epsilon \right) \leq e^{-\kappa n \epsilon^2},
\text{ then }
P \left(\left \lvert \abs{X_n} - \abs{c} \right \lvert \geq \epsilon \right) \leq e^{-\kappa n \abs{c}^2 \epsilon^2}.
\een
\end{applem}
\begin{applem}[Concentration of Scalar Inverses]
\label{inverses} Assume $c \neq 0$ and $0<\e <1$. %$\e < \frac{1}{2}\abs{c}$.
\ben
\text{If } P\left(\left \lvert X_n - c \right \lvert \geq \epsilon \right) \leq e^{-\kappa n \epsilon^2},
\text{ then }
P\left(\left \lvert X_n^{-1} - c^{-1} \right \lvert \geq \epsilon \right) \leq 2 e^{-n \kappa \e^2 c^2 \min\{c^2, 1\}/4}.
\een
\end{applem}

\section{Gaussian and Sub-Gaussian Concentration}  \label{app:gauss_lemmas}

\begin{applem}
\label{lem:normalconc}
For a standard Gaussian random variable $Z$ and  $\e > 0$,
$P\left( \abs{Z} \geq \e \right) \leq 2e^{-\frac{1}{2}\e^2}$.
\end{applem}
%%------------------------

\begin{applem}[$\chi^2$-concentration]
For  $Z_i$, $i \in [n]$ that are i.i.d. $\sim \mc{N}(0,1)$, and  $0 \leq \e \leq 1$,
\[P\left(\left \lvert \frac{1}{n}\sum_{i=1}^n Z_i^2 - 1\right \lvert \geq \e \right) \leq 2e^{-n \e^2/8}.\]
\label{subexp}
\end{applem}

\begin{applem}
\cite{BLMConc} Let $X$ be a centered sub-Gaussian random variable with variance factor $\nu$, i.e., $\ln \expec[e^{tX}] \leq \frac{t^2 \nu}{2}$, $\forall t \in \mathbb{R}$. Then $X$ satisfies:
\begin{enumerate}
\item For all $x> 0$, $P(X > x)  \vee P(X <-x) \leq e^{-\frac{x^2}{2\nu}}$, for all $x >0$.
\item For every integer $k \geq 1$,
\be
\expec[X^{2k}] \leq 2 (k !)(2 \nu)^k \leq (k !)(4 \nu)^k.
\label{eq:subgauss_moments}
\ee
\end{enumerate}
\label{lem:subgauss}
\end{applem}

\section{Other Useful Lemmas}
\label{app:other}

\begin{applem} (Products of Lipschitz Functions are PL2) Let $f:\mathbb{R}^p \to \mathbb{R}$ and $g:\mathbb{R}^p \to \mathbb{R}$ be Lipschitz continuous.  Then the product function $h:\mathbb{R}^p \to \mathbb{R}$ defined as $h(x) := f(x) g(x)$ is pseudo-Lipschitz of order 2.
\label{lem:Lprods}
\end{applem}

\begin{applem} \label{lem:PLexamples}
Let $\phi: \mathbb{R}^{t+2} \rightarrow \mathbb{R}$ be $PL(2)$. Let $(c_1, \ldots, c_{t+1})$ be constants.  The function $\tilde{\phi}: \mathbb{R}^{t+1} \rightarrow \mathbb{R}$ defined as
\begin{align}
\tilde{\phi}\left(v_1, \ldots, v_t,  w\right) &= \mathbb{E}_{Z}\left[\phi\left(v_1, \ldots, v_t, \sum_{r=1}^{t} c_{r} v_{r} + c_{t+1} Z, w\right)\right]
\label{eq:PLex2}
\end{align}
where $Z \sim \mc{N}(0,1)$, is then also PL(2).

%both of the following functions are also $PL(2)$.  
%\begin{itemize}
%\item For $0 \leq i \leq N$,  $\phi_i: \mathbb{R} \rightarrow \mathbb{R}$ defined as
%\be
%\phi_i(s_i) := \phi\left(h^1_i, ..., h^{t}_i, \sum_{r=0}^{t-1} c_{r} h^{r+1}_i + c_t s_{i}, \beta_{0_i}\right)
%\label{eq:PLex1}
%\ee
%is pseudo-Lipschitz.
%\end{itemize}
\end{applem}

\begin{applem}
For any scalars $a_1, ..., a_t$ and positive integer $m$, we have  $\left(\abs{a _1} + \ldots + \abs{a_t} \right)^m \leq t^{m-1} \sum_{i=1}^t \abs{a_i}^m$.
Consequently, for any vectors $\un{u}_1, \ldots, \un{u}_t \in \mathbb{R}^N$, $\norm{\sum_{k=1}^t \un{u}_k}^2 \leq t \sum_{k=1}^t \norm{\un{u}_k}^2$.
\label{lem:squaredsums}
\end{applem}

\section{Concentration with Dependencies} 
\label{app:conc_dependent}

We first list some notation that will be used frequently in this section. Let $S\subset \mathbb{R}^d$ for some $d\in\mathbb{N}$ be a state space and $\pi$ a probability measure on $S$. Let $f:S \rightarrow \mathbb{R}$ be a measurable function. We use the following notation:

\begin{itemize}
\item The sup-norm: $\|f\|_{\infty} := \sup_{x\in S} |f(x)|$;
\item The $L^2(\pi)$-norm: for measurable function $f$, $\|f\|_{2,\pi}^2 := \int_{S} |f(x)|^2 \pi(dx)$; for signed measure $\nu$,
\begin{equation*}
\|\nu\|_{2,\pi}^2:=\begin{cases}
\int_S \left\vert\frac{d\nu}{d\pi}\right\vert^2 d\pi &\text{if } \nu\ll\pi,\\
\infty &\text{otherwise},
\end{cases}
\end{equation*}
where the $L^2(\pi)$-norm for $\nu\ll\pi$ \footnote{For two measures $\nu$ and $\gamma$, $\nu \ll \gamma$ denotes that $\nu$ is absolutely continuous w.r.t.\ $\gamma$, and $\frac{d\nu}{d\gamma}$ denotes the Radon-Nikodym derivative.} is induced from the inner-product: for $\mu,\nu\ll \pi$,
\begin{equation}
\langle \mu,\nu\rangle_{\pi}:= \int_S \frac{d\nu}{d\pi} \frac{d\mu}{d\pi} d\pi.
\label{eq:innerproduct_L2pi}
\end{equation}
\item The expected value: $\mathbb{E}_\pi f:= \int_{S} f(x) \pi(dx)$;
\item The set of all $\pi$-square-integrable functions: $L^2(\pi):= \{f:\mathbb{R}\rightarrow \mathbb{R}: \|f\|_{2,\pi}<\infty\}$
\item The set of all zero-mean $\pi$-square-integrable functions: $L_0^2(\pi):= \{f:\mathbb{R}\rightarrow \mathbb{R}: \mathbb{E}_\pi f = 0, \|f\|_{2,\pi}<\infty\}$, where the subscript $0$ represents zero-mean.
\end{itemize}

The following lemma exists in the literature and is stated here, without proof, for completeness.  The proof can be found in the citation.  Lemma \ref{lem:roberts} tells us that if a Markov chain is reversible and geometrically ergodic as defined in Definition \ref{def:geom_ergo}, then its associated linear operator has a spectral gap, the level of which controls the chain's mixing time.

\begin{applem} \cite[Theorem 2.1]{roberts1997}
\label{lem:roberts}
Consider a Markov chain with state space $S$, probability transition measure $r(x, dx')$, stationary probability measure $\gamma$, and linear operator $R$ associated with $r(x, d x')$ such that $\nu R(dx)=\int_{S} r(x',d x)\nu(dx')$ for measure $\nu$. %is defined as $R h(x)=\int_{S} h(x') r(x,d x')$ for a measurable function $h$ and $\nu R(dx)=\int_{S} r(x',d x)\nu(dx')$ for a measure $\nu$.  
If the Markov chain is reversible %(hence $R$ is self-adjoint) 
and geometrically ergodic (Definition \ref{def:geom_ergo}), then $R$ has an $L^2(\gamma)$ spectral gap. That is, for each signed measure $\nu$ with $\nu(S)=0$ and $ \|\nu\|_{2,\gamma}<\infty$,  there is a $0 < \rho < 1$ such that $\|\nu R\|_{2,\gamma} \leq \rho \|\nu\|_{2,\gamma}.$
\end{applem}
Notice that the definition of spectral gap above is identical to the definition provided in \ref{def:geom_ergo}. To see this, note that $\gamma$ is an eigen-function of $R$ with eigenvalue 1, $R$ is self-adjoint since the chain is reversible, and the eigen-functions of a self-adjoint operator are orthogonal, hence the rest of the eigen-functions are in the space that is perpendicular to $\gamma$, which is $\left\{ \nu\ll \pi\left\vert \langle \nu,\gamma \rangle_{\gamma} =0\right.\right\},$ where $\langle \nu,\gamma \rangle_{\gamma}=\int_S \frac{d\nu}{d\gamma} \frac{d\gamma}{d\gamma} d\gamma = \nu(S)$ by the definition of inner-product in \eqref{eq:innerproduct_L2pi}.

In the following, Lemmas \ref{lem:exp_PL_subgauss_vector_conc}, \ref{lem:indp_couple_geoErgodic}, and \ref{lem:spectral_gap_preserved} are preparations for the proofs of Lemmas \ref{lem:PL_overlap_gauss_conc} and \ref{lem:PLMCconc_new}, which are our new contributions.  Lemma \ref{lem:exp_PL_subgauss_vector_conc} gives a technical result about pseudo-Lipschitz functions with sub-Gaussian input.

\begin{applem}
\label{lem:exp_PL_subgauss_vector_conc}
Let $X \in\mathbb{R}^d$ be a random vector whose entries have a sub-Gaussian marginal distribution with variance factor $\nu$ as in Lemma \ref{lem:subgauss}.  Let $\tilde{X}$ be an independent copy of $X$.  If $f:\mathbb{R}^d\rightarrow \mathbb{R}$ is a pseudo-Lipschitz function with parameter $L$, then the expectation $\mathbb{E}\left[\exp\left(rf(X)\right)\right]$ satisfies the following for $0<r<[5L(2d\nu+24d^2\nu^2)^{1/2}]^{-1}$
%\be
%\begin{split}
%\mathbb{E}\left[\exp\left(rf(\underline{X})\right)\right]
%&\leq \mathbb{E}\left[\exp\left(r(f(\underline{X})-f(\tilde{\underline{X}})\right)\right]\\
%\textcolor{blue}{\mathbb{E}\left[\exp\left(r(f(\underline{X})-f(\tilde{\underline{X}})\right)\right]} &\leq \frac{1}{1-25L^2(d\nu+12d^2\nu^2)r^2},\quad\text{ for } \quad 0<r<\frac{1}{5L(d\nu+12d^2\nu^2)^{1/2}}\\
%\textcolor{blue}{\mathbb{E}\left[\exp\left(r(f(\underline{X})-f(\tilde{\underline{X}})\right)\right]} &\leq \exp\left(50L^2(d\nu+12d^2\nu^2)r^2\right),\quad\text{ for } \quad 0<r<\frac{1}{5L(2d\nu+24d^2\nu^2)^{1/2}}.
%\label{eq:threebounds}
%\end{split}
%\ee
\be
\begin{split}
\mathbb{E}[e^{rf(X)}] &\leq \mathbb{E}[e^{r(f(X)-f(\tilde{X})}] \leq [1-25 r^2 L^2(d\nu+12d^2\nu^2)]^{-1} \leq e^{50r^2 L^2(d \nu+12d^2\nu^2)}.
\label{eq:threebounds}
\end{split}
\ee
\end{applem}

\begin{proof}
Assume, without loss of generality $\mathbb{E}[f(X)] = 0.$  By Jensen's inequality, $\mathbb{E}[\exp(-rf(\tilde{X}))]\leq\exp(-r\mathbb{E}[f(\tilde{X})])= 1.$  Therefore,
\begin{align*}
\mathbb{E}\left[\exp\left(rf(X)\right)\right]&\leq \mathbb{E}\left[\exp\left(rf(X)\right)\right]\mathbb{E}[\exp(rf(\tilde{X}))]=\mathbb{E}[\exp(r(f(X)-f(\tilde{X})))],
\end{align*}
which provides the first upper bound in \eqref{eq:threebounds}.  Next,
\be
\begin{split}
\mathbb{E}[e^{r(f(X)-f(\tilde{X}))}] \overset{(a)}{\leq}\mathbb{E}&[e^{rL(1+\|X\|+\|\tilde{X}\|)\|X-\tilde{X}\|}] =\sum_{q=0}^{\infty}\frac{(rL)^q}{q!}\mathbb{E}[((1+\|X\|+\|\tilde{X}\|)\|X-\tilde{X}\|)^q] \\
& \qquad \overset{(b)}{=} \sum_{k=0}^{\infty}\frac{(rL)^{2k}}{(2k)!}\mathbb{E}[((1+\|X\|+\|\tilde{X}\|)(\|X\|+\|\tilde{X}\|))^{2k}],
\label{eq:lemmafirst1}
\end{split}
\ee
where step $(a)$ follows pseudo-Lipschitz property and step $(b)$ holds because the odd order terms are zero, along with triangle inequality.  Now consider the expectation in the last term in the string given in \eqref{eq:lemmafirst1}.
\begin{align*}
&\mathbb{E}[((1+ \norm{X} + \| \tilde{X} \| ) (\|X\|+\|\tilde{X}\|))^{2k}] = \mathbb{E}[(\|X\|+\|\tilde{X}\| + \norm{X}^2 + \| \tilde{X} \|^2 + 2 \|X\|\|\tilde{X}\|)^{2k}] \\
&\overset{(c)}{\leq}  5^{2k-1}(2 \mathbb{E} \|X\|^{2k} + 2\mathbb{E} \norm{X}^{4k} + 2^{2k}\mathbb{E}[\|X\|^{2k}\|\tilde{X}\|^{2k}] ) \overset{(d)}{\leq} 5^{2k-1}(4 (k !)(2 d \nu)^k + 4(2k)!(2 d \nu)^{2k} + 4(k !)^2(4d \nu)^{2k})
\end{align*}
In the above step $(c)$ follows from Lemma \ref{lem:squaredsums} and step $(d)$ from another application of Lemma \ref{lem:squaredsums} and Lemma \ref{lem:subgauss}. Now plugging the above back into \eqref{eq:lemmafirst1}, we find
\begin{align*}
\mathbb{E}&[e^{r(f(X)-f(\tilde{X}))}] \leq \sum_{k=0}^{\infty}\frac{(5rL)^{2k}}{5(2k)!}(4 (k !)(2 d \nu)^k + 4(2k)!(2 d \nu)^{2k} + 4 (k !)^2(4d \nu)^{2k} )\\
&\overset{(e)}{\leq} 1 + \frac{4}{5}\sum_{k=1}^{\infty}(5rL)^{2k}((d \nu)^k + (2 d \nu)^{2k} + 2^k(2d \nu)^{2k} ) \leq \sum_{k=0}^{\infty}(25r^2L^2)^k(d\nu+12d^2v^2)^k \\
&\overset{(f)}{=} \frac{1}{1-25 r^2 L^2(d\nu+12d^2\nu^2)}\overset{(g)}{\leq}\exp(50 r^2 L^2(d\nu+12d^2\nu^2)r^2),\text{ for } 0 < r< \frac{1}{5L\sqrt{2d\nu+24d^2\nu^2}},
\end{align*}
%\begin{align*}
%\mathbb{E}\left[e^{r(f(\underline{X})-f(\tilde{\underline{X}}))}\right] &\overset{(c)}{\leq}1+\sum_{k=1}^{\infty}\frac{(rL)^{2k}}{(2k)!}\mathbb{E}\left[5^{2k-1}\left(\|\underline{X}\|^{2k}+\|\underline{X}\|^{4k}+\|\tilde{\underline{X}}\|^{2k}+\|\tilde{\underline{X}}\|^{4k}+2^{2k}\|\underline{X}\|^{2k}\|\tilde{\underline{X}}\|^{2k}\right)\right]\\
%&\overset{(d)}{\leq}1+\sum_{k=1}^{\infty}\frac{(5rL)^{2k}}{5(2k)!}\left(4k!(2d\nu)^k+4(2k)!(2dv)^{2k}+2^{2k}4(k!)^2(2d\nu)^{2k}\right)\\
%&\overset{(e)}{\leq}1+\frac{4}{5}\sum_{k=1}^{\infty}(5rL)^{2k}\left((d\nu)^k+(2dv)^{2k}+2^k(2d\nu)^{2k}\right)\\
%&\overset{(f)}{\leq}1+\frac{4}{5}\sum_{k=1}^{\infty}(25r^2L^2)^k\left(d\nu+12d^2v^2\right)^k
%\leq \sum_{k=0}^{\infty}(25r^2L^2)^k\left(d\nu+12d^2v^2\right)^k\\
%&=\frac{1}{1-25L^2(d\nu+12d^2\nu^2)r^2}\overset{(g)}{\leq}\exp\left(50L^2(d\nu+12d^2\nu^2)r^2\right),\text{for }r<\frac{1}{5L(2d\nu+24d^2\nu^2)^{1/2}},
%\end{align*}
where step $(e)$ follows from the fact that $2^k (k!)^2 \leq (2k)!$, which can be seen by noting $\frac{(2k)!}{k!} = \prod_{j=1}^k (k+ j) = k! \prod_{j=1}^k \left(\frac{k}{j} + 1\right) \geq (k!) 2^k,$ step $(f)$ follows for $0 < r < [25 L^2(d\nu+12d^2\nu^2)]^{-1/2}$ providing the second bound in \eqref{eq:threebounds}, and step (g) uses the inequality $(1-x)^{-1} \leq e^{2x}$ for $x\in[0, 1/2]$ for the final bound in \eqref{eq:threebounds}. 
%
%We now demonstrate step (d):
%\begin{align*}
%\mathbb{E}\left[\|\underline{X}\|^{2k}\right]&=\mathbb{E}\left[\left(\sum_{i=1}^d \underline{X}_i^2\right)^k\right]\overset{(a)}{\leq}\mathbb{E}\left[d^{k-1}\sum_{i=1}^d \underline{X}_i^{2k}\right]=d^{k-1}\sum_{i=1}^d \mathbb{E}\left[\underline{X}_i^{2k}\right]\\
%&\overset{(b)}{\leq} d^{k-1}d2k!(2\nu)^k = 2k!(2d\nu)^k,
%\end{align*}
%where step (a) follows Lemma A.14 in \cite{RushV16} and step (b) follows Lemma A.10 in \cite{RushV16}.
\end{proof}

Lemma \ref{lem:indp_couple_geoErgodic} says that if $\{\tilde{X}_i\}_{i \in \mathbb{N}}$ and $\{X_i\}_{i \in \mathbb{N}}$ are independent, reversible, geometrically ergodic Markov chains then the process defined as $\{(X_i,\tilde{X}_i)\}_{i \in \mathbb{N}}$ is also reversible and geometrically ergodic.

\begin{applem}
\label{lem:indp_couple_geoErgodic}
Let $\{X_i\}_{i \in \mathbb{N}}$ be a time-homogeneous Markov chain on a state space $S$ with stationary probability measure $\gamma$. Assume that $\{X_i\}_{i \in \mathbb{N}}$ is reversible, geometrically ergodic on $L^2(\gamma)$ as defined in Definition \ref{def:geom_ergo}. Let $\{\tilde{X}_i\}_{i \in \mathbb{N}}$ be an independent copy of $\{X_i\}_{i \in \mathbb{N}}$. Then the new sequence defined as $\{(X_i,\tilde{X}_i)\}_{i \in \mathbb{N}}$ is a Markov chain on $S\times S$ that is reversible and geometrically ergodic on $L^2(\gamma \times \gamma)$.
\end{applem}

\begin{proof}
Assume $\{X_i\}_{i \in \mathbb{N}}$ has transition probability measure $r(x,dx')$.  Since $\{X_i\}_{i \in \mathbb{N}}$ is independent of $\{\tilde{X}_i\}_{i \in \mathbb{N}}$, we have that the transition probability measure of
$\{(X_i,\tilde{X}_i)\}_{i \in \mathbb{N}}$ is $\tilde{r}((x,\tilde{x}),(dx',d\tilde{x}'))=r(x,dx')r(\tilde{x},d\tilde{x}')$, and the stationary probability measure of $\{(X_i,\tilde{X}_i)\}_{i \in \mathbb{N}}$ is $\tilde{\gamma}(dx,d\tilde{x})=\gamma(dx)\gamma(d\tilde{x})$. In what follows, we demonstrate that $\tilde{r}((x,\tilde{x}),(dx',d\tilde{x}'))$ and $\tilde{\gamma}(dx,d\tilde{x})$ satisfy the reversibility and geometric ergodicity as defined in Definition \ref{def:geom_ergo}.
 
The reversibility of the coupled chain follows from the reversibility of the individual chains:
\begin{align*}
\tilde{r}((x,\tilde{x}),(dx',d\tilde{x}'))\tilde{\gamma}(dx,d\tilde{x}) = r(x,dx')\gamma(dx)r(\tilde{x},d\tilde{x}')\gamma(d\tilde{x}) &= r(x',dx)\gamma(dx')r(\tilde{x}',d\tilde{x})\gamma(d\tilde{x}')\\ &= \tilde{r}((x',\tilde{x}'),(dx,d\tilde{x}))\tilde{\gamma}(dx',d\tilde{x}').
\end{align*} 

To prove geometric ergodicity, we want to show that there is $\rho<1$ such that for each {\em probability measure} $\tilde{\nu}=\nu\times\nu\in L^2(\tilde{\gamma})=\{\tilde{\nu}\ll \tilde{\gamma}:\int_{S\times S}\abs{\frac{\tilde{\nu}(dz)}{\tilde{\gamma}(dz)}}\tilde{\gamma}(dz)<\infty\}$, there is $\tilde{C}_{\nu}<\infty$ such that
\begin{equation*}
\sup_{(A, \tilde{A})\in\mathcal{B}(S \times S)} \abs{\int_{S \times S} \tilde{r}^n(z,(A,\tilde{A}))\tilde{\nu}(dz)  - \tilde{\gamma}(A,\tilde{A})} \leq \tilde{C}_\nu \rho,
\end{equation*}
where $\mc{B}(S \times S)$ is the Borel sigma-algebra on $S \times S$.  Notice that
\begin{align*}
&\abs{\int_{S \times S} \tilde{r}^n(z,(A,\tilde{A}))\tilde{\nu}(dz)  - \tilde{\gamma}(A,\tilde{A})} = \abs{\int_S r^n(x,A)\nu(dx) \int_Sr^n(\tilde{x},\tilde{A})\nu(d\tilde{x}) - \gamma(A)\gamma(\tilde{A})}\\
& = \left\vert\left(\int_Sr^n(x,A)\nu(dx) - \gamma(A)\right) \left(\int_Sr^n(\tilde{x},\tilde{A})\nu(d\tilde{x}) - \gamma(\tilde{A})\right) \right.\\
& \qquad \left. +  \gamma(\tilde{A}) \left( \int_Sr^n(x,A)\nu(dx) - \gamma(A)\right) + \gamma(A) \left( \int_Sr^n(\tilde{x},\tilde{A})\nu(d\tilde{x}) - \gamma(\tilde{A})\right)\right\vert\\
& \overset{(a)}{\leq} \abs{\int_Sr^n(x,A)\nu(dx) - \gamma(A)} \abs{\int_Sr^n(\tilde{x},\tilde{A})\nu(d\tilde{x}) - \gamma(\tilde{A})}\\
& \qquad + \abs{\int_Sr^n(x,A)\nu(dx) - \gamma(A)} + \abs{\int_Sr^n(\tilde{x},\tilde{A})\nu(d\tilde{x}) - \gamma(\tilde{A})},
\end{align*}
where step $(a)$ used triangle inequality and $0 \leq \gamma(A) \leq 1$ for all $A \in \mathcal{B}(S)$. Taking the supremum of both sides of the above,
\begin{align*}
\sup_{(A,\tilde{A})\in\mathcal{B}(S\times S)} &\abs{\int_{S \times S} \tilde{r}^n(z,(A,\tilde{A}))\tilde{\nu}(dz)  - \tilde{\gamma}(A,\tilde{A})} \\
%& \leq \sup_{(A,\tilde{A})\in\mathcal{B}(S\times S)} \left\{ \abs{\int_Sr^n(x,A)\nu(dx) - \gamma(A)} \abs{\int_S r^n(\tilde{x},\tilde{A})\nu(d\tilde{x}) - \gamma(\tilde{A})}\right.\\
%&\quad  \left. + \abs{\int_Sr^n(x,A)\nu(dx) - \gamma(A)} + \abs{\int_Sr^n(\tilde{x},\tilde{A})\nu(d\tilde{x}) - \gamma(\tilde{A})}\right\}\\
& \leq \sup_{A\in\mathcal{B}(S)} \abs{\int_Sr^n(x,A)\nu(dx) - \gamma(A)} \sup_{\tilde{A}\in\mathcal{B}(S)} \abs{\int_Sr^n(\tilde{x},\tilde{A})\nu(d\tilde{x}) - \gamma(\tilde{A})} \\
& \quad + \sup_{A\in\mathcal{B}(S)} \abs{\int_Sr^n(x,A)\nu(dx) - \gamma(A)} + \sup_{\tilde{A}\in\mathcal{B}(S)} \abs{\int_Sr^n(\tilde{x},\tilde{A})\nu(d\tilde{x}) - \gamma(\tilde{A})}\\
&\overset{(a)}{\leq} C_\nu^2\rho^{2n} + 2 C_\nu\rho^n \overset{(b)}{<} (C^2_\nu + 2 C_\nu)\rho^n,
\end{align*}
where we have $\tilde{C}_\nu := C^2_\nu + 2 C_\nu < \infty$. Step $(a)$ follows from the fact that $\{X_i\}_{i\in\mathbb{N}}$ is geometrically ergodic and the definition of such in Definition \ref{def:geom_ergo} and step $(b)$ since $0 < \rho<1$. 
\end{proof}

Lemma \ref{lem:spectral_gap_preserved} says that if the Markov chain $\{\tilde{X}_i\}_{i \in \mathbb{N}}$ is reversible and geometrically ergodic then the process $\{Y_i\}_{i \in \mathbb{N}}$ defined as $Y_i=(X_{di-d+1},...,X_{di})$ has a spectral gap, the level of which controls the process's mixing time.

\begin{applem}
\label{lem:spectral_gap_preserved}
Let $\{X_i\}_{i \in \mathbb{N}}$ be a time-homogeneous Markov chain on a state space $S$ with stationary probability measure $\gamma$. Assume that $\{X_i\}_{i \in \mathbb{N}}$ is reversible, geometrically ergodic on $L^2(\gamma)$ as defined in Definition \ref{def:geom_ergo}. Define $\{Y_i\}_{i \in \mathbb{N}}$ as $Y_i=(X_{di-d+1},...,X_{di})\in S^d$, where $d$ is an integer. Then $\{Y_i\}_{i \in \mathbb{N}}$ is a stationary, time-homogeneous Markov chain with transition probability measure $p(y,dy')$ and stationary probability measure $\pi$. Moreover, the linear operator $P$ defined as $Ph(y):=\int_{S^d} h(y')p(y,dy')$ satisfies
\begin{equation}
\beta_P:=\sup_{h\in L^2_0(\pi)} \frac{\|Ph\|_{2,\pi}}{\|h\|_{2,\pi}} < 1.
\label{eq:lem6_betaP}
\end{equation}
\end{applem}

\begin{proof}
The Markov property and time-homogeneous property follow directly by the construction of $\{Y_i\}_{i\in\mathbb{N}}$.
We now verify that $\pi$ is a stationary distribution for $p(y,dy')$. That is, we need to show that $\int_{S^d}p(y,dy')\pi(dy) = \pi(dy')$.
Assume $\{X_i\}_{i \in \mathbb{N}}$ has transition probability measure $r(x,dx')$.  First we write $p(y,dy')$ and $\pi$ in terms of $r(x,dx')$ and $\gamma$:
\begin{align}
\pi(dy) &= \pi(dy_1,...,dy_d) = \prod_{i=2}^d r(y_{i-1},dy_i) \gamma(dy_1)\nonumber\\
p(y,dy') &= P\left(Y_{2}\in dy' | Y_1=y \right)
=P\left(X_{d+1}\in dy_1',...,X_{2d}\in dy_d'|X_1= y_1,...,X_d= y_{d}\right) \nonumber \\
& \quad =P\left(X_{d+1}\in dy_1',...,X_{2d}\in dy_d'|X_{d}= y_{d}\right) = r(y_d,dy_1')\prod_{i=2}^d r(y_{i-1}',dy_i').
\label{eq:lem6_pip}
\end{align}
Then we have
\begin{align*}
&\int_{y \in S^d} p(y,dy')\pi(dy) \overset{(a)}{=} \int_{y\in S^d} r(y_d,dy_1')\prod_{i=2}^d r(y_{i-1}',dy_i') \prod_{i=2}^d r(y_{i-1},dy_{i})\gamma(dy_1)\\
&\quad = \prod_{i=2}^d r(y_{i-1}',dy_i') \int_{y \in S^d} r(y_d,dy_1') \prod_{i=2}^d r(y_{i-1},dy_{i})\gamma(dy_1) \overset{(b)}{=} \prod_{i=2}^d r(y_{i-1}',dy_i') \gamma(dy_1')=\pi(dy'),
\end{align*}
where step $(a)$ follows from \eqref{eq:lem6_pip}, and step $(b)$ since $\gamma$ is the stationary probability measure for $r(x,dx')$.
Hence, we have verified that $\pi$ is a stationary probability measure for $p(y,dy')$.

We now prove \eqref{eq:lem6_betaP}.  Note $\beta_P$ is a property of the Markov chain $\{Y_i\}_{i \in \mathbb{N}}$. If $\{Y_i\}_{i \in \mathbb{N}}$ is reversible and geometrically ergodic, then we would be able show \eqref{eq:lem6_betaP} using Lemma \ref{lem:roberts} directly. However, $\{Y_i\}_{i \in \mathbb{N}}$ is non-reversible, hence, we instead relate $\beta_P$ to a similar property for the original $\{X_i\}_{i\in\mathbb{N}}$ chain, which we assume is reversible and geometrically ergodic, then use Lemma \ref{lem:roberts}.

Take arbitrary $h\in L^2_0(\pi)$, we have
\begin{equation}
\frac{\|Ph\|^2_{2,\pi}}{\|h\|^2_{2,\pi}}= \frac{\int_{S^d}\left(\int_{S^d} h(y')p(y,dy')\right)^2\pi(dy)}{\int_{S^d} h^2(y) \pi(dy)}.
\label{eq:lem6_1}
\end{equation}
First consider the numerator of \eqref{eq:lem6_1}. Plugging in
the expressions for $p(y,dy')$ and $\pi(dy)$ defined in \eqref{eq:lem6_pip}, we write the numerator as
\begin{align}
&\int_{S^d}\left( \int_{S^d}  h(y')r(y_d,dy_1')\prod_{i=2}^d r(y_{i-1}',dy_i') \right)^2 \prod_{i=2}^d r(y_{i-1},dy_i)\gamma(dy_1)\nonumber\\
&\overset{(a)}{=}\int_{S} \left( \int_{S^d} h(y')r(y_d,dy_1')\prod_{i=2}^d r(y_{i-1}',dy_i') \right)^2 \gamma(dy_d)
\overset{(b)}{=} \int_{S}  \left(  \int_{S} \tilde{h}(y_1')r(y_d,dy_1') \right)^2 \gamma(dy_d) \overset{(c)}{=} \|R\tilde{h}\|^2_{2,\gamma}.
\label{eq:lem6_numerator}
\end{align}
Step $(a)$ holds because $\gamma$ is the stationary probability measure for $r(x,dx')$ and the integrand inside the square does not involve $(y_1,...,y_{d-1})$. In step $(b)$, the function $\tilde{h}:\mathbb{R}  \rightarrow\mathbb{R}$ is defined as
\begin{equation}
\tilde{h}(y_1'):=\int_{S^{d-1}} h((y_1',...,y_d'))\prod_{i=2}^d r(y_{i-1}',dy_i').
\label{eq:lem6_htilde}
\end{equation}
In step $(c)$, the operator $R$ is defined as $R\tilde{h}(x):=\int_{S} \tilde{h}(x')r(x,dx')$.

We next show that $\tilde{h}\in L^2_0(\gamma)$ for $\tilde{h}$ defined in \eqref{eq:lem6_htilde}.  Notice that
\begin{align*}
\int_{S} \tilde{h}(y_1')\gamma(dy_1') &\overset{(a)}{=} \int_{S^d}h((y_1',...,y_d'))\pi(dy_1',...,dy_d') = \int_{S^d} h(y') \pi(dy') \overset{(b)}{=}0.
\end{align*}
Step $(a)$ follows by plugging in the definition of $\tilde{h}$ given in \eqref{eq:lem6_htilde} and the expression for $\pi$ from \eqref{eq:lem6_pip}. Step $(b)$ holds because $h\in L^2_0(\pi)$. The fact that $\|\tilde{h}\|_{2, \gamma} < \infty$ follows by an application of Jensen's Inequality and the original assumption $\norm{h}_{2, \pi} < \infty$. Hence, $\tilde{h}\in L^2_0(\gamma)$.

Next we consider the denominator of \eqref{eq:lem6_1}.
\begin{align}
& \int_{S^d}  h^2(y) \pi(dy) = \int_{S^d} h^2((y_1,...,y_d)) \prod_{i=2}^{d}r(y_{i-1},dy_i)\gamma(dx_1) \nonumber\\
& \qquad \overset{(a)}{\geq} \int_S \left(\int_{S^{d-1}} h((y_1,...,y_d)) \prod_{i=2}^{d}r(y_{i-1},dy_i) \right)^2\gamma(dy_1)\overset{(b)}{=} \int_S \tilde{h}^2(y_1) \gamma(dy_1) = \|\tilde{h}\|^2_{2,\gamma},
\label{eq:lem6_deno}
\end{align}
where step $(a)$ follows from Jensen's inequality and step $(b)$ uses the definition of $\tilde{h}$ given in \eqref{eq:lem6_htilde}.  Combining \eqref{eq:lem6_numerator} and \eqref{eq:lem6_deno}, we have $\forall h\in L^2_0(\pi)$, $\frac{\|Ph\|_{2,\pi}}{\|h\|_{2,\pi}} \leq \frac{\|R\tilde{h}\|_{2,\gamma}}{\|\tilde{h}\|_{2,\gamma}},$ where $\tilde{h}$ is defined in \eqref{eq:lem6_htilde} and we have $\tilde{h}\in L^2_0(\gamma)$ as demonstrated above. 
Let $\tilde{\mathsf{H}}\subset L_0^2(\gamma)$ be the collection of functions defined in \eqref{eq:lem6_htilde} for all $h\in L_0^2(\pi)$. Then we have
\begin{equation}
\beta_P = \sup_{h\in L_0^2(\pi)} \frac{\|Ph\|_{2,\pi}}{\|h\|_{2,\pi}} \leq \sup_{\tilde{h} \in \tilde{\mathsf{H}}} \frac{\|R\tilde{h}\|_{2,\gamma}}{\|\tilde{h}\|_{2,\gamma}} \overset{(a)}{\leq} \sup_{\tilde{h}\in L_0^2(\gamma)} \frac{\|R\tilde{h}\|_{2,\gamma}}{\|\tilde{h}\|_{2,\gamma}} = \beta_{R},
\label{eq:betaPbetaR}
\end{equation}
where step $(a)$ holds because $\tilde{\mathsf{H}}\subset L_0^2(\gamma)$.

Finally, let us show $\beta_R<1$.
By Lemma \ref{lem:roberts}, we have that for each {\em signed measure} $\nu \in L^2(\gamma)$ with $\nu(S)=0$, we have
%\begin{equation}
%\int_{S}  \abs{\frac{\int_{S}r(x',dx) \nu(dx')}{\gamma(dx)}}^2 \gamma(dx) \leq \rho \int_{S} \abs{\frac{\nu(dx)}{\gamma(dx)}}^2 \gamma(dx).
%\label{eq:lem6_2}
%\end{equation} 
%Define $h(x):=\nu(dx)/\gamma(dx)$, which is well-defined since $\nu \ll \gamma$.
\begin{equation}
\int_{S}  \abs{\frac{d(\nu R)}{d\gamma}}^2 d\gamma \leq \rho \int_{S} \abs{\frac{d\nu}{d\gamma}}^2 d\gamma.
\label{eq:lem6_2}
\end{equation} 
Define $h:=d\nu/d\gamma$, which is well-defined since $\nu \ll \gamma$.
By the reversibility, we have
\begin{equation*}
\frac{\int_{S }r(x',dx)\nu(dx')}{\gamma(dx)} = \int_{S}\frac{r(x,dx') \nu(dx')}{\gamma(dx')} = \int_{S} h(x') r(x,dx'),
\end{equation*}
Therefore, \eqref{eq:lem6_2} can be written as $\int_{S} \left( \int_{S } h(x') r(x,dx') \right)^2 \gamma(dx) \leq \rho \int_{S} (h(x))^2 \gamma(dx),$ for all $\nu$ such that $0 =\nu(S) = \int_{S} (\nu(dx)/\gamma(dx)) \gamma(dx) = \int_{S} h(x)\gamma(dx)$.
Therefore, $\beta_{R} = \sup_{h\in L_0^2(\gamma)} \frac{\|Rh\|_{2,\gamma}}{\|h\|_{2,\gamma}} \leq \rho <1.$ We have shown the result of \eqref{eq:lem6_betaP} by showing that that $\beta_P \leq \beta_R <1$.
\end{proof}

The following three lemmas are the key lemmas for proving Lemma \ref{lem:main_lem} and, therefore, our main result, Theorem \ref{thm:main_amp_perf}, as well.  The next lemma shows us that a normalized sum of pseudo-Lipschitz functions with Gaussian input vectors concentrate at their expected value.

\begin{applem} 
\label{lem:PL_overlap_gauss_conc}
Let $Z_1,Z_2, \ldots$ be i.i.d.\ standard Gaussian random variables. 
Define $Y_i = (Z_{i},...,Z_{i+d-1})$, for $i=1,...,n$ and let $f_i:\mathbb{R}^d\rightarrow\mathbb{R}$ be pseudo-Lipschitz functions.  Then, for $\e \in (0,1)$,
%Then, $\forall t \in (0, u(L,d))$, where $u(L,d)$ is a function of the maximum Lipschitz constant, $L$, and the dimension, $d$, %$u(L,d)=10(2d+24d^2)^{1/2}L(2d^2-d)$, 
there exists constants $K, \kappa>0$, independent of $n, \e$, such that
\begin{equation*}
P\left(\left\vert\frac{1}{N}\sum_{i=1}^{N}\left(f_i(Y_i)-\mathbb{E}\left[f_i(Y_i)\right]\right)\right\vert\geq \e\right)\leq Ke^{-\kappa n\e^2}.
\end{equation*}
\end{applem}

\begin{proof}
Without loss of generality, assume $\mathbb{E}\left[f_i(Y_i)\right]=0,$ for all $i \in [n]$. In what follows we demonstrate
the upper-tail bound:
\begin{equation}
P\left(\frac{1}{n}\sum_{i=1}^{n}f_i(Y_i)\geq \e \right)\leq Ke^{-\kappa n\e^2},
\label{eq:upper.tail}
\end{equation}
and the lower-tail bound follows similarly.  Together they provide the desired result.

Using the Cram\'{e}r-Chernoff method:
\begin{align}
P\left(\frac{1}{n}\sum_{i=1}^{n}f_i(Y_i)\geq \e \right) &=P\left(e^{r\sum_{i=1}^{n}f_i(Y_i)}\geq e^{nr \e}\right) \leq e^{-nr \e}\mathbb{E}\left[e^{r\sum_{i=1}^{n}f_i(Y_i)}\right] \quad \text{ for } r>0.
\label{eq:CC1}
\end{align}
Let $L_i$ be the pseudo-Lipschitz parameters associated with functions $f_i$  for $i=1,...,n$ and define $L := \max_{i\in[n]} L_i$.  In the following, we will show that
\begin{equation}
\mathbb{E}\left[e^{r\sum_{i=1}^{n}f_i(Y_i)}\right] \leq \exp\left(\kappa' n r^2\right), \text{ for } 0<r<[5Ld\sqrt{2d+24d^2}]^{-1},
\label{eq:lem1_0}
\end{equation}
where $\kappa'$ is any constant that satisfies $\kappa'\geq 150L^2d(d+12d^2)$. Then plugging \eqref{eq:lem1_0} into \eqref{eq:CC1}, we can obtain the desired result in \eqref{eq:upper.tail}: $P\left(\frac{1}{n}\sum_{i=1}^{n}f_i(Y_i)\geq \e \right) \leq \exp\{-n(r\e - \kappa'  r^2)\}.$  Set $r = \e/(2\kappa')$, the choice that maximizes the term $(r\e - \kappa'  r^2)$ over $r$ in the exponent in the above.  We can ensure that for $\e \in (0,1)$, $r$ falls within the region required in \eqref{eq:lem1_0} by choosing $\kappa'$ large enough.

Now we show \eqref{eq:lem1_0}.  Define index sets $I_j:=\{j+kd \, | \, k=0,...,\lfloor \frac{n-j}{d}\rfloor\}$ for $j=1,...,d$, let $C_j$ denote the cardinality of $I_j$. 
We notice that for any fixed $j$, the $Y_i$'s are i.i.d.\ for $i\in I_j$.  For example, if $j=1$ then the index set $I_1 = \{1, 1+d, 1+2d, \ldots, 1 + \lfloor \frac{n-1}{d} \rfloor d\}$ and $Y_1 = (Z_1, \ldots, Z_d)$ is independent of $Y_{1+d} = (Z_{1+d}, \ldots, Z_{2d})$, which are both independent of $Y_{1+2d} = (Z_{2d+1}, \ldots, Z_{3d})$, and so on.  Also, we have $[n] = \cup_{j=1}^d I_j$, and $I_j \cap I_s=\emptyset$, for $j\neq s$, making the collection $I_1, I_2, \ldots, I_d$ a partition of $i \in [n]$. Therefore, $\sum_{i=1}^{n}f_i(Y_i) = \sum_{j=1}^d\sum_{i\in I_j} f_i(Y_i) = \sum_{j=1}^{d}p_j \cdot \frac{1}{p_j}\sum_{i\in I_j}f_i(Y_i),$ where $0 < p_1,...,p_d <1$ are probabilities satisfying $\sum_{j=1}^d p_j=1$.  Using the above,
\begin{align}
&\mathbb{E}\left[\exp\left(r\sum_{i=1}^{n}f_i(Y_i)\right)\right]
=\mathbb{E}\left[\exp\left(\sum_{j=1}^{d} p_j \cdot \frac{r}{p_j}\sum_{i\in I_j}f_i(Y_i)\right)\right] \overset{(a)}{\leq} \sum_{j=1}^d p_j \mathbb{E}\left[\exp\left(\frac{r}{p_j}\sum_{i\in I_j}f_i(Y_i)\right)\right]  \nonumber \\
&\overset{(b)}{=}\sum_{j=1}^d p_j \prod_{i\in I_j}\mathbb{E}\left[\exp\left(\frac{r}{p_j}f_i(Y_i)\right)\right] \overset{(c)}{\leq} \sum_{j=1}^d p_j \exp\left(\frac{50C_jL^2r^2 (d+12d^2)}{p_j^2}\right),
\label{eq:lem1_1}
\end{align}
where step $(a)$ follows from Jensen's inequality, step $(b)$ from the fact that the $Y_i$'s are independent for $i\in I_j$,
and step $(c)$ from Lemma~\ref{lem:exp_PL_subgauss_vector_conc} noting that the marginal distribution of any element of $Y_i$ is Gaussian and therefore sub-Gaussian with variance factor $\nu=1$ and restriction 
\begin{equation}
0 < r < [5L\sqrt{2d + 24d^2}]^{-1} \max_j p_j.
\label{eq:r_region}
\end{equation}

Let $p_j=\sqrt{C_j}/C$, where $C=\sum_{j=1}^d \sqrt{C_j}$ ensuring that $\sum_{j=1}^d p_j = 1$. Then, we have
\begin{align*}
\sum_{j=1}^d p_j \exp\left(\frac{50 C_jL^2r^2 (d+12d^2)}{p_j^2}\right)= e^{50C^2L^2r^2 (d+12d^2)} &\overset{(a)}{\leq} e^{150dL^2(d+12d^2)nr^2} \leq e^{\kappa' n r^2},
%\label{eq:lem1_2}
\end{align*}
whenever $\kappa' \geq 150dL^2(d+12d^2)$.
%Therefore, plugging results \eqref{eq:lem1_1} and \eqref{eq:lem1_2} into \eqref{eq:CC1} we find:
%\begin{equation}
%P\left(\frac{1}{n}\sum_{i=1}^{n}f_i(Y_i)\geq \e\right)\leq 
%\exp\left(-nr \e+50C^2L^2r^2(d+12d^2)\right) \overset{(d)}{=}\exp\left(-\frac{n^2 \e^2}{200L^2C^2 (d+12d^2)}\right).
%\label{eq:lem1_3}
%\end{equation}
%where step $(d)$ is obtained by setting $r=n \e /(100 L^2 C^2(d+12d^2))$, which minimizes $-nr \e+50 C^2L^2r^2 (d+12d^2)$ over $r>0$.  Note that this specification for $r$ determines the effective region for $\e$ since step $(c)$ above requires $0 < r < [5L\sqrt{2d + 24d^2}]^{-1} \max_j C_j/C.$  We perform this calculation -- to specify the effective region for $\e$ -- at the end.  Finally, notice that
In the above, step $(a)$ follows from:
\begin{align*}
C^2&=\left(\sum_{j=1}^d\sqrt{C_j}\right)^2=\sum_{j=1}^d C_j +\sum_{j=1}^d\sum_{k\neq j}\sqrt{C_jC_k} \overset{(b)}{\leq} n +d(d-1)C_1  
\overset{(c)}{\leq} dn + 2d(d-1) <3dn, 
%\label{eq:lem1_4}
\end{align*} 
where step $(b)$ holds because $C_1= \max_{j\in[d]}C_j$ and step $(c)$ holds because $C_1=\lfloor\frac{n-1}{d}\rfloor+1\leq\frac{n}{d}+2$.  

Finally, we consider the effective region for $r$ as required in \eqref{eq:r_region}.
Notice that $\max_j p_j = \sqrt{C_1}/C > 1/d$. Hence, if we require $0<r< [5Ld\sqrt{2d+24d^2}]^{-1}$, then \eqref{eq:r_region} is satisfied.

%Plugging \eqref{eq:lem1_4} into \eqref{eq:lem1_3}, we find
%\begin{equation*}
%P\left(\frac{1}{n}\sum_{i=1}^{n}f_i(Y_i)\geq  \e \right)\leq \exp\left(-\frac{n\e^2}{200(d+12d^2)L^2(d+2d(d-1)/n)}\right).
%\end{equation*}
%
%To complete the proof, we calculate the effective region for $\e$. We require
%\begin{equation}
%r = \frac{n \e}{100 L^2 C^2(d+12d^2)} < \frac{ \max_j \sqrt{C_j}}{5CL\sqrt{2d + 24d^2}},
%\label{eq:lem1_5}
%\end{equation}
%from which it follows that we require,
%{\color{red}
%\begin{equation*}
%\e < 10 L \sqrt{2d+24d^2}\frac{C \max_j \sqrt{C_j}}{n}.
%\end{equation*}
%Notice that 
%\begin{equation*}
%\frac{C \max_j \sqrt{C_j}}{n} = \frac{\left(\sum_j \sqrt{C_j}\right)\sqrt{C_1}}{n} \geq \frac{C_1}{n} \geq \frac{n-1}{dn} \geq \frac{1}{2d},\quad \text{for } n\geq 2. 
%\end{equation*}
%Hence, by requiring $\e < 5 L \sqrt{2d+24d^2}/d$, the requirement for $r$ in \eqref{eq:lem1_5} is satisfied except the uninteresting case where $n<2$.
%}
%\begin{align*}
%\e < \frac{10LC^2\sqrt{2d+24d^2}}{n} \leq \frac{10L(2d^2 - d)\sqrt{2d+24d^2}}{n},
%\end{align*}
%where the last inequality follows from \eqref{eq:lem1_4} and the fact that $d \geq 1$.
\end{proof}

The following lemma shows us that a normalized sum of pseudo-Lipschitz functions with Markov chain input vectors concentrate at its expected value under certain conditions on the Markov chain.
\begin{applem}
\label{lem:PLMCconc_new}
Let $\{\beta_i\}_{i\in\mathbb{N}}$ be a time-homogeneous, stationary Markov chain on a bounded state space $S \subset \mathbb{R}$.
%, meaning there is $M>0$ such that $|x|\leq M, \forall x \in S$.  
Denote the transition probability measure of $\{\beta_i\}_{i\in\mathbb{N}}$ by $r(x,dy)$ and stationary probability measure by $\gamma$. 
Assume that the Markov chain is reversible and geometrically ergodic on $L^2(\gamma)$ as defined in Definition \ref{def:geom_ergo}.  
%Further assume that $\gamma$ has finite second and fourth moment.
%, denoted by $m_2$ and $m_4$, respectively.

Define $\{X_i\}_{i \in [n]}$ as $X_i=(\beta_i,...,\beta_{i+d-1})\in S^d$.  
Let $f:\mathbb{R}^d\rightarrow\mathbb{R}$ be a measurable function that satisfies the pseudo-Lipschitz condition.
%with pseudo-Lipschitz constant $L$.  
Then, for all $\e\in (0,1)$,
%Then, for $0< \e <u(L, d, M, m_2, m_4)$, where $u(L, d, M, m_2, m_4)$ is a function of the pseudo-Lipschitz constant, $L$, the dimension, $d$, the bound of the state space, $M$, and the moments of the stationary probability measure $\gamma$, $m_2$ and $m_4$, 
there exists constants $K,\kappa>0$ that are independent of $n,\e$, such that $P\left(\left\vert\frac{1}{n}\sum_{i=1}^n f(X_i)-\mathbb{E}_\pi f \right\vert\geq  \e \right)\leq Ke^{-\kappa n \e^2},$ where the probability measure $\pi$ is defined as $\pi(dx)=\pi(dx_1,...,dx_d) := \prod_{i=2}^d r(x_{i-1},dx_{i})\gamma(dx_1).$
\end{applem}

\begin{proof}
First, we split $\{X_i\}_{i \in [n]}$ into $d$ subsequences, each containing every $d^{th}$ term of $\{X_i\}_{i \in [n]}$, beginning from $1, 2, \ldots, d$.  Label these $\{X^{(1)}_i\}_{i\in [n_1]},...,\{X^{(d)}_i\}_{i \in [n_d]}$ with $\{X^{(s)}_i\}_{i \in [n_s]}:=\{X_{s+kd}:k=1,...,n_s\}$, where $n_s = \lfloor \frac{n-d-s+1}{d} \rfloor$, for $s=1,...,d$.  

Notice that $\sum_{i=1}^n f(X_i) = \sum_{s=1}^d \sum_{i=1}^{n_s} f(X^{(s)}_i)$. Using Lemma~\ref{sums}, we have
\begin{equation}
P\left( \left\vert \frac{1}{n}\sum_{i=1}^{n} f(X_i) - \mathbb{E}_\pi f  \right\vert \geq \epsilon \right) \leq \sum_{s=1}^{d} P\left( \left\vert \frac{1}{n_s}\sum_{i=1}^{n_s} f(X^{(s)}_i) - \mathbb{E}_\pi f \right\vert \geq \frac{n\epsilon}{dn_s} \right).
\label{eq:lem2_new_split}
\end{equation}

In the following, without loss of generality, we assume $\mathbb{E}_{\pi}f=0$ and demonstrate the upper-tail bound for $\{X^{(1)}_i\}_{i \in [n_1]}$:
\begin{equation}
P\left(\frac{1}{n_1}\sum_{i=1}^{n_1} f(X_i^{(1)})\geq \epsilon \right) \leq Ke^{-\kappa n_1 \epsilon^2}.
\label{eq:lem2_new_3}
\end{equation}
The lower-tail bound follows similarly, as do the corresponding results for $s = 2, 3, \ldots, d$.  Together using \eqref{eq:lem2_new_split} these provide the desired result. Using the Cram\'{e}r-Chernoff method: for $r>0$,
\begin{equation}
P\left(\frac{1}{n_1}\sum_{i=1}^{n_1} f(X^{(1)}_i)\geq \e \right)
=P\left(\exp\{r\sum_{i=1}^{n_1} f(X_i^{(1))}\}\geq \exp\{rn_1\e\}\right)
\leq \exp\{-rn_1 \e\}\mathbb{E}\left[\exp\{r\sum_{i=1}^{n_1} f(X^{(1)}_i)\}\right].
\label{eq:lem2_new_4}
\end{equation}
In what follows we will upper bound the expectation $\mathbb{E}\left[e^{r\sum_{i=1}^{n_1}f(X^{(1)}_i)}\right]$ to show \eqref{eq:lem2_new_3}.

Let $\{\tilde{X}^{(1)}_i\}_{i \in [n_1]}$ be an independent copy of $\{X^{(1)}_i\}_{i \in [n_1]}$.
By Jensen's inequality, we have
\begin{equation*}
\mathbb{E}\left[\exp\left\{-r\sum_{i=1}^{n_1}f(\tilde{X}^{(1)}_i)\right\}\right]
\geq \exp\left\{-r\mathbb{E}\left[\sum_{i=1}^{n_1}f(\tilde{X}^{(1)}_i)\right]\right\}
= \exp\left\{-r\sum_{i=1}^{n_1}\mathbb{E}\left[f(\tilde{X}^{(1)}_i)\right]\right\}=1.
\end{equation*}
Therefore,
\begin{equation}
\begin{split}
\mathbb{E}\left[\exp\left\{r\sum_{i=1}^{n_1}f(X^{(1)}_i)\right\}\right]
&\leq\mathbb{E}\left[\exp\left\{r\sum_{i=1}^{n_1}f(X^{(1)}_i)\right\}\right]\mathbb{E}\left[\exp\left\{-r\sum_{i=1}^{n_1}f(\tilde{X}^{(1)}_i)\right\}\right] \\
& =\mathbb{E}\left[\exp\left\{r\sum_{i=1}^{n_1}\left(f(X^{(1)}_i)-f(\tilde{X}^{(1)}_i)\right)\right\}\right].
\label{eq:lem2_new_4aa}
\end{split}
\end{equation}
Let $Z^{(1)}_i:=(X^{(1)}_i,\tilde{X}^{(1)}_i)$, and $g(Z^{(1)}_i):=f(X^{(1)}_i)-f(\tilde{X}^{(1)}_i)$ for $i = 1, 2, \ldots, n_1$.
We have shown $\mathbb{E}[\exp\{r\sum_{i=1}^{n_1}f(X^{(1)}_i)\}]\leq\mathbb{E}[\exp\{r\sum_{i=1}^{n_1}g(Z^{(1)}_i)\}]$ and therefore,
in what follows we provide an upper bound for $\mathbb{E}[\exp\{r\sum_{i=1}^{n_1}g(Z^{(1)}_i)\}]$ which can be used in \eqref{eq:lem2_new_4}.

We begin by demonstrating some properties of the sequence $\{Z^{(1)}_i\}_{i \in [n_1]}$, which will be used in the proof.
By construction, $\{Z^{(1)}_i\}_{i \in [n_1]}$ is a time-homogeneous Markov chain on state space $D=S^d\times S^d$. Denote its marginal probability measure by $\mu$ and transition probability measure by $q(z,dz')$. In order to obtain more useful properties, it is helpful to relate $\{Z^{(1)}_i\}_{i \in [n_1]}$ to the original Markov chain $\{\beta_i\}_{i\in\mathbb{N}}$, which we have assumed to be reversible and geometrically ergodic. 

The construction of $\{Z^{(1)}_i\}_{i \in [n_1]}$ can alternatively be thought of as follows. Let $\{\tilde{\beta}_i\}_{i\in\mathbb{N}}$ be an independent copy of $\{\beta_i\}_{i\in\mathbb{N}}$. Then by Lemma \ref{lem:indp_couple_geoErgodic}, $\{(\beta_i,\tilde{\beta}_i)\}_{i \in\mathbb{N}}$ is reversible and geometrically ergodic.
Also notice that the elements of $\{Z^{(1)}_i\}_{i \in [n_1]}$ consist of successive non-overlapping elements of $\{(\beta_i,\tilde{\beta}_i)\}_{i \in\mathbb{N}}$, same as the construction of $\{Y_i\}_{i\in\mathbb{N}}$ in Lemma \ref{lem:spectral_gap_preserved}. Therefore, the results in Lemma \ref{lem:spectral_gap_preserved} imply that the marginal probability measure $\mu$ is a stationary measure of the transition probability measure $q(z,dz')$.
Moreover, the linear operator $Q$ defined as 
\begin{equation}
\label{eq:lem2_new_4a}
Qh(z):=\int_D h(z')q(z,dz')
\end{equation}
satisfies:
\begin{equation}
\beta_Q:=\sup_{h\in L^2_0(\mu)} \frac{\|Qh\|_{2,\mu}}{\|h\|_{2,\mu}} <1.
\label{eq:lem2_new_betaQ}
\end{equation}
With the result $\beta_Q<1$, we are now ready to bound  $\mathbb{E}[\exp\{r\sum_{i=1}^{n_1}g(Z^{(1)}_i)\}]$, where we will use a method similar to the one introduced in \cite[Section 4]{lezaud1998}.

Define $m(z) := \exp\left(rg(z)\right)$, for all $z \in D$, and so we can represent the expectation that we hope to upper bound in the following way:
\be
\mathbb{E}[\exp\{r\sum_{i=1}^{n_1} g(Z^{(1)}_i)\}] = \mathbb{E}\left[\prod_{i=1}^{n_1} m(Z^{(1)}_i)\right].
\label{eq:lem2_new_05}
\ee
To provide an upper bound for \eqref{eq:lem2_new_05}, we first define a sequence $\{a_i\}_{i \in [n_1]}$ as $a_0=1$ and
\begin{equation}
a_i = \mathbb{E}[\exp\{r\sum_{j=1}^i g(Z^{(1)}_j)\}] = \mathbb{E}\left[\prod_{j=1}^{i} m(Z^{(1)}_j)\right], \quad \text{ for } 1\leq i\leq n_1.
\label{eq:lem2_new_05a}
\end{equation}
Note then that $a_{n_1}$ equals the expectation in \eqref{eq:lem2_new_05} and we have
\begin{align}
a_{n_1} = \mathbb{E}\left[\prod_{i=1}^{n_1} m(Z^{(1)}_i) \right] &\overset{(a)}{=} \int_{D^{n_1}} \mu(dz_1) m(z_1) \prod_{i=2}^{n_1} q(z_{i-1},dz_i)m(z_i) \nonumber \\
&= \int_{D^{n_1-1}} \mu(dz_1) m(z_1) \prod_{i=2}^{n_1-1} q(z_{i-1},dz_i)m(z_i) \int_D q(z_{n_1-1},dz_{n_1})m(z_{n_1}).
\label{eq:lem2_new_06}
\end{align}
In step $(a)$ we use the fact that $\{Z^{(1)}_i\}_{i \in [n_1]}$ is a Markov Chain in its stationary distribution, $\mu$, with probability transition measure $q(z, dz')$.  Now, let $b_1:= \mathbb{E}_{\mu}m$, which is a constant value, and $m_1:=m-b_1$.  Then $m(z_{n_1}) = b_1 + m_1(z_{n_1})$, and so it follows from \eqref{eq:lem2_new_06},
\begin{align}
a_{n_1} &=  \int_{D^{n_1-1}} \mu(dz_1) m(z_1) \prod_{i=2}^{n_1-1} q(z_{i-1},dz_i)m(z_i) \int_D q(z_{n_1-1},dz_{n_1})\left(b_1+m_1(z_{n_1})\right) \nonumber \\
&= b_1\int_{D^{n_1-1}} \mu(dz_1) m(z_1) \prod_{i=2}^{n_1-1} q(z_{i-1},dz_i)m(z_i) \nonumber \\
& \qquad + \int_{D^{n_1-1}} \mu(dz_1) m(z_1) \prod_{i=1}^{n_1-1} q(z_{i-1},dz_i)m(z_i) \int_D q(z_{n_1-1},dz_{n_1})m_1(z_{n_1}) \nonumber \\
&\overset{(b)}{=} a_{n_1-1}b_1 +  \int_{D^{n_1-1}} \mu(dz_1) m(z_1) \prod_{i=2}^{n_1-1} q(z_{i-1},dz_i)m(z_i) Qm_1(z_{n_1-1}). \label{eq:lem2_new_07}
\end{align}
Step $(b)$ uses the definition of $a_{n_1-1}$ given in \eqref{eq:lem2_new_05a} and the linear operator defined in \eqref{eq:lem2_new_4a}.  Now consider the integral in \eqref{eq:lem2_new_07}, which we split as in \eqref{eq:lem2_new_06} in the following:
\begin{align*}
\int_{D^{n_1-1}} &\mu(dz_1) m(z_1) \prod_{i=2}^{n_1-1} q(z_{i-1},dz_i)m(z_i)Qm_1(z_{n_1-1}) \\
&= \int_{D^{n_1-1}} \mu(dz_1) m(z_1) \prod_{i=2}^{n-2} q(z_{i-1},dz_i)m(z_i) \int_D q(z_{n_1-2},dz_{n_1-1})m(z_{n_1-1})Qm_1(z_{n_1-1}).
 \label{eq:lem2_new_8}
%&=  \int_{D^{n_s-2}} \mu(dz_1) \prod_{i=2}^{n-2} \left(q(z_{i-1},dz_i)m(z_i)\right)\int_D q(z_{n_s-2},dz_{n_s-1})\left(b_2+m_2(z_{n_s-1})\right),
\end{align*}
Then by defining $b_2:=\mathbb{E}_{\mu}\left[mQm_1\right]$, which is again a constant value, and $m_2 := mQm_1-b_2$, we can represent $a_{n_1}$ as the following sum using the above and step like those in \eqref{eq:lem2_new_07}.
\begin{align}
a_{n_1} &= a_{n_1-1}b_1 + a_{n_1-2}b_2 + \int_{D^{n_1-2}} \mu(dz_1) m(z_1) \prod_{i=2}^{n-2} q(z_{i-1},dz_i)m(z_i)Qm_2(z_{n_1-2}).
\end{align}
Continuing in this way -- defining constant values $b_i:=\mathbb{E}_{\mu}\left[mQm_{i-1}\right]$ and $m_i:=mQm_{i-1}-b_i$ for $i=2,...,n_1$, then splitting the integral as in \eqref{eq:lem2_new_8} -- we represent $a_{n_1}$ recursively as $a_{n_1}=\sum_{i=1}^{n_1} b_ia_{n_1-i}$. 

Again, our goal is to provide an upper bound for $a_{n_1}$ which we can establish through the recursive relationship $a_{n_1}=\sum_{i=1}^{n_1} b_ia_{n_1-i}$ if we can upper bound $b_1,...,b_{n_1}$.  First consider $b_1$. Let $Z \sim \mu$. % and $X, \tilde{X} \in \mathbb{R}^d \sim \pi$ independent.
\begin{equation*}
b_1 = \mathbb{E}\left[\exp\{rg(Z)\}\right] = \mathbb{E}\left[\lim_{n\rightarrow\infty}\sum_{k=0}^{n} \frac{r^k}{k!}(g(Z))^k\right].
\end{equation*}
%Define the partial sum as $s_n:=\sum_{k=0}^n \frac{r^k}{k!}(g(Z))^k$. Moreover, notice that 
Consider the partial sum $\sum_{k=0}^n \frac{r^k}{k!}(g(Z))^k$. Moreover, notice that
\begin{equation*}
\sup_{z\in D} |g(z)| = \sup_{x\in S^d}\sup_{\tilde{x}\in S^d} |f(x)-f(\tilde{x})| \overset{(a)}{\leq} \sup_{x\in S^d} \sup_{\tilde{x}\in S^d}L(1+\|x\|+\|\tilde{x}\|)\|x-\tilde{x}\| \overset{(b)}{\leq} L(1+2\sqrt{d}M)(2\sqrt{d}M),
\end{equation*}
where step $(a)$ holds since $f(\cdot)$ is pseudo-Lipschitz with constant $L$ and step $(b)$ due to $\|x-\tilde{x}\|\leq \|x\|+\|\tilde{x}\|$ and the boundedness of $S^d$: $\|x\|\leq M\sqrt{d}$ for some constant $M>0$ and all $x \in S^d$. Let $M_g=L(1+2\sqrt{d}M)(2\sqrt{d}M)$.  Then for each $n$, %$\abs{s_n}$ is bounded by
\begin{equation*}
\sum_{k=0}^n \frac{r^k}{k!}(g(Z))^k \leq \sup_{z\in D} \sum_{k=0}^n \frac{r^k}{k!}|g(z)|^k \leq\sum_{k=0}^n \frac{r^k}{k!} M_g^k \leq \sum_{k=0}^\infty \frac{r^k}{k!}M_g^k = \exp\{rM_g\}.
\end{equation*}
Since the constant $\exp\{rM_g\}$ is integrable with respect to any proper probability measure, we have 
\be
\begin{split}
&b_1 = \mathbb{E}\left[\lim_{n\rightarrow\infty}\sum_{k=0}^{n} \frac{r^k}{k!}(g(Z))^k\right] \overset{(a)}{=} \lim_{n\rightarrow\infty}\sum_{k=0}^{n} \frac{r^k}{k!}\mathbb{E}[(g(Z))^k] \overset{(b)}{\leq} 1 + \mathbb{E}[ (g(Z))^2 ]\sum_{k=2}^\infty \frac{r^kM_g^{k-2}}{k!}\\
&= 1 + \frac{r^2\mathbb{E}[ (g(Z))^2]}{2}\sum_{k=2}^\infty \frac{(rM_g)^{k-2}}{k!/2} \overset{(c)}{\leq} 1 + \frac{r^2\mathbb{E}[ (g(Z))^2]}{2}\sum_{k=2}^\infty \frac{(rM_g)^{k-2}}{(k-2)!}  = 1 + \frac{r^2\mathbb{E}[ (g(Z))^2]}{2}\exp\{rM_g\},%\qquad \text{for } 0< r < \frac{1}{M_g},
\label{eq:lem2_new_b1}
\end{split}
\ee
where step $(a)$ follows the dominated convergence theorem, step $(b)$ holds since $\mathbb{E}[g(Z)]=0$ and $\mathbb{E}[(g(Z))^k]\leq M_g^{k-2}\mathbb{E}[ (g(Z))^2]$, and step $(c)$ holds since $(k-2)! = k!/(k(k-1)) \leq k!/2 $ for $k \geq 2$ with the convention $0!=1$.

Next we'll bound $b_i$ for $i = 2, 3, \ldots$.  To do this we first establish an upper bound on $\norm{m_i}_{2, \mu}$ with the norm defined in \eqref{eq:lem2_new_betaQ}.
\begin{align*}
&\|m_i\|_{2,\mu}=\|mQm_{i-1}-b_i\|_{2,\mu}=\sqrt{\|mQm_{i-1}\|_{2,\mu}^2 -b_i^2} \leq \|mQm_{i-1}\|_{2,\mu} \\
&\overset{(a)}{\leq} \exp\{rM_g\}\|Qm_{i-1}\|_{2,\mu} \overset{(b)}{\leq} \exp\{rM_g\}\beta_Q\|m_{i-1}\|_{2,\mu}.
\end{align*}
Step $(a)$ holds since $\sup_{z\in D} m(z)=\sup_{z\in D} \exp\{rg(z)\}\leq \exp\{rM_g\}$.
Step $(b)$ holds since $E_{\mu}m_i=0$, for all $i=1,...,n$ by construction, and so $\|Qm_i\|_{2,\mu}\leq \beta_Q\|m_i\|_{2,\mu}$ by \eqref{eq:lem2_new_betaQ}.
Hence, extending the above result recursively, we find
\be
\|m_i\|_{2,\mu}\leq (\exp\{rM_g\}\beta_Q)^{i-1}\|m_1\|_{2,\mu}.
\label{eq:lem2_new_9}
\ee
Let $\langle f_1, f_2 \rangle_{\mu} = \int f_1(z) f_2(z) \mu(dz)$.  We use this to bound $b_i$ in the following by noting that $b_i = \mathbb{E}_{\mu}[mQm_{i-1}] = \langle m, Q m_{i-1} \rangle_{\mu}=\langle m_1+b_1, Q m_{i-1} \rangle_{\mu}=\langle m_1, Q m_{i-1} \rangle_{\mu}$, where the last equality holds because
\begin{align*}
\langle b_1, Q m_{i-1} \rangle_{\mu} &= b_1\int_{z\in D} Qm_{i-1}(z) \mu(dz) = b_1 \int_{z\in D} \int_{z'\in D} m_{i-1}(z')q(z,dz')\mu(dz)\\
&\overset{(a)}{=} b_1 \int_{z'\in D} m_{i-1}(z') \int_{z\in D} q(z,dz')\mu(dz) \overset{(b)}{=} b_1 \int_{z\in D} m_{i-1}(z')\mu(dz')\overset{(c)}{=}0.
\end{align*}
In the above, step $(a)$ follows from Fubini's Theorem, step $(b)$ follows from the fact that $\mu$ is the stationary distribution of $q(z,dz')$, and step $(c)$ follows from the construction of $m_i$'s, which says that $\mathbb{E}_{\mu}m_i=0,$ for $i = 2, 3, \ldots$. Then,
\begin{equation}
\begin{split}
b_i =  \langle m_1,Qm_{i-1} \rangle_\mu \overset{(c)}{\leq} \|m_1\|_{2,\mu}\|Qm_{i-1}\|_{2,\mu}
 &\overset{(d)}{\leq} \beta_Q(\beta_Q e^{rM_g})^{i-2}\|m_1\|_{2,\mu}^2, %\overset{(e)}{\leq} \beta_Q(\beta_Q e^{rM_g})^{i-2}\|m_1\|_{2,\mu}^2,
\label{eq:lem2_new_9a}
\end{split}
\end{equation}
where step $(c)$ follows Cauchy-Schwarz inequality and step $(d)$ follows from the fact that $\|Qm_{i-1}\|_{2,\mu}\leq \beta_Q\|m_{i-1}\|_{2,\mu}$ by \eqref{eq:lem2_new_betaQ} and \eqref{eq:lem2_new_9}.  %, and step $(e)$ from the fact that $\|m_1\|_{2,\mu}^2 = \|m\|_{2,\mu}^2 - b_1^2$ and therefore $\|m\|_{2,\mu}^2 \leq \|m_1\|_{2,\mu}^2$. \textcolor{cyan}{Something is wrong with step $(e).$  If $\|m_1\|_{2,\mu}^2 = \|m\|_{2,\mu}^2 - b_1^2$ then $\|m_1\|_{2,\mu}^2 + b_1^2 = \|m\|_{2,\mu}^2$ and so $\|m_1\|_{2,\mu}^2 \leq \|m\|_{2,\mu}^2$.  I think this is ok, because $\|m\|_{2,\mu}^2 \leq 1 + 2r^2 \mathbb{E}[ (g(Z))^2]\exp\{2rM_g\}$ (or something like this by an argument similar to \eqref{eq:lem2_new_b1}) and so replacing $m_1$ by $m$ after step $(e)$ should work.} 
Now let $Z \sim \mu$ and we bound $\|m_1\|_{2,\mu}^2$ as follows
\begin{align}
\|m_1\|_{2,\mu}^2 = \mathbb{E}[e^{2rg(Z)}]-(\mathbb{E}[e^{rg(Z)}])^2
&\overset{(f)}{\leq} 1 + 2r^2\mathbb{E}[(g(Z))^2]e^{2rM_g}- e^{2r\mathbb{E}[g(Z)]} \overset{(g)}{=} 2r^2\mathbb{E}[(g(Z))^2]e^{2rM_g}, \label{eq:lem2_new_11}
\end{align}
where step $(f)$ uses similar approach to that used to bound $b_1$ in \eqref{eq:lem2_new_b1} and Jensen's inequality, and step $(g)$ follows since $\mathbb{E}[g(Z)] = 0$.

Therefore, from \eqref{eq:lem2_new_b1}, \eqref{eq:lem2_new_9a}, and \eqref{eq:lem2_new_11} we have
\begin{align}
b_1&\leq 1 + \frac{r^2\mathbb{E}[(g(Z))^2]}{2}\exp\{rM_g\} \quad \text{ and } \quad b_i\leq \beta_Q(\beta_Q\exp\{rM_g\})^{i-2} 2r^2\mathbb{E}[(g(Z))^2 ]\exp\{2rM_g\}.
\label{eq:lem2_new_12}
\end{align}
%\textcolor{cyan}{Per comment above, I think it's $b_i\leq \beta_Q(\beta_Q\exp\{rM_g\})^{i-2} (1 + 2r^2\mathbb{E}[(g(Z))^2 ]\exp\{2rM_g\})$}.  
Let $X,\tilde{X}\sim \pi$ independent. Notice that
\begin{align*}
&\mathbb{E}[(g(Z))^2] = \mathbb{E}[
(f(X)-f(\tilde{X}))^2] \overset{(a)}{\leq} L^2 \mathbb{E}[((1+\|X\|+\|\tilde{X}\|)\|X-\tilde{X}\|)^2]\\
& \overset{(b)}{\leq} 5L^2\left(2\mathbb{E}[\|X\|^2] + 2\mathbb{E}[\|X\|^4] + 4\mathbb{E}[\|X\|^2]\mathbb{E}[\|\tilde{X}\|^2]\right)\\
&\overset{(c)}{\leq} 10L^2 \left( \sum_{i=1}^d \mathbb{E}[X_i^2] + d\sum_{i=1}^d \mathbb{E}[X_i^4] + 2\left(\sum_{i=1}^d \mathbb{E}[X_i^2]\right)\left(\sum_{i=1}^d \mathbb{E}[\tilde{X}_i^2]\right)\right)\overset{(d)}{=} 10L^2 \left( d \textsf{m}_2 + d^2 \textsf{m}_4 + 2 d^2 \textsf{m}_2^2 \right),
\end{align*}
where step $(a)$ holds since $f(\cdot)$ is pseudo-Lipschitz with constant $L>0$, step $(b)$ uses $\|X-\tilde{X}\|\leq \|X\| + \|\tilde{X}\|$, Lemma \ref{lem:squaredsums}, and  the fact that $X$ and $\tilde{X}$ are i.i.d., step $(c)$ uses Lemma \ref{lem:squaredsums}, and in step $(d)$, $\textsf{m}_2$ and $\textsf{m}_4$ denote the second and fourth moment of $\gamma$, respectively. Because $\gamma$ is defined on a bounded state space, $\textsf{m}_2$ and $\textsf{m}_4$ are finite.

Let $\mathsf{b}^2=10L^2 \left( d \textsf{m}_2 + d^2 \textsf{m}_4 + 2 d^2 \textsf{m}_2^2 \right)$, $\mathsf{a}=\frac{1}{2}\mathsf{b}^2 \exp\{rM_g\}$, and $\alpha=\beta_Q \exp\{rM_g\}$. Choose $r< (1-\beta_Q)/M_g$, then we have $0 < \alpha<1$ since $1-\beta_Q < -\ln \beta_Q$.
Using these bounds and notation, \eqref{eq:lem2_new_12} becomes
\begin{align}
b_1&\leq 1+\mathsf{a} r^2 \quad \text{ and } \quad b_i\leq \alpha^{i-1}4\mathsf{a}r^2.
\label{eq:lem2_new_13}
\end{align}
%\textcolor{cyan}{Per comment above, I think we would have $b_i \leq \alpha^{i-1}(\exp\{-r M_g\} + 4\mathsf{a}r^2) \leq \alpha^{i-1}(\exp\{-(1 - \beta_Q)\} + 4\mathsf{a}r^2)  \leq \alpha^{i-1}(1 + 4\mathsf{a}r^2)$.}

We now bound $a_1,...,a_{n_1}$ by induction.  We will show $a_i\leq [\phi(r)]^i$, where $\phi(r)=1+Cr^2$ for some $C\geq 4\textsf{a}$ that is independent of $i$.  For $i=1$, $a_1 = b_1\leq 1+4 \textsf{a} r^2.$ Hence, the hypothesis $a_i\leq  [\phi(r)]^i$ is true for $i=1$. %\textcolor{cyan}{For the updated result we want $\phi(r)=C(1+ 4\textsf{a}r^2)$ for some constant $C>1$.} 
Suppose that the hypothesis is true for $i \leq n_1-1$, then
\begin{align}
&a_{n_1} = b_1a_{n_1-1} + \sum_{i=2}^{n_1} b_ia_{n_1-i} \leq (1+ 4\textsf{a} r^2) [\phi(r)]^{n_1-1} + \sum_{i=2}^{n_1} 4 \textsf{a} r^2 \alpha^{i-1} [\phi(r)]^{n_1-i},
\label{eq:lem2_new_14}
\end{align}
where the final inequality in the above follows by \eqref{eq:lem2_new_13} and the inductive hypothesis.  Consider only the second term on the right side of \eqref{eq:lem2_new_14},
\begin{align*}
&\sum_{i=2}^{n_1} 4 \textsf{a} r^2 \alpha^{i-1} [\phi(r)]^{n_1-i} =  4 \textsf{a} r^2\alpha^{n_1-1}\sum_{i=2}^{n_1}  [\alpha^{-1} \phi(r)]^{n_1-i} \\
& \qquad = 4 \textsf{a} r^2\alpha^{n_1-1} \left(\frac{1- \left(\phi(r)\alpha^{-1} \right)^{n_1-1}}{1- \phi(r)\alpha^{-1}} \right) = 4 \textsf{a} r^2 \left(\frac{ \alpha [\phi(r)]^{n_1-1}-\alpha^{n_1}}{\phi(r)-\alpha} \right) \leq \frac{4\textsf{a}r^2 \alpha [\phi(r)]^{n_1-1}}{\phi(r)-\alpha},
\end{align*}
where the final inequality follows since $\textsf{a}, \alpha > 0$. %\textcolor{cyan}{I think we would have $\sum_{i=2}^{n_1} (1+4 \textsf{a} r^2) \alpha^{i-1} [\phi(r)]^{n_1-i} \leq  (1+4\textsf{a}r^2) \frac{ \alpha [\phi(r)]^{n_1-1}}{\phi(r)-\alpha}$.}  
Then plugging the above result into \eqref{eq:lem2_new_14}, we find
\begin{align*}
a_{n_1} \leq (1+4\textsf{a} r^2) [\phi(r)]^{n_1-1} + \frac{4\textsf{a}r^2 \alpha [\phi(r)]^{n_1-1}}{\phi(r)-\alpha} \leq [\phi(r)]^{n_1-1}\left(1+ \frac{4 \textsf{a} r^2 \phi(r)}{\phi(r)-\alpha}\right) &\leq [\phi(r)]^{n_1-1}\left(1+ \frac{4 \textsf{a} r^2}{1-\alpha}\right),
\end{align*}
where the final inequality follows since $\phi(r) \geq 1$.   %\textcolor{cyan}{Updated we get: 
%\[a_{n_1} \leq  (1+4\textsf{a} r^2) [\phi(r)]^{n_1-1} \left[1 + \frac{ \alpha}{\phi(r)-\alpha}\right] \leq  (1+4\textsf{a} r^2) [\phi(r)]^{n_1-1} \left[1 + \frac{ 1}{1-\alpha}\right] .\]} 
Therefore, let $C=4\mathsf{a}(1-\alpha)^{-1} > 4 \mathsf{a}$, since $0< \alpha < 1$, and so $\phi(r)=1+4 \textsf{a}r^2(1 - \alpha)^{-1}$. It follows from the above then,
\begin{equation}
a_{n_1}\leq \left(1+ \frac{4 \textsf{a} r^2}{1-\alpha}\right)^{n_1}=e^{n_1 \ln\left(1+4\textsf{a}r^2(1-\alpha)^{-1}\right)}\leq e^{n_1 4\textsf{a}r^2(1-\alpha)^{-1}},
\label{eq:lem2_new_15}
\end{equation}
where the final inequality uses the fact that $\ln(1 + x) \leq x$ for $x \geq 0$.  %\textcolor{cyan}{Updated, we let $C = 1 + \frac{1}{1-\alpha} > 1$ since $0 < \alpha < 1$.  Then $a_{n_1} \leq (1+4\textsf{a} r^2)^{n_1}\left[1 + \frac{ 1}{1-\alpha}\right]^{n_1} = \exp\left\{n_1 \log(1+4\textsf{a} r^2)+n_1 \log\left[1 + \frac{ 1}{1-\alpha}\right]\right\} \leq \exp\left\{n_1 4\textsf{a} r^2+n_1(1-\alpha)^{-1}\right\}.$}

Finally, from \eqref{eq:lem2_new_4}, \eqref{eq:lem2_new_4aa}, and the bound in \eqref{eq:lem2_new_15},
\begin{align*}
P\left(\frac{1}{n_1}\sum_{i=1}^{n_1} f(X^{(1)}_i)\geq \e\right) \leq \exp\left(-n_1\left(r\e-4\textsf{a}r^2(1-\alpha)^{-1}\right)\right)\overset{(a)}{=}\exp\left( -n_1 \left( r \e - \frac{ 2\textsf{b}^2  r^2 e^{r M_g}}{1-\beta_Q e^{rM_g}}\right)\right),
\end{align*}
where step $(a)$ follows from the fact that $\mathsf{a}=\mathsf{b}^2e^{rM_g}/2$ and $\alpha = \beta_Q e^{rM_g}$.  Now let us consider the term in the exponent in the above for the cases where (\emph{i}) $\mathsf{b}^2\geq M_g$ and (\emph{ii}) $\mathsf{b}^2<M_g$ separately, and then combine the results in the two cases to obtain a desired bound for all $\e \in (0,1)$.
%\textcolor{cyan}{The new upper bound is therefore,
%\[P\left(\frac{1}{n_1}\sum_{i=1}^{n_1} f(X^{(1)}_i)\geq \e\right) \leq \exp\left(-n_1\left(r\e- 4\textsf{a} r^2 - (1-\alpha)^{-1}\right)\right)\overset{(a)}{=}\exp\left( -n_1 \left( r \e - 2\mathsf{b}^2e^{rM_g} r^2 - (1-\beta_Q e^{rM_g})^{-1} \right)\right),\]}

First (\emph{i}) $\mathsf{b}^2\geq M_g$.
Notice for every $0< \epsilon<4 \mathsf{b}^2/M_g$, if we let $r=(1-\beta_Q)\e/(4\textsf{b}^2)$, then $r<(1-\beta_Q)/M_g$ as required before. We show whenever $0< \epsilon\leq \mathsf{b}^2/M_g$, we can obtain a desired bound.  Then the condition in the lemma statement, $\e \in (0,1)$, falls within this effective region.
\begin{align*}
&r \e - \frac{ 2\textsf{b}^2  r^2 e^{r M_g}}{1-\beta_Q e^{rM_g}} \overset{(a)}{=} \frac{(1-\beta_Q)\e^2}{4  \textsf{b}^2}-\frac{(1-\beta_Q)^2\e^2}{8 \textsf{b}^2} \cdot \frac{\exp\left(\frac{(1-\beta_Q) M_g \e}{4 \textsf{b}^2}\right)}{1-\beta_Q\exp\left(\frac{(1-\beta_Q) M_g \e}{4 \textsf{b}^2}\right)}\\
&= \frac{(1-\beta_Q)\e^2}{8 \textsf{b}^2}\left(1- \frac{\exp\left(\frac{(1-\beta_Q) M_g \e}{4 \textsf{b}^2}\right)-1}{1-\beta_Q\exp\left(\frac{(1-\beta_Q) M_g \e}{4 \textsf{b}^2}\right)}\right) \overset{(b)}{\geq} \frac{(1-\beta_Q)\e^2}{8 \textsf{b}^2}\left(1- \frac{\frac{(1-\beta_Q) M_g \e}{3 \textsf{b}^2}}{1-\beta_Q\left(1+\frac{(1-\beta_Q) M_g \e}{3 \textsf{b}^2}\right)}\right)\\
&= \frac{(1-\beta_Q)\e^2}{8 \textsf{b}^2}\left(1- \frac{ M_g \e}{ 2\textsf{b}^2 + (\textsf{b}^2-\beta_Q M_g \e)}\right) \overset{(c)}{\geq} \frac{(1-\beta_Q)\e^2}{8 \mathsf{b}^2}\left( 1- \frac{\e}{ 2}\right),\qquad\text{for }0 < \e\leq \mathsf{b}^2/M_g.
\end{align*}
In the above, step $(a)$ by plugging in $r=(1-\beta_Q)\e/(4\mathsf{b}^2)$, step $(b)$ holds since $e^x\leq 1+4x/3$ for $x\leq 1/2$, and step $(c)$ holds since $\e\leq \mathsf{b}^2/M_g$, so $(\textsf{b}^2-\beta_Q M_g \e)>0$, and the fact $\mathsf{b}^2\geq M_g$.

Next consider (\emph{ii}) $\mathsf{b}^2 < M_g$.
In this case, set $r=(1-\beta_Q)\e/(4 M_g)$. Hence, $r<(1-\beta_Q)/M_g$ for $\e \in (0,1)$, and then
\begin{align*}
&r \e - \frac{ 2\textsf{b}^2  r^2 e^{r M_g}}{1-\beta_Q e^{rM_g}} \overset{(a)}{>} r \e - \frac{ 2M_g  r^2 e^{r M_g}}{1-\beta_Q e^{rM_g}}
\overset{(b)}{=} \frac{(1-\beta_Q)\e^2}{4  M_g}-\frac{(1-\beta_Q)^2\e^2}{8 M_g} \cdot \frac{\exp\left(\frac{(1-\beta_Q) \e}{4}\right)}{1-\beta_Q\exp\left(\frac{(1-\beta_Q) \e}{4}\right)}\\
&= \frac{(1-\beta_Q)\e^2}{8 M_g}\left(1- \frac{\exp\left(\frac{(1-\beta_Q) \e}{4 }\right)-1}{1-\beta_Q\exp\left(\frac{(1-\beta_Q) \e}{4 }\right)}\right) \overset{(c)}{\geq} \frac{(1-\beta_Q)\e^2}{8 M_g}\left( 1- \frac{\e}{2}\right),\qquad\text{for }0 < \e \leq 1.
\end{align*}
In the above, step $(a)$ holds since $\mathsf{b}^2<M_g$, step $(b)$ by plugging in $r=(1-\beta_Q)\e/(4 M_g)$, and step $(c)$ follows similar calculation as in case (\emph{i}).

Combining the results in the two cases, we conclude that for all $\e\in (0,1)$, the following is satisfied:
\begin{equation*}
r \e - \frac{ 2\textsf{b}^2  r^2 e^{r M_g}}{1-\beta_Q e^{rM_g}} \geq \frac{(1-\beta_Q)\e^2}{8\max(M_g,\mathsf{b}^2)}\left(1- \frac{\e}{2}\right).
\end{equation*}
Hence, for $\e \in (0,1)$,
\begin{equation}
P\left(\frac{1}{n_1}\sum_{i=1}^{n_1} f(X^{(1)}_i)\geq \e\right)\leq\exp\left(\frac{-(1-\beta_Q)n_1\e^2}{8\max(M_g,\mathsf{b}^2)}\left(1-\frac{\e}{2}\right)\right) \leq\exp\left(\frac{-(1-\beta_Q)n_1\e^2}{16\max(M_g,\mathsf{b}^2)}\right).
\label{eq:lem2_new_16}
\end{equation}

Therefore, using \eqref{eq:lem2_new_split} and the fact that we can show a similar result for each $s = 2, 3, \ldots, d$, we have for $\e \in (0,1)$,
\begin{align}
&P\left(\abs{\frac{1}{n}\sum_{i=1}^n f(X_i) - \mathbb{E}_\pi f} \geq \e \right) \leq \sum_{s=1}^{d} P\left( \left\vert \frac{1}{n_s}\sum_{i=1}^{n_s} f(X^{(s)}_i) - \mathbb{E}_\pi f \right\vert \geq \frac{n\epsilon}{dn_s} \right) \nonumber \\
&\qquad \overset{(a)}{\leq} \sum_{s=1}^d \exp\left(\frac{-(1-\beta_Q) n^2 \e^2}{16 n_s \max(M_g,\mathsf{b}^2)}\right) \overset{(b)}{\leq} d \exp\left(\frac{-(1-\beta_Q) n \e^2}{16d\max(M_g,\mathsf{b}^2)}\right),
\label{eq:lem2_new_17}
\end{align}
where step $(a)$ follows \eqref{eq:lem2_new_16} and step $(b)$ holds since $n/n_s \geq n/n_1 = n/(\lfloor n/d\rfloor-1) \geq d,$ for all $s \in [d]$.  
To complete the proof, we recall that $\mathsf{b}^2=10L^2 \left( d \textsf{m}_2 + d^2 \textsf{m}_4 + 2 d^2 \textsf{m}_2^2 \right)$ and $M_g=L(1+2\sqrt{d}M)(2\sqrt{d}M)$.

\end{proof}

\bibliographystyle{ieeetr}
{\small{
\bibliography{NSdenoisers}
}}

\end{document}